\documentclass{CSML}

\def\dOi{12(3:12)2016}
\lmcsheading%
{\dOi}
{1--62}
{}
{}
{Oct.~30, 2015}
{Sep.~22, 2016}
{}

\ACMCCS{[{\bf Theory of computation}]: Computational complexity and
  cryptography---Complexity theory and logic; Logic}
 \subjclass{[{\bf Theory of computation}]: Computation by abstract devices ---Complexity Measures and Classes; [{\bf Theory of computation}]: Mathematical logic and formal languages---Mathematical logic}

\newtheorem{cLa}{Claim}

\newcommand{\amp}{_{\it amplitude}}
\newcommand{\spa}{_{\it space}}
\newcommand{\tim}{_{\it time}}

\newcommand{\cltw}{\mbox{\bf CL12}}

\newcommand{\areleven}{\mbox{\bf CLA11}}

\newcommand{\thr}{\areleven_{\mathcal A}^{\mathcal R}}

 \newcommand{\blank}{\mbox{{\scriptsize {\sc Blank}}}}
 \newcommand{\transition}{{\sc Update Sketch}}
 \newcommand{\call}{{\sc Fetch Symbol}}
 \newcommand{\main}{{\sc Make History}}
 \newcommand{\simm}{\mbox{\sc Sim}}
 \newcommand{\mainn}{\mbox{\sc Main}}
 \newcommand{\loopp}{\mbox{\sc Loop}}
 \newcommand{\routine}{\mbox{\sc Routine}}
 \newcommand{\repeatt}{\mbox{\sc Repeat}}
\newcommand{\restart}{\mbox{\sc Restart}}

 \newcommand{\bit}{\mbox{\it Bit}}
 \newcommand{\even}{\mbox{\scriptsize {\it even}}}
 \newcommand{\odd}{\mbox{\scriptsize {\it odd}}}

\newcommand{\intimpl}{\mbox{\hspace{2pt}$\circ$\hspace{-0.14cm} \raisebox{-0.043cm}{\Large --}\hspace{2pt}}}

\newcommand{\successor}{\mbox{\hspace{1pt}\boldmath $'$}}

\newcommand{\oo}{\bot}            
\newcommand{\pp}{\top}            
\newcommand{\xx}{\wp}

\newcommand{\seq}[1]{\langle #1 \rangle}

\newcommand{\gneg}{\mbox{\small $\neg$}}                  %game negation
\newcommand{\mli}{\hspace{2pt}\mbox{\small $\rightarrow$}\hspace{2pt}}                      %strong reduction
\newcommand{\cla}{\mbox{$\forall$}}      %blind universal quantifier
\newcommand{\cle}{\mbox{$\exists$}}        %blind existential quantifier
\newcommand{\mld}{\hspace{2pt}\mbox{\small $\vee$}\hspace{2pt}}     %multiplicative disjunction
\newcommand{\mlc}{\hspace{2pt}\mbox{\small $\wedge$}\hspace{2pt}}   %multiplicative conjunction
\newcommand{\ade}{\mbox{\large $\sqcup$}}      %additive existential quantifier
\newcommand{\ada}{\mbox{\large $\sqcap$}}      %additive universal quantifier
\newcommand{\add}{\hspace{2pt}\mbox{\small $\sqcup$}\hspace{2pt}}                     %additive disjunction
 
\newcommand{\tlg}{\bot}               %classical \bot; trivially lost elementary game
\newcommand{\twg}{\top}               %classical \top; trivially won elementary game

\theoremstyle{italic}
\theoremstyle{italic}
\theoremstyle{plain}
\theoremstyle{plain}
\theoremstyle{plain}
\theoremstyle{plain}
\newenvironment{idea}{{\em Proof idea.} }{\  \rule{1.5mm}{1.5mm} \vspace{.15in} }
\usepackage{amsfonts} 
\usepackage{hyperref}
\theoremstyle{plain} %\crefname{satz}{Satz}{S\"atze}
\begin{document}

\title[Build your own clarithmetic II]{Build your own clarithmetic II: Soundness}

\author[G.~Japaridze]{Giorgi Japaridze}
\address{Department of Computing Sciences, Villanova University, 800 Lancaster Avenue, Villanova, PA 19085, USA}	%required
\urladdr{http://www.csc.villanova.edu/$\sim$japaridz/}  
\email{giorgi.japaridze@villanova.edu}  
%\thanks{thanks 1, optional.}	%optional


\begin{abstract}
  \noindent {\em Clarithmetics} are number theories based on {\em computability logic}. Formulas of these 
theories represent interactive computational problems, and their ``truth'' is understood as existence of an algorithmic
solution. Various complexity constraints on such solutions induce various versions of clarithmetic. The
present paper introduces a parameterized/schematic version $\mbox{\bf CLA11}_{P_4}^{P_1,P_2,P_3}$. By  tuning the three parameters $P_1,P_2,P_3$ in an essentially mechanical manner, one automatically obtains sound and complete theories with respect to a wide range of target {\em tricomplexity}
classes, i.e., combinations of time (set by $P_3$), space (set by $P_2$) and so called amplitude (set by $P_1$) complexities. Sound in the sense that
every theorem $T$ of the system represents an interactive number-theoretic computational problem with
a solution from the given tricomplexity class and, furthermore, such a solution can be automatically extracted from a
proof of $T$. And complete in the sense that every interactive number-theoretic problem with a solution
from the given tricomplexity class is represented by some theorem of the system. Furthermore, through tuning the 4th parameter $P_4$, at the cost of sacrificing recursive axiomatizability but not simplicity or elegance, 
the above {\em extensional
completeness} can be strengthened to {\em intensional completeness}, according to which every formula 
representing a problem with a solution from the given tricomplexity class is a theorem of the system. This article
is published in two parts. The previous Part I has introduced the system and proved its completeness, while the
present Part II is devoted to proving soundness. 
\end{abstract}

\maketitle

\tableofcontents

\section{Getting started} 
Being a continuation of \cite{AAAI}, this article fully relies on the terminology and notation introduced in its predecessor, with which --- or, at least, with the first two sections of which --- the reader is assumed to be already familiar. 
Just like \cite{AAAI}, this article further relies on \cite{cl12}, and  familiarity with that self-contained, tutorial-style paper (with proofs omitted) is another prerequisite for reading the present piece of writing.

The  sole purpose if the present article is to prove the soundness of the system $\areleven$ introduced in \cite{AAAI}. Specifically, the goal is to prove clause 3 of Theorem 2.6 of \cite{AAAI}, which, slightly paraphrased,  reads:

\begin{quote}{\em If a theory $\thr$ is  regular, then  there is an effective procedure that takes an arbitrary extended $\areleven^{\mathcal R}_{{\mathcal A}!}$-proof of an arbitrary sentence $X$ and constructs an    
${\mathcal R}$ tricomplexity  solution for $X$. 
}\end{quote}

\noindent Our  soundness proof is written so that it can be read independently of the completeness proof given in 
\cite{AAAI}. 

 Let us get started right now. 
Assuming that a theory $\thr$ is regular, the above-displayed statement  can be verified  by induction on the number of steps in an extended $\areleven^{{\mathcal R}}_{A!}$-proof   of $X$. The basis of this induction is a rather straightforward observation that all axioms have ${\mathcal R}$ tricomplexity  solutions. Namely,  in the case of Peano axioms  
such a ``solution'' is simply a machine that does nothing. All axioms from $\mathcal A$ have ${\mathcal R}$ tricomplexity  solutions by condition 1 of Definition 2.5 of \cite{AAAI}; furthermore, according to the same condition, such solutions can be effectively obtained even when $\mathcal A$ is infinite. Finally, the Successor, Log and Bit axioms  can be easily seen to have linear  amplitude, logarithmic space and polynomial (in fact, linear) time solutions and, in view of conditions 2 and 3  of Definition 2.2 of \cite{AAAI}, such solutions are automatically also ${\mathcal R}$ tricomplexity solutions.  
As for the inductive step, it will be taken care of  by the later-proven Theorems \ref{fe}, \ref{feb9ee} and \ref{feb9eei},  according to which   the rules of   Logical Consequence, $\mathcal R$-Comprehension and $\mathcal R$-Induction preserve --- in a constructive sense --- the property of having an ${\mathcal R}$ tricomplexity solution.

\section{Soundness of Logical Consequence}\label{sfe}
%\marginpar{sfe}

As we remember from \cite{cl12}, CoL understands algorithmic  strategies as interactive Turing machines called HPMs (Hard-Play Machines).\label{xHPM} 

\begin{thm}\label{fe} 
%\marginpar{fe}
Consider any regular boundclass triple $\mathcal R$. 
There is an ($\mathcal R$-independent) effective procedure\footnote{Here and later in the similar Theorems \ref{feb9ee} and \ref{feb9eei},  as one can easily guess, $\mathcal R$-independence of a procedure  means that the procedure is the same regardless of what particular value $\mathcal R$ assumes.}
that takes an arbitrary $\cltw$-proof \ $\mathbb{P}$ of an arbitrary $\mathbb{L}$-sequent {\em $E_1,\ldots,E_n$ $ \intimpl F$}, 
arbitrary HPMs ${\mathcal N}_1,\ldots,{\mathcal N}_n$ and constructs an HPM $\mathcal M$ such that,  if ${\mathcal N}_1,\ldots,{\mathcal N}_n$ are $\mathcal R$ tricomplexity solutions of $E_1,\ldots,E_n$, respectively, then $\mathcal M$ is an $\mathcal R$  tricomplexity solution of $F$. 
\end{thm}

\begin{proof} Such an effective procedure is nothing but the one whose existence is stated in Theorem 11.1 of \cite{cl12}. Consider
an arbitrary $\cltw$-proof $\mathbb{P}$ of an arbitrary $\mathbb{L}$-sequent $E_1,\ldots,E_n \intimpl F$, and  
arbitrary HPMs ${\mathcal N}_1,\ldots,{\mathcal N}_n$. Let $\mathcal M$\label{xmm1} be the HPM constructed for/from these parameters by the above procedure. 

Assume  ${\mathcal R}$ is a regular boundclass triple, and   ${\mathcal N}_1,\ldots,{\mathcal N}_n$ are ${\mathcal R}$ tricomplexity solutions of $ E_1,\ldots,E_n$, respectively.   All three components of  ${\mathcal R}$ are linearly closed by condition 3 of Definition 2.2 of \cite{AAAI} and, being boundclasses, they are also closed under syntactic variation. This means that, for some common triple 
$\bigl(\mathfrak{a}(x),\mathfrak{s}(x),\mathfrak{t}(x)\bigr)\in{\mathcal R}\amp\times {\mathcal R}\spa\times 
{\mathcal R}\tim$ of unary bounds, all $n$ machines run in tricomplexity $(\mathfrak{a},\mathfrak{s},\mathfrak{t})$. That is, we have:        

%\begin{quote} \ \vspace{-8pt}
\begin{enumerate}[label={\bf(\roman*)}]
%{\em {\bf (i)} 
\item For each $i\in\{1,\ldots,n\}$, ${\mathcal N}_i$ is an
  $\mathfrak{a}$ amplitude, $\mathfrak{s}$ space and $\mathfrak{t}$
  time solution of $E_i$. %\vspace{-8pt} \
\end{enumerate}%\end{quote}
In view of conditions 2 and  5 of Definition 2.2 of \cite{AAAI}, we may further assume that:
%\begin{quote} \ \vspace{-8pt}
\begin{enumerate}[label={\bf(\roman*)},start=2]
%{\em {\bf (ii)} 
\item For any $x$, $\mathfrak{a}(x)\geq x$.

%{\bf (iii)}
\item For any $x$, $\mathfrak{s}(x)\geq \log (x)$.

%{\bf (iv)}
\item For any $x$, $\mathfrak{t}(x)\geq x$ and $\mathfrak{t}(x)\geq\mathfrak{s}(x)$.
%}\vspace{-8pt}\
\end{enumerate}%\end{quote}

\noindent Now, remembering that $E_i$ stands for $E_{i}^{\dagger}$, our condition (i) is the same as condition (i) of Theorem 11.1 of \cite{cl12} with $^\dagger$ in the role of $^*$. Next, taking into account that $0$ is the only constant that may appear in the $\mathbb{L}$-sequent $E_1,\ldots,E_n\intimpl F $ and hence the native magnitude of the latter is $0$, 
our condition (ii) is the same as condition (ii) of Theorem 11.1 of \cite{cl12}. Finally, our conditions (iii) and (iv) are the same as conditions (iii) and (iv) of Theorem 11.1 of \cite{cl12}.

 Then, according to that theorem, there are numbers 
$\mathfrak{b}$ and $\mathfrak{d}$ such that 
$\mathcal M$ is an $\mathfrak{a}^{\mathfrak{b}}(\ell)$ amplitude, 
$O\bigl( \mathfrak{s}(\mathfrak{a}^{\mathfrak{b}}(\ell))\bigr)$ space and 
$O\bigl(  (\mathfrak{t}\bigl(\mathfrak{a}^{\mathfrak{b}}(\ell) ) )^{\mathfrak{d}}\bigr)$ time solution of $F$. 
 But, by  condition 2 (if $\mathfrak{b}=0$) or 4 (if $\mathfrak{b}>0$) of Definition 2.2 of \cite{AAAI}, we have $\mathfrak{a}^{\mathfrak{b}}(\ell)\preceq {\mathcal R}\amp$, meaning that  $\mathcal M$ runs in amplitude ${\mathcal R}\amp$. The  fact $\mathfrak{a}^{\mathfrak{b}}(\ell)\preceq {\mathcal R}\amp$, again by    condition 4 of Definition 2.2 of \cite{cl12}, further implies that  $\mathfrak{s}\bigl(\mathfrak{a}^{\mathfrak{b}}(\ell)\bigr)\preceq {\mathcal R}\spa$ and $ \mathfrak{t} \bigl(\mathfrak{a}^{\mathfrak{b}}(\ell) \bigr) \preceq {\mathcal R}\tim$. The fact $ \mathfrak{t} \bigl(\mathfrak{a}^{\mathfrak{b}}(\ell) \bigr) \preceq {\mathcal R}\tim$,  in turn, by condition 3 of Definition 2.2 of \cite{cl12}, further implies that 
$\bigl(\mathfrak{t} (\mathfrak{a}^{\mathfrak{b}}(\ell) )\bigr)^{\mathfrak{d}}\preceq {\mathcal R}\tim$.  
Now, by Remark 2.4 of  \cite{cl12}, the facts  $\mathfrak{s}\bigl(\mathfrak{a}^{\mathfrak{b}}(\ell)\bigr)\preceq 
{\mathcal R}\spa$ and $\bigl(\mathfrak{t} (\mathfrak{a}^{\mathfrak{b}}(\ell) )\bigr)^{\mathfrak{d}}\preceq 
{\mathcal R}\tim$, together with the earlier observation that $\mathcal M$ runs in $O\bigl( \mathfrak{s}(\mathfrak{a}^{\mathfrak{b}}(\ell))\bigr)$ space and $O\bigl(  (\mathfrak{t}\bigl(\mathfrak{a}^{\mathfrak{b}}(\ell) ) )^{\mathfrak{d}}\bigr)$ time, imply that $\mathcal M$ runs in space ${\mathcal R}\spa$
and time ${\mathcal R}\tim$. 
To summarize, $\mathcal M$ runs in tricomplexity ${\mathcal R}$, as desired. 
\end{proof}

\section{Soundness of Comprehension}\label{sss4}
%\marginpar{sss4}

\begin{thm}\label{feb9ee} 
%\marginpar{feb9ee} 
Consider any regular boundclass triple $\mathcal R$. 
There is an ($\mathcal R$-independent) effective procedure that takes an arbitrary application\footnote{Here and elsewhere in similar contexts, an ``application'' means an ``instance'', i.e., a particular premise-conclusion pair. In the case of $\mathcal R$-Comprehension, it is fully determined by the comprehension formula and the comprehension bound.} of $\mathcal R$-Comprehension, an arbitrary HPM $\mathcal N$ and constructs an HPM $\mathcal M$ such that, if  $\mathcal N$ is an $\mathcal R$ tricomplexity solution of the premise, then $\mathcal M$ is an $\mathcal R$ tricomplexity solution of the conclusion. 
\end{thm}

The rest of this section is devoted to a proof of the above theorem. Consider a regular boundclass triple $\mathcal R$. Further consider an HPM $\mathcal N$, and an application  
\begin{equation}\label{r3}
\frac{p(y)\add\gneg p(y)}{\ade |x| \leq \mathfrak{b}|\vec{s}|\cla y<\mathfrak{b}|\vec{s}| \bigl(\bit(y,x)\leftrightarrow p(y)\bigr)}\end{equation}
 of $\mathcal R$-Comprehension.  
Let $\vec{v}=v_1,\ldots,v_n$ be a list of all free variables of $p(y)$ other than $y$, and let us correspondingly rewrite (\ref{r3}) as 
%\marginpar{r33}
\begin{equation}\label{r33}
\frac{p(y,\vec{v})\add\gneg p(y,\vec{v})}{\ade |x| \leq \mathfrak{b}|\vec{s}|\cla y<\mathfrak{b}|\vec{s}| \bigl(\bit(y,x)\leftrightarrow p(y,\vec{v})\bigr)}.\end{equation}
By condition 1  of Definition 2.2 of \cite{AAAI}, from the bound $\mathfrak{b}(\vec{s})$ we can effectively extract an $\mathcal R$ tricomplexity solution of $\ada \ade z(z=\mathfrak{b}|\vec{s}|)$. Fix such a solution/algorithm and call it {\sc Algo}. 

Assume $\mathcal N$ is an $(\mathfrak{a},\mathfrak{s},\mathfrak{t})\in{\mathcal R}\amp\times{\mathcal R}\spa\times
{\mathcal R}\tim  $ tricomplexity solution of the  premise of (\ref{r33}).   
We want to construct an $ {\mathcal R} $ tricomplexity solution $\mathcal M$\label{xmm2} for the conclusion of (\ref{r33}). 
It should be noted that, while our claim of $\mathcal M$'s being an $ {\mathcal R} $ tricomplexity solution of the conclusion of (\ref{r33}) relies on the assumption that we have just made regarding $\mathcal N$,  our  {\em construction} of $\mathcal M$ itself does not depend on that
 assumption. It should also be noted that we construct $\mathcal M$ as  a single-work-tape machine. 

   This is how $\mathcal M$ works. At the beginning, it puts the symbol $\#$ into its buffer. Then it waits till Environment specifies constants $\vec{a}$ and $\vec{b}$ for the free variables  $\vec{s}$ and $\vec{v}$ of the conclusion of (\ref{r33}). If Environment never does so, then $\mathcal M$ is an automatic winner. Otherwise, the game is brought down to $\ade |x| \leq \mathfrak{b}|\vec{a}|\cla y<\mathfrak{b}|\vec{a}| \bigl(\bit(y,x)\leftrightarrow p(y,\vec{b})\bigr)$. Now,  using {\sc Algo},   
$\mathcal M$ computes  and remembers the value $c$ of  $\mathfrak{b}|\vec{a}|$. Condition 5 of Definition 2.2 of \cite{AAAI}
guarantees that    $c$ can be remembered with ${\mathcal R}\spa$ space. Thus, recalling  that {\sc Algo} runs in $\mathcal R$ tricomplexity, the steps taken by $\mathcal M$ so far do not take us beyond $\mathcal R$ and hence, in view of Remark 2.4 of \cite{AAAI},   
 can be ignored in our asymptotic analysis when arguing that $\mathcal M$ runs in $\mathcal R$ tricomplexity.
After these initial steps, $\mathcal M$ starts  acting according to the following procedure:\vspace{10pt}

\noindent{\bf Procedure} \routine:

{\em Step 1}. If $c= 0$, enter a move state and retire.   Otherwise, if $c\geq 1$, simulate the play of the  premise of (\ref{r33}) by $\mathcal N$ in the scenario where, at the very beginning of the play, $\mathcal N$'s adversary chose the same constants $\vec{b}$ for the  variables $\vec{v}$ as Environment did in the real play of the conclusion   and, additionally, chose $j$ for $y$, where $j=c-1$.  If (when) the simulation shows that, at some point, $\mathcal N$ chose the $\add$-disjunct 
$\gneg p(j,\vec{b})$, decrement the value of $c$ by $1$ and repeat the present step. And if (when) the simulation shows that, at some point, $\mathcal N$   chose the $\add$-disjunct $p(j,\vec{b})$, decrement the value of $c$ by $1$, put the bit  $1$ into the buffer, and go to Step 2.

{\em Step 2}. If $c= 0$, enter a move state and retire. Otherwise, if $c\geq 1$, simulate the play of the premise of (\ref{r33}) by $\mathcal N$ in the scenario where, at the very beginning of the play, $\mathcal N$'s adversary chose the same constants $\vec{b}$ for the   variables $\vec{v}$ as Environment did in the real play of the conclusion   and, additionally, chose $j$ for $y$, where $j=c-1$.  If (when) the simulation shows that, at some point, $\mathcal N$ chose the $\add$-disjunct $\gneg p(j,\vec{b})$ (resp. $p(j,\vec{b})$), decrement the value of $c$ by $1$, put the bit $0$ (resp. $1$) into the buffer, and repeat the present step.\vspace{10pt}

It is not hard to see that, what $\mathcal M$ did while following the above routine was that it constructed, in its buffer, the 
 constant $d$  with   $|d| \leq \mathfrak{b}|\vec{a}|\mlc \cla y<\mathfrak{b}|\vec{a}| \bigl(\bit(y,d)\leftrightarrow p(y,\vec{b})\bigr)$, and then made $\#d$ as its only move in the play. 
 This means that $\mathcal M$ is a solution of the conclusion of (\ref{r33}),  as desired. And, of course, 
our construction of $\mathcal M$ is effective. It thus remains to see that $\mathcal M$ runs in $ {\mathcal R}$ tricomplexity. 
In what follows, we implicitly rely on Remark 2.4 of \cite{AAAI}, the monotonicity of bounds and the obvious fact that  the background of 
any cycle of the simulated  $\mathcal N$ does not exceed the background of (the cycles of) $\mathcal M$ throughout its work within  \routine. 
The latter  is the case because all moves that reside on $\mathcal N$'s imaginary run tape --- namely, the moves (containing) $\vec{b}$ --- also reside on $\mathcal M$'s run tape. 

Since $\#d$ is the only move that $\mathcal M$ makes, our earlier observation  $|d| \leq \mathfrak{b}|\vec{a}|$ immediately implies that $\mathcal M$ runs in amplitude $\mathfrak{b}\in{\mathcal R}\amp$, as desired.

Next, observe that the space that $\mathcal M$ consumes while performing \routine\ is just the space needed to remember the value of the variable $c$, plus the space needed to simulate $\mathcal N$. The value of $c$  never exceeds  $\mathfrak{b}|\vec{a}|$, remembering which, as we have already observed, does not take us beyond the target ${\mathcal R}\spa$. 
In order to simulate $\mathcal N$, on its work tape $\mathcal M$ does not need to keep track   of $\mathcal N$'s run tape, because information on that content is available on $\mathcal M$'s own run tape. So, $\mathcal M$ (essentially) only needs to keep track 
of $\mathcal N$'s work-tape contents. By our assumption, $\mathcal N$ runs in space $\mathfrak{s}$. Therefore,  keeping track of its 
work-tape contents takes $O(\mathfrak{s})$ space, which is again within ${\mathcal R}\spa$. To summarize, $\mathcal M$ runs in space ${\mathcal R}\spa$, as desired. 

Finally, taking into account that  $\mathcal N$ runs in time $\mathfrak{t}$ and space $\mathfrak{s}$, it is clear that the time needed for any given iteration of either step of \routine\ is $O(\mathfrak{t}\times \mathfrak{s})$. This is so because simulating each step of $\mathcal N$ takes $O(\mathfrak{s})$ time, and there are $O(\mathfrak{t}
)$ steps to simulate. Altogether, there are $O(\mathfrak{b})$ iterations of either Step 1 or Step 2 of \routine. So, $\mathcal M$ runs in time $O(\mathfrak{t}\times \mathfrak{s}\times \mathfrak{b})$. Then, in view of the fact  that both 
$\mathfrak{s}\in{\mathcal R}\spa\preceq {\mathcal R}\tim$ and $\mathfrak{b}\in{\mathcal R}\amp\preceq 
{\mathcal R}\tim$  (condition 5 of Definition 2.2 of \cite{AAAI}), we find that   $\mathcal M$ runs in time 
$O(\mathfrak{t}\times\mathfrak{t}_1\times\mathfrak{t}_2)$ for some $\mathfrak{t}_1,\mathfrak{t}_2\in {\mathcal R}\tim$. But ${\mathcal R}\tim$ is polynomially closed  (condition 3 of Definition 2.2 of \cite{AAAI}), thus containing $\mathfrak{t}\times\mathfrak{t}_1\times\mathfrak{t}_2$.  So, $\mathcal M$ runs in time ${\mathcal R}\tim$, as desired.

\section{Providence, prudence, quasilegality and unconditionality}\label{sq}
%\marginpar{sq}

In this section we establish certain terminology and  facts necessary for our subsequent proof of the soundness of the induction rule. 

A {\bf numeric (lab)move}\label{xnm} means a (lab)move ending in $\# b$ for some constant $b$. We shall refer to such a $b$ as the {\bf numer}\label{xnmr} of the (lab)move. To make the ``numer'' function total, we stipulate that the  numer of a non-numeric move is $0$ (is the empty string $\epsilon$, that is). 

Consider a bounded formula  $F$. Let $n$ be the number of 
occurrences of choice quantifiers  in $F$, and $\mathfrak{b}_1(\vec{z}_1),\ldots,\mathfrak{b}_n(\vec{z}_n)$ be the bounds used in those occurrences. Let $f(z)$ be the unarification (cf. \cite{cl12}, Section 12) of 
$\max(\mathfrak{b}_1(\vec{z}_1),\ldots,\mathfrak{b}_n(\vec{z}_n))$. Here and elsewhere, as expected,  
$\max(x_1,\ldots,x_n)$\label{xmax}   stands for the greatest   of the numbers $x_1,\ldots,x_n$, and is understood as $0$ if $n=0$. 
Finally, let $\mathfrak{G}$ be the function defined by $\mathfrak{G}(z)=\max(f(z),f^2(z),\ldots,f^n(z))$. Here, as in Section \ref{sfe},  $f^i(z)$ denotes the $n$-fold composition of 
$f$ with itself, i.e., $f(f(\ldots(f(z))\ldots))$, with ``$f$'' repeated $i$ times.  Then we call the functions $f$ and $\mathfrak{G}$ the {\bf subaggregate bound}\label{xsubggf} and the  
  {\bf superaggregate bound}\label{xab} of $F$, respectively. As an aside, for our purposes, a ``much smaller'' function could have been taken in the role of superaggregate bound, but why try to economize.

\begin{lem}\label{lagg}
%\marginpar{lagg}
Assume $\mathcal R$ is a regular boundclass triple, $F$ is an ${\mathcal R}\spa$-bounded formula, and $\mathfrak{G}$ is the superaggregate bound of $F$. Then  $\mathfrak{G}\preceq {\mathcal R}\spa$.
\end{lem}
 
\begin{proof} Assume the conditions of the lemma. Further let $n$, $\mathfrak{b}_1(\vec{z}_1),\ldots,\mathfrak{b}_n(\vec{z}_n)$, $f$ be as in the paragraph preceding Lemma \ref{lagg}. Take a note of the fact that $\mathfrak{b}_1(\vec{z}_1),\ldots,\mathfrak{b}_n(\vec{z}_n)\in{\mathcal R}\spa$. 
If all  tuples $\vec{z}_1,\ldots,\vec{z}_n$ are empty, then ($f$ and hence) $\mathfrak{G}$ is  a constant function and, 
by the linear closure of ${\mathcal R}\spa$, \ $\mathfrak{G}\preceq {\mathcal R}\spa$. Suppose now at least one of the tuples $\vec{z}_1,\ldots,\vec{z}_n$ is nonempty. Pick one variable $z$ among $\vec{z}_1,\ldots,\vec{z}_n$, and consider the pterm $\mathfrak{u}(z)$ obtained from $\mathfrak{b}_1(\vec{z}_1)+\ldots+\mathfrak{b}_n(\vec{z}_n)$ as a result of replacing all variables $\vec{z}_1,\ldots,\vec{z}_n$ by $z$. Since  ${\mathcal R}\spa$ is closed under syntactic variation as well as under $+$, we have $\mathfrak{u}(z)\in {\mathcal R}\spa$. But obviously $f(z)\preceq \mathfrak{u}(z)$. Thus, $f(z)\preceq  {\mathcal R}\spa$. 
In view of condition 4 of Definition 2.2 of \cite{AAAI}, $f(z)\preceq {\mathcal R}\spa$ can be seen to imply 
$f^2(z)\preceq {\mathcal R}\spa$, $f^3(z)\preceq {\mathcal R}\spa$, \ldots. Consequently, by the closure of ${\mathcal R}\spa$ under $+$, 
$f(z)+f^2(z)+\ldots+f^{n}(z)\preceq {\mathcal R}\spa$. But $\mathfrak{G}(z)\preceq f(z)+f^2(z)+\ldots+f^{n}(z)$. Thus, 
$\mathfrak{G}\preceq {\mathcal R}\spa$.
\end{proof}

Recall from \cite{cl12} that a {\bf provident}\label{xprovid} computation   branch  of a given HPM $\mathcal M$ is one containing infinitely many configurations with empty buffer contents (intuitively meaning that $\mathcal M$ has actually made all moves that it has ever started to construct in its buffer).  Then, given a constant game $G$, $\mathcal M$ is said to play $G$ {\bf providently} iff every computation branch of $\mathcal M$ that spells a $\oo$-legal run of $G$ is provident. And $\mathcal M$ is a {\bf provident solution} of $G$ iff $\mathcal M$ is a solution of $G$ and plays it providently. 

Let $H(\vec{y})=H(y_1,\ldots,y_n)$ be a bounded formula with all free variables displayed, $\mathfrak{G}$ be the superaggregate bound of $H(\vec{y})$, and  $\vec{c}=c_1,\ldots,c_n$ be an $n$-tuple of constants. We say that a move $\alpha$ is a {\bf prudent} move\label{xprm} of $H(\vec{c})$ iff the size of the numer of $\alpha$ does not exceed  $\mathfrak{G}|\max(\vec{c})|$.  
The {\bf $H(\vec{c})$-prudentization}\label{xpz}  of $\alpha$ is defined as the following move $\alpha'$. If $\alpha$ is a prudent move of $H(\vec{c})$, then $\alpha'=\alpha$. Suppose now $\alpha$ is not a prudent move of $H(\vec{c})$, meaning that $\alpha$ is a numeric move $\beta\#b$ with an   ``oversized'' numer $b$. In this case we stipulate that $\alpha'=\beta\# a$, where $a$ (as a bitstring) is the longest initial segment of $b$ such that $\beta\# a$ is a prudent move of $H(\vec{c})$. 

Further consider any run $\Gamma$ and either player $\xx\in\{\pp,\oo\}$. We say that $\Gamma$ is a 
{\bf $\xx$-prudent run}\label{xprrun} of $H(\vec{c})$ iff all $\xx$-labeled moves of $\Gamma$ are prudent moves of $H(\vec{c})$.  
When we simply say ``prudent''  without indicating a player, it means both $\pp$-prudent and $\oo$-prudent. 

Further consider any machine $\mathcal M$. By saying that $\mathcal M$  {\bf plays $H(\vec{c})$  prudently},\label{xprpl} we shall mean that, whenever 
$\seq{\oo c_1,\ldots,\oo c_n,\Gamma}$ is a $\oo$-legal run of $\ada H(\vec{y})$ 
 generated by $\mathcal M$, $\Gamma$ is a $\pp$-prudent run of $H(\vec{c})$.  On the other hand, when we say that $\mathcal M$ {\bf plays $H(\vec{y})$ prudently}, we mean that, for any $n$-tuple $\vec{c}$ of constants, $\mathcal M$ plays $H(\vec{c})$ prudently.   
 A {\bf prudent  solution}\label{xprsol} of $H(\vec{y})$ means an HPM that wins $H(\vec{y})$ --- wins $\ada H(\vec{y})$, that is --- and plays $H(\vec{y})$ prudently.

\begin{lem}\label{reason}
%\marginpar{reason} 
There is an effective procedure that takes an arbitrary  bounded formula $H(\vec{y})$,  
 an arbitrary HPM $\mathcal N$ and constructs an HPM $\mathcal L$ such that, for any regular boundclass triple $\mathcal R$, if $H(\vec{y})$ is ${\mathcal R}\spa$-bounded and $\mathcal N$ is an $\mathcal R$ tricomplexity solution of $H(\vec{y})$, then $\mathcal L$ is a provident and prudent $\mathcal R$ tricomplexity solution of $H(\vec{y})$. 
\end{lem}

\begin{idea}    $\mathcal L$ is a machine that waits till $\ada H(\vec{y})$ is brought down to $H(\vec{c})$ for some constants $\vec{c}$ and then,
 through  simulating and mimicking $\mathcal N$ within the specified complexity constraints,  
plays $H(\vec{c})$ just  as $\mathcal N$ would play it, with essentially the only difference that  each (legal) move $\alpha$ made by $\mathcal N$ is made by $\mathcal L$ in the prudentized form $\alpha'$.  This does not decrease the chances of ${\mathcal L}$ (compared with  those of $\mathcal N$) to win:  
imprudent  moves are at best inconsequential and at worst disadvantageous (resulting in a loss of the corresponding subgame) for a player, so,  if the machine wins the game while it makes the imprudent  move  $\alpha$, it would just as well (and ``even more so'') win the game if it had made the prudent move  $\alpha'$ instead. This is how prudence is achieved. As for providence, $\mathcal L$ achieves it by never putting anything into its buffer unless it has already decided to make a move, after seeing that the simulated $\mathcal N$ has moved. 

Of course, the above strategy may yield  some discrepancies between the contents of $\mathcal L$'s run tape and $\mathcal N$'s imaginary run tape: it is possible  that the latter is showing a ($\pp$-labeled) move $\alpha$ while the former is showing only a properly smaller part (prudentization) $\alpha'$ of $\alpha$.   To neutralize this problem, every time the simulated $\mathcal N$ is trying to read some symbol $b$ of $\alpha$ on its run tape, $\mathcal L$ finds $b$ through resimulating the corresponding portion of the work of $\mathcal N$. This, of course,  results in $\mathcal L$'s being slower than $\mathcal N$; yet, due to $\mathcal R$'s   being regular, things can be arranged so that the  running time of $\mathcal L$ still remains within the admissible limits. 
\end{idea}

A detailed proof of Lemma \ref{reason}, which  materializes the above idea, is given in Appendix \ref{sap1}. 
It can be omitted rather safely by a reader so inclined. The same applies to the forthcoming Lemma \ref{vasa}, whose proof idea is presented in this section and whose relatively detailed proof is given in Appendix \ref{sap2}. 
 
When $\Gamma$ is a run, we let 
\[\mbox{\it $\Gamma^\top$ \ (resp. $\Gamma^\bot$)}\label{xcb}\]
denote the result of deleting in $\Gamma$ all $\oo$-labeled (resp. $\pp$-labeled) moves. 

For a constant game $A$ and run $\Gamma$,  we say that $\Gamma$ is a {\bf $\pp$-quasilegal}\label{xqr} (resp. {\bf $\oo$-quasilegal}) run of $A$ iff there is a legal run $\Delta$ of $A$ such that $\Delta^\top=\Gamma^\top$ (resp. $\Delta^\bot=\Gamma^\bot$). If we say ``{\bf quasilegal}''\label{xqwsda} without the prefix ``$\pp$-'' or ``$\oo$-'', it is to be understood as ``both $\pp$-quasilegal and $\oo$-quasilegal''. 
We say that an HPM $\mathcal M$ {\bf plays $A$ quasilegally}\label{xqp} iff every run generated by $\mathcal M$ is a $\pp$-quasilegal run of $A$. A {\bf quasilegal solution}\label{xqs} of $A$ is a solution of $A$ that plays $A$ quasilegally. 

Our definitions of ``$\mathcal M$ plays \ldots\ providently'' and ``$\mathcal M$ plays \ldots\ prudently'', just like our earlier \cite{cl12} definitions of running within  given  complexity bounds,  only look at  (computation branches  that spell) $\oo$-legal runs of a given game. Below we define  stronger --- ``unconditional'' ---  versions of such concepts, where the adversary's  having made an illegal move is no longer an excuse for the player to stop acting in the expected manner. Namely:

We say that an HPM $\mathcal M$  plays {\bf unconditionally providently},\label{xupp} or that $\mathcal M$ is {\bf unconditionally provident}, iff all computation branches of $\mathcal M$ are provident (note that the game that is being played is no longer relevant). 

Consider an HPM $\mathcal M$, a bounded formula  $H=H(\vec{y})=H(y_1,\ldots,y_n)$ with all free variables displayed, and an $n$-tuple $\vec{c}=c_1,\ldots,c_n$ of constants.
We say that $\mathcal M$  {\bf plays $H(\vec{c})$ unconditionally prudently}\label{xup} iff, whenever
$\seq{\oo c_1,\ldots,\oo c_n,\Gamma}$ is a run (whether it be $\oo$-legal or not) generated by $\mathcal M$, $\Gamma$ is a $\pp$-prudent run of $H(\vec{c})$.  Next, when we say that $\mathcal M$ {\bf plays $H(\vec{y})$ unconditionally prudently}, we mean that, for any $n$-tuple $\vec{c}$ of constants,  $\mathcal M$ plays $H(\vec{c})$ unconditionally prudently.

The following definition of the unconditional versions of our   complexity concepts is  obtained from Definition 5.2 of \cite{cl12} by simply dropping the condition ``$\oo$-legal'' on the plays considered, and also removing any mention of a game $A$ that is being played because the latter is no longer relevant.

\begin{defi}\label{deftcsunc}
%\marginpar{deftcsunc}
Let $\mathcal M$ be an HPM, and $h$ a unary arithmetical function (if $h$ is not unary, then it should be replaced by its unarification according to Convention 12.2 of \cite{cl12}).  We say that: 
\begin{enumerate}[label=\arabic*.]
\item {\bf $\mathcal M$ runs (plays) in unconditional amplitude\label{x0amplitude} $h$} iff,  in every computation branch of $\mathcal M$, whenever   $\mathcal M$ makes a move $\alpha$, the magnitude of $\alpha$  does not exceed   $h(\ell)$,  where $\ell$ is the background of $\alpha$; 

\item {\bf $\mathcal M$  runs (plays) in unconditional space $h$} iff,  in every computation branch of $\mathcal M$,  the spacecost of any given clock cycle $c$ does not exceed  $h(\ell)$,   where $\ell$ is the background of $c$; 

\item {\bf $\mathcal M$ runs (plays) in unconditional time $h$} iff, in every computation branch of $\mathcal M$, whenever   $\mathcal M$ makes a move $\alpha$,   the timecost of $\alpha$ does not exceed  $h(\ell)$,  where $\ell$ is the background of $\alpha$.  
\end{enumerate}
\end{defi}\medskip

\noindent The above definition and the related concepts naturally --- in the same way as in the old, ``conditional'' cases --- extend from bounds 
(as functions)  to boundclasses, as well as  bound triples or boundclass triples.  For instance,  where  $\mathcal C$ is a boundclass, we say that 
 $\mathcal M$ runs (plays) in {\bf unconditional time $\mathcal C$} iff it runs in unconditional time $h$ for some $h\in{\mathcal C}$; where 
$\mathcal R$ is a boundclass triple, we say that $\mathcal M$ runs (plays) in {\bf unconditional tricomplexity $\mathcal R$} iff it  runs in unconditional 
 amplitude ${\mathcal  R}\amp$, unconditional space ${\mathcal R}\spa$ and unconditional time ${\mathcal R}\tim$; etc.

\begin{lem}\label{vasa}
%\marginpar{vasa}
There is an effective procedure that takes an arbitrary bounded formula $H(\vec{y})$,
 an arbitrary HPM $\mathcal L$ and constructs an HPM $\mathcal M$ such that, as long as $\mathcal L$ is a provident solution of $H(\vec{y})$, 
 the following conditions are satisfied: 
\begin{enumerate}[label=\arabic*.]
\item $\mathcal M$ is a quasilegal and unconditionally provident solution of $H(\vec{y})$.

\item If $\mathcal L$ plays $H(\vec{y})$ prudently, then $\mathcal M$ plays $H(\vec{y})$ unconditionally prudently.

\item For any arithmetical functions $\mathfrak{a},\mathfrak{s},\mathfrak{t}$,    
if $\mathcal L$ plays $H(\vec{y})$ in 
tricomplexity $(\mathfrak{a},\mathfrak{s},\mathfrak{t})$, then $\mathcal M$ plays in unconditional tricomplexity $(\mathfrak{a},\mathfrak{s},\mathfrak{t})$.
\end{enumerate}
\end{lem}    

\noindent\begin{idea}\label{0id2} 
In our preliminary attempt of constructing $\mathcal M$, 
we let it be a machine that works exactly like $\mathcal L$, except that $\mathcal M$ retires as soon as it detects that the play has gone illegal. This way, unlike $\mathcal L$, $\mathcal M$ is precluded from using Environment's illegal actions as an excuse for some undesirable  behavior of its own, such as making inherently illegal or oversized moves, or using excessive resources. That is, while $\mathcal L$  ``behaves well'' only on the condition of Environment playing legally, $\mathcal M$ is guaranteed to ``behave well'' unconditionally, because in legal cases $\mathcal M$'s behavior coincides with that of $\mathcal L$, and in illegal cases $\mathcal M$ simply does not ``behave'' at all. An unretired or not-yet-retired $\mathcal M$ consumes exactly the same amount of time and space as $\mathcal L$ does, because keeping track of whether the play has gone illegal only requires maintaining a certain bounded amount of information, which can be done through state (rather than work-tape) memory and hence done without any time or space overhead whatsoever. The only problem with the above solution is that $\mathcal M$'s buffer may not necessarily be empty at the time we want it to retire, and if so, then $\mathcal M$ is not unconditionally provident. This minor complication is neutralized by letting $\mathcal M$, before retiring, extend (if necessary) the buffer content to a shortest possible move adding which to the already generated run does not destroy its $\pp$-quasilegality, and  then empty the buffer by making such a move in the play.  \end{idea}

In what follows, we will be using the word ``{\bf reasonable}''\label{xreas} (``{\bf reasonably}'') as an abbreviation of ``quasilegal(ly) and unconditionally prudent(ly)''. ``{\bf Unreasonable}'' (``{\bf unreasonably}''), as expected, means ``not reasonable'' (``not reasonably''). 
We can now strengthen Lemma \ref{reason} as follows:

\begin{lem}\label{mr22}
%\marginpar{mr22}
There is an effective procedure that takes an arbitrary bounded formula $H(\vec{y})$,
 an arbitrary HPM $\mathcal N$ and constructs an HPM $\mathcal M$ such that, for any regular boundclass triple $\mathcal R$, if $H(\vec{y})$ is ${\mathcal R}\spa$-bounded and $\mathcal N$ is an $\mathcal R$ tricomplexity solution of $H(\vec{y})$, then $\mathcal M$ is a reasonable, unconditionally provident and unconditionally $\mathcal R$ tricomplexity solution of $H(\vec{y})$. 
\end{lem}

\begin{proof} Immediately from Lemmas \ref{reason} and \ref{vasa}. \end{proof}

\section{Soundness of Induction}\label{ssind}
%\marginpar{ssind}

\begin{thm}\label{feb9eei} 
%\marginpar{feb9eei}
Consider any regular boundclass triple $\mathcal R$. 
There is an ($\mathcal R$-independent) effective procedure that takes an arbitrary application of $\mathcal R$-Induction,  arbitrary HPMs ${\mathcal N},{\mathcal K}$ and constructs an HPM $\mathcal M$ such that, if  $\mathcal N$ and $\mathcal K$ are $\mathcal R$ tricomplexity solutions of the two premises,  then $\mathcal M$ is an $\mathcal R$ tricomplexity solution of the conclusion.
\end{thm}

The rest of this long section is devoted to a proof of the above theorem. It should be noted that some ideas used in our proof are borrowed from \cite{cla5}. 

Consider any regular boundclass triple  $\mathcal R$  and any   application 
 %\marginpar{r2}
\begin{equation}\label{r2}
 \frac{ F(0) \hspace{30pt} F(x)\mli F(x\successor) }{ x\leq \mathfrak{b}|\vec{s}|\mli F(x) }\end{equation}
 of ${\mathcal R} $-Induction. 
Assume  \[\vec{v}=v_1,\ldots,v_{\mathfrak{v}}\] --- fix this number $\mathfrak{v}$\label{xvaq} --- are exactly the free variables of $F(x)$ other than $x$ listed in the lexicographic order, and let
 us correspondingly rewrite (\ref{r2}) as 
%\marginpar{r22}
\begin{equation}\label{r22}
 \frac{ F(0,\vec{v}) \hspace{30pt} F(x,\vec{v})\mli F(x\successor,\vec{v}) }{ x\leq \mathfrak{b}|\vec{s}|\mli F(x,\vec{v}) }.\end{equation}
Further, assume that  $\mathcal N$ and $\mathcal K$ are $ {\mathcal R}$ tricomplexity  solutions of the left and the right premise of (\ref{r22}), respectively. In view of Lemma \ref{mr22}, we may and will assume that $\mathcal N$ and $\mathcal K$ are reasonable, unconditionally provident  and unconditionally  $ {\mathcal R}$ tricomplexity  solutions of the corresponding premises.
In view of the closure of all three components of $\mathcal R$ under syntactic variation, in combination with the other relevant closure conditions from Definition 2.2 of \cite{AAAI},  there is one common triple  \[(\mathfrak{a},\mathfrak{s},\mathfrak{t})\in   {\mathcal R}\amp\times {\mathcal R}\spa\times {\mathcal R}\tim\]  of unary bounds  such that both $\mathcal N$ and $\mathcal K$ run in unconditional $(\mathfrak{a},\mathfrak{s},\mathfrak{t})$ tricomplexity.  Fix these $\mathfrak{a},\mathfrak{s},\mathfrak{t}$ for the rest of this section.

We want to construct an ${\mathcal R}$ tricomplexity solution $\mathcal M$\label{xmm3} of the conclusion of (\ref{r22}).  It is important to point out that, as in the case of Comprehension, our construction of $\mathcal M$ does not rely on the assumptions on $\mathcal N$ and $\mathcal K$ that we have just made.  Also, the pathological case of $F(x,\vec{v})$ having no 
free occurrences of $x$ is trivial and,  for the sake of simplicity, we exclude it from our considerations. $\mathcal M$ will be designed as a machine with a single work tape.

As usual in such cases, we adopt the  Clean Environment Assumption (cf. Section 8 of \cite{cl12}), according to which $\mathcal M$'s adversary never makes illegal moves of the game under consideration.

At the beginning,  our $\mathcal M$  waits for Environment to choose constants for all free  variables of the conclusion of (\ref{r22}).  We rule out the possibility that the adversary never does so, because then $\mathcal M$ is an automatic winner trivially running in zero amplitude, zero space and zero time unless it deliberately tries not to. 
For the rest of this section, assume 
 $k\label{xkqa}$ is the  constant chosen for the variable $x$,
 $\vec{c}=\vec{c}_1,\ldots,\vec{c}_{\mathfrak{v}}$ are the constants chosen for $\vec{v}$, and $\vec{d}$ are the constants chosen for $\vec{s}$. Since the case of $k= 0$ is straightforward and not worth paying separate attention, for further simplicity considerations we will  assume for the rest of this section that $k\geq 1$. From now on, we shall write $F'(x)$\label{xfp} as an abbreviation of $F(x,\vec{c})$. 

The above event of Environment's initial choice of constants brings the conclusion of (\ref{r22}) down to $k\leq \mathfrak{b}|\vec{d}|\mli F(k,\vec{c})$, i.e., to $k\leq \mathfrak{b}|\vec{d}|\mli F'(k)$. $\mathcal M$ computes $\mathfrak{b}|\vec{d}|$ and compares it with $k$.  By condition 1 of Definition 2.2 of \cite{AAAI},  this can be done in space ${\mathcal R}\spa$ and time 
${\mathcal R}\tim$. If 
 $k\leq \mathfrak{b}|\vec{d}|$ is false, $\mathcal M$ retires, obviously being the winner and satisfying the expected complexity conditions. For the rest of this section, we rule out this straightforward case and, in the scenarios that we consider, assume that   $k\leq \mathfrak{b}|\vec{d}|$ is true.\label{212014}  

We shall write ${\mathcal H}_{0}$\label{xhnill1} as an abbreviation of the phrase ``${\mathcal N}$  in the scenario where the adversary, at the beginning of the play, has chosen the constants $\vec{c}$ for the variables $\vec{v}$\hspace{2pt}''. So, for instance, when saying that ${\mathcal H}_0$ moves on cycle $t$, it is to be understood as that, in the above scenario, ${\mathcal N}$ moves on cycle $t$. As we see, strictly speaking, ${\mathcal H}_0$ is not a separate ``machine'' but rather it is just ${\mathcal N}$ in a certain partially fixed scenario.\footnote{The beginning of that scenario is fixed but the continuations may vary.}  Yet, for convenience and with some abuse of language, in the sequel we may terminologically and even conceptually treat ${\mathcal H}_0$ as if it was a machine in its own right --- namely, the machine that works just like ${\mathcal N}$ does in the scenario where the adversary, at the beginning of the play, has chosen the constants $\vec{c}$ for the variables $\vec{v}$.  Similarly, for any $n\geq 1$, we will write 
${\mathcal H}_{n}$\label{xhnill2} for the ``machine'' that works just like ${\mathcal K}$ does in the scenario where the adversary, at the beginning of the play, has chosen the constants $\vec{c}$ for the variables $\vec{v}$ and the constant $n- 1$ for the variable $x$. So, ${\mathcal H}_{0}$ (thought of as a machine) wins the constant game $F'(0)$ and, for each $n\geq 1$, ${\mathcal H}_{n}$ wins the constant game $F'(n-1)\mli F'(n)$.

In the same style as the notation ${\mathcal H}_n$ is used, we write ${\mathcal M}_k$\label{xmk} for the ``machine'' that works just like ${\mathcal M}$ does after the above event of Environment's having chosen $k$,  $\vec{c}$ and $\vec{d}$ for $x$, $\vec{v}$ and $\vec{s}$, respectively. So, in order to complete our description of $\mathcal M$, it will suffice to simply define ${\mathcal M}_k$ and say that, after Environment has chosen constants for all free variables of the conclusion of (\ref{r22}),  $\mathcal M$ continues playing like (``turns itself into'') ${\mathcal M}_k$. Correspondingly,  in showing that $\mathcal M$ wins $\ada \bigl(x\leq \mathfrak{b}|\vec{s}|\mli F(x,\vec{v})\bigr)$, it will be sufficient to show that ${\mathcal M}_k$ wins 
 $k\leq \mathfrak{b}|\vec{d}|\mli F'(k)$. 

\begin{rem}\label{apology}
%\marginpar{apology}
It should be noted that our treating of ${\mathcal H}_0,\ldots,{\mathcal H}_k$ and ${\mathcal M}_k$ as ``machines''  may occasionally generate some ambiguity or terminological inconsistencies, for which the author  wants to apologize in advance. For instance, when talking about the content of ${\mathcal H}_0$'s run tape or the run spelled by a given computation branch of ${\mathcal H}_0$, ${\mathcal N}$'s adversary's initial moves $\oo c_1,\ldots,\oo c_{\mathfrak{v}}$ may or may not be meant to be included. Such ambiguities or inconsistencies, however, can usually be easily resolved based on the context. 
\end{rem}

In the informal description below, 
we use the term ``{\bf synchronizing}''\label{xsync} to mean applying copycat between two (sub)games of the form $A$ 
and $\gneg A$. This means mimicking  one player's moves in $A$ as the other player's moves in $\gneg A$, and vice versa. The effect achieved this way is that the games to which $A$ and $\gneg A$  eventually evolve (the final positions hit by them, that is) will be of the form $A'$ and $\gneg A'$ --- that is, one will remain the negation of the other, so that one will be won by a given player iff the other is lost by the same player.  

The idea underlying the work of ${\mathcal M}_k$ can be summarized by saying that what ${\mathcal M}_k$ does is a synchronization  between $k+ 2$ games, real or imaginary (simulated). Namely:
\begin{itemize}
\item It synchronizes the imaginary play of $F'(0)$ by ${\mathcal H}_{0}$ with the antecedent of the imaginary play of $F'(0)\mli F'(1)$ by ${\mathcal H}_{1}$.
\item For each $n$ with $1\leq n< k$, it synchronizes the consequent of the imaginary play of $F'(n- 1)\mli F'(n)$ by ${\mathcal H}_{n}$ with the 
antecedent of the imaginary play of $F'(n)\mli F'(n+ 1)$ by ${\mathcal H}_{n+ 1}$.
\item It (essentially) synchronizes the consequent of the imaginary play of  $F'(k- 1)\mli F'(k)$ by ${\mathcal H}_{k}$ with the real play in the consequent of $k\leq \mathfrak{b}|\vec{d}|\mli F'(k)$.
\end{itemize}
Therefore, since ${\mathcal H}_{0}$ wins $F'(0)$ and each ${\mathcal H}_{n}$ with $1\leq n\leq k$ wins $F'(n- 1)\mli F'(n)$, ${\mathcal M}_k$ wins $k\leq \mathfrak{b}|\vec{d}|\mli F'(k)$ and thus $\mathcal M$ wins (the $\ada$-closure of) $x\leq \mathfrak{b}|\vec{s}|\mli F(x,\vec{v})$, as desired. 

If space complexity was of no concern, a synchronization in the above-outlined style could be achieved by simulating all imaginary plays in parallel. Our general case does not allow us doing so though, and synchronization should be conducted in a very careful way. Namely, a parallel simulation of all plays is not
possible, because there are  up to $\mathfrak{b}|\vec{s}|$ simulations to perform, and there is no guarantee that this does not take us beyond the ${\mathcal R}\spa$ space limits.  So, instead, simulations should be performed is some sequential rather than parallel manner, with subsequent simulations recycling the space used by the previous ones, and with the overall procedure keeping forgetting the results of most previous simulations and recomputing the same information over and over many times.  We postpone our description of how ${\mathcal M}_k$ exactly works to Section \ref{sn}, after having elaborated all necessary preliminaries in Sections \ref{sagree}-\ref{saggr}.     

\subsection{Soon enough or never}\label{sagree}
%\marginpar{sagree}

\begin{nota}\label{not1}
%\marginpar{not1}
{\em We agree that throughout the rest of Section  \ref{ssind}: 
\begin{enumerate}[label=\arabic*.]
\item $\mathfrak{l}$\label{xl} denotes the length $|a|$ of the greatest constant $a$ among $k,\vec{c},\vec{d}$.

\item $\mathfrak{e}_\top$ (resp. $\mathfrak{e}_\bot$) is the maximum number of $\pp$-labeled (resp. $\oo$-labeled) moves in any legal run of $F'(0)$,\label{xe} and $\mathfrak{e}=\mathfrak{e}_\top+\mathfrak{e}_\bot$. 

\item $\mathfrak{G}$\label{xagrba} is the superaggregate bound of $F(x,\vec{v})$.

\item $\mathfrak{L}(w,u)$\label{xlla}  abbreviates   
\[\mathfrak{r}\times (u+1)^{\mathfrak{g}}\times\Bigl((\mathfrak{v}+1)\times (w+2)+2\mathfrak{e}\bigl(\mathfrak{G}(w) + \mathfrak{h} +2\bigr)+1\Bigr) \times \mathfrak{q}^{\mathfrak{g}u} \times 2\mathfrak{e},\] where  $\mathfrak{v}$, as we remember, is the number of variables in $\vec{v}$, and:
\begin{itemize}
  \item $\mathfrak{r}$\label{xrqa}  is the maximum number of states  of the two  machines ${\mathcal N}$ and ${\mathcal K}$;
  \item $\mathfrak{g}$\label{xgqa}  is the maximum number of work tapes  of the two  machines ${\mathcal N}$ and ${\mathcal K}$;
  \item $\mathfrak{q}$\label{xqqa}  is the maximum number of symbols that may ever appear on any of the  tapes of the two machines ${\mathcal N}$ and ${\mathcal K}$;
  \item $\mathfrak{h}$\label{xhqa} is the length of the longest string $\beta$ containing no $\#$ such that   $\beta$ is a prefix of some move of some legal run of $F'(0)$. 
\end{itemize}
\end{enumerate} }  
\end{nota}\medskip  %%%JK

\noindent In the sequel, we may say about a machine or its adversary that it plays so and so (reasonably, prudently, etc.) without mentioning the context-setting game that is played. As expected, it will be understood that, in such cases, the game is:
 $\ada  \bigl(x\leq \mathfrak{b}|\vec{s}|\mli F(x,\vec{v})\bigr)$ if the machine is $\mathcal M$;
 $\ada F(0,\vec{v})$ if the machine is $\mathcal N$;
 $\ada \bigl(F(x,\vec{v})\mli F(x\successor,\vec{v})\bigr)$ if the machine is $\mathcal K$;
 $F'(0)$ if the machine is ${\mathcal H}_0$;
 $F'(n-1)\mli F'(n)$ if the machine is ${\mathcal H}_n$ with $1\leq n\leq k$;
 and $k\leq \mathfrak{b}|\vec{d}|\mli F'(k)$ if the machine is ${\mathcal M}_k$.  

Below, $\Upsilon_0$ denotes the sequence of $\mathfrak{v}$   $\oo$-labeled moves signifying the choice of the constants $\vec{c}$ for the free variables $\vec{v}$  of $F(0,\vec{v})$ --- that is, $\Upsilon_0=\seq{\oo \#c_1,\ldots,\oo \#c_{\mathfrak{v}}}$.
And  
$\Upsilon_n$, for $n\in\{1,\ldots,k\}$, denotes the sequence of $\mathfrak{v}+1$ $\oo$-labeled moves signifying the choice of the constants $n-1$ and $\vec{c}$ for the free variables $x$ and $\vec{v}$ of $F(x,\vec{v})\mli F(x\successor,\vec{v}) $, respectively.

Whenever we say that ${\mathcal H}_n$'s {\bf adversary plays quasilegally},\label{xgh}
 we shall mean that we are only considering the runs $\Gamma$ generated by ${\mathcal H}_n$ (i.e., runs $\seq{\Upsilon_0,\Gamma}$ generated by $\mathcal N$ and runs  $\seq{\Upsilon_n,\Gamma}$ generated by $\mathcal K$) such that $\Gamma$ is a $\oo$-quasilegal  run of $F'(0)$ (if $n=0$) or $F'(n-1)\mli F'(n)$ (if $n\geq 1$). Similarly for the adversary's playing {\bf unconditionally prudently} or  {\bf reasonably}. By  the {\bf symbolwise length}\label{xsyl} of a position $\Phi$ we shall mean  the number of cells that $\Phi$ takes when spelled on the run tape. Similarly for labmoves. 

\begin{lem}\label{m29aaa}
%\marginpar{m29aaa}
For any $n\in\{0,\ldots,k\}$, at any time in any play by  ${\mathcal H}_n$, as long as ${\mathcal H}_n$'s  adversary
plays reasonably, the symbolwise length of the position
spelled on the run tape of ${\mathcal H}_n$ does not exceed $(\mathfrak{v}+1)\times (\mathfrak{l}+2)+2\mathfrak{e}\bigl(\mathfrak{G}(\mathfrak{l}) + \mathfrak{h} +2\bigr)$.
\end{lem} 

\begin{proof}   Any position spelled on the run tape of ${\mathcal H}_n$ looks like $\seq{\Upsilon_n,\Gamma}$. The symbolwise length of the $\Upsilon_n$ part is at most $(\mathfrak{v}+1)\times (\mathfrak{l}+2)$, with $\mathfrak{v}+1$ being the (maximum) number of labmoves in $\Upsilon_n$ and $\mathfrak{l}+2$ being the maximum symbolwise length of each labmove, including the prefix  $\oo\#$. 
 By our assumption, ${\mathcal H}_n$\footnote{$\mathcal N$ (if $n\geq 1$) or $\mathcal K$ (if $n=0$), to be more precise.}  plays  reasonably.  The present lemma additionally assumes that so does ${\mathcal H}_n$'s adversary. If so, it is obvious that the symbolwise length  of no labmove in the $\Gamma$ part  can exceed $\mathfrak{G}(\mathfrak{l}) + \mathfrak{h} +2$; and there are at most 
$2\mathfrak{e}$ such labmoves. The symbolwise length of the $\Gamma$ part is thus at most $2\mathfrak{e}\bigl(\mathfrak{G}(\mathfrak{l}) + \mathfrak{h} +2\bigr)$.
\end{proof}

The following lemma states that the ${\mathcal H}_n$'s move soon enough or never, with $\mathfrak{L}$ acting as a ``statute of limitations'' function:   
\begin{lem}\label{m29a}
%\marginpar{m29a}
Consider any machine ${\mathcal H}_n\in\{{\mathcal H}_0,\ldots,{\mathcal H}_k\}$, and any cycle (step, time) $c$ of any play by ${\mathcal H}_n$. Assume that $u$ is the spacecost of 
cycle $c+\mathfrak{L}(\mathfrak{l},u)$. Further assume that  
the adversary of ${\mathcal H}_n$ plays reasonably, and it does not move  at any time $d$ with $d> c$.  Then ${\mathcal H}_n$ does not move at any time $d$ with $d> c+\mathfrak{L}(\mathfrak{l},u)$.
\end{lem}

\begin{proof} Assume the conditions of the lemma and, remembering that (not only ${\mathcal H}_n$'s adversary but also) ${\mathcal H}_n$ plays reasonably,  answer the following question: How many different configurations of ${\mathcal H}_n$ --- ignoring the buffer content component --- are there that may emerge in the play between (including) steps $c$ and $c+\mathfrak{L}(\mathfrak{l},u)$? We claim that this quantity cannot exceed  $\mathfrak{L}(\mathfrak{l},u)$. Indeed, there are at most $\mathfrak{r}$ possibilities for the state component of such a configuration. These possibilities are accounted for by the 1st of the five factors of $\mathfrak{L}(\mathfrak{l},u)$. Next, clearly there are at most $(u+1)^{\mathfrak{g}}$ possibilities for the locations of the work-tape heads,\footnote{Remember that a scanning head of an HPM can never move beyond the leftmost blank cell.}  which is accounted for by the 2nd factor of $\mathfrak{L}(\mathfrak{l},u)$. Next, in view of Lemma \ref{m29aaa}, there are at most $(\mathfrak{v}+1)\times (\mathfrak{l}+2)+2\mathfrak{e}\bigl(\mathfrak{G}(\mathfrak{l}) + \mathfrak{h} +2\bigr)+1$ possible locations of the run-tape head, and this number is accounted for by the 3rd factor of $\mathfrak{L}(\mathfrak{l},u)$. Next,  there are at most $\mathfrak{q}^{\mathfrak{g}u}$ possibilities for the contents of the $\mathfrak{g}$ work tapes, and this number is accounted for by the 4th factor of 
$\mathfrak{L}(\mathfrak{l},u)$. Finally, the run-tape content can change (be extended) at most $2\mathfrak{e}$ times,  and this number is accounted for by the 5th factor of $\mathfrak{L}(\mathfrak{l},u)$. Thus, there are at most $\mathfrak{L}(\mathfrak{l},u)$ possible configurations (ignoring the buffer content component), as promised. If so, some configuration repeats itself between steps $c$ and $c+\mathfrak{L}(\mathfrak{l},u)$, meaning that ${\mathcal H}_n$ is in a loop which will be repeated again and again forever.  Within that loop ${\mathcal H}_n$ makes no moves, for otherwise the run-tape-content component of the configurations would keep changing (expanding). \end{proof}  

\subsection{The procedure \texorpdfstring{$\simm$}{Sim}}\label{ssim}
%\marginpar{ssim} 

We define an {\bf organ}\label{xorgan}  to be a pair $O=(\vec{\alpha},p)$, where $\vec{\alpha}$, called the 
{\bf payload}\label{xpayload} of $O$, is a (possibly empty) finite sequence of moves, and $p$, called the 
{\bf scale}\label{xscale} of $O$, is a positive integer. 

A {\bf signed organ}\label{xsuperorgan} $S$  is $-O$ or $+O$, where $O$ is an organ. In the first case we say that $S$ is 
{\bf negative},\label{xneggg} and in the second case we say that it is  {\bf positive}.\label{xpostv} The {\bf payload} and the {\bf scale} of such an $S$ mean those of $O$.

A {\bf body}\label{xbody} is a  tuple $B=(O_1,\ldots,O_s)$ of organs.
 The number $s$ is said to be the {\bf size}\label{xsize} of such a body $B$.

A {\bf $\simm$-appropriate triple}\label{xsat} is $(A,B,n)$, where $n\in\{0,\ldots,k\}$, $B$ is a nonempty body, and $A$ is a body required to be empty if $n=0$.

Our ${\mathcal M}_k$ simulates the work of the machines ${\mathcal H}_0,\ldots,{\mathcal H}_k$ through running the procedure $\simm$\label{xsim1} defined below. This procedure takes a $\simm$-appropriate triple $(A,B,n)$ as an argument, and returns  a pair $(S,u)$, where $S$ is a  signed organ and $u$ is a natural number. We indicate this relationship by writing $\simm_n(A,B)=(S,u)$.
We usually understand $\simm_{n}$ as the two-argument procedure --- and/or the corresponding function --- resulting from fixing the third argument of $\simm$ to $n$. Similarly for the later-defined $\simm_{n}^{\bullet}$, $\simm_{n}^{\leftarrow}$,
 $\simm_{n}^{\rightarrow}$.

We first take a brief informal  look at $\simm_n$ with $1\leq n\leq k$ ($\simm_0$ needs to be considered separately). Assume $A=\bigl((\vec{\alpha}_1,p_1),\ldots,(\vec{\alpha}_a,p_a)\bigr)$ and $B=\bigl((\vec{\beta}_1,q_1),\ldots,(\vec{\beta}_b,q_b)\bigr)$.  The argument 
 $(A,B)$ determines the scenario of the work of ${\mathcal H}_n$ that needs to be simulated. In this scenario, the moves made by ${\mathcal H}_n$'s adversary in the antecedent (resp. consequent) of $F'(n-1)\mli F'(n)$ come from 
$ \vec{\alpha}_1,\ldots,\vec{\alpha}_a$ (resp. $\vec{\beta}_1,\ldots,\vec{\beta}_b$). The simulation starts by ``{\em fetching}'' the organ $(\vec{\beta}_1,q_1)$  from $B$ and  tracing the first $q_1$ steps of ${\mathcal H}_n$ in the scenario where, at the very beginning of the play, i.e., on clock cycle $0$, the adversary made the moves $\vec{\beta}_1$ in the consequent of $F'(n-1)\mli F'(n)$, all at once. 
Which organ  is fetched next   depends on how things have evolved so far, namely, on whether within the above $q_1$ 
steps ${\mathcal H}_n$ has responded by a nonempty or empty sequence $\vec{\nu}$ of moves in the consequent of $F'(n-1)\mli F'(n)$. If 
$\vec{\nu}\not=\seq{}$, then the next organ  to be fetched will be  the first not-yet-fetched organ of $B$, i.e., $(\vec{\beta}_2,q_2)$; and if 
$\vec{\nu}=\seq{}$, then  the next organ  to be fetched  will be the first not-yet-fetched organ of $A$, i.e.,  $(\vec{\alpha}_1,p_1)$. 
After fetching such an organ $(\vec{\delta},r)\in\{(\vec{\beta}_2,q_2),(\vec{\alpha}_1,p_1)\}$,  the simulation of ${\mathcal H}_n$ rolls back to the point $w$ at which ${\mathcal H}_n$ made its last move (if there are no such moves, then $w=0$), and continues from there for additional $r$ steps in the scenario where, at the very beginning of the episode, i.e., at step $w$, ${\mathcal H}_n$'s imaginary adversary responded by the moves $\vec{\delta}$, all at once, in the corresponding component (consequent if $\vec{\nu}\not=\seq{}$ and antecedent if $\vec{\nu}=\seq{}$) of $F'(n-1)\mli F'(n)$. 
 As in the preceding case, what to fetch next ---  the leftmost not-yet-fetched organ of $B$ or that of $A$ --- depends on whether within the above $r$ steps (i.e., steps $w$ through $w+r$) ${\mathcal H}_n$ responds by a nonempty or an empty sequence of moves in the consequent of $F'(n-1)\mli F'(n)$. 
And similarly for the subsequent steps: whenever ${\mathcal H}_n$  responds to the last series $\vec{\beta}_i$ (resp. $\vec{\alpha}_i$) of the imaginary adversary's moves with a nonempty   sequence $\vec{\nu}$ of moves in the consequent of   $F'(n-1)\mli F'(n)$ within $ {q_i}$ (resp. $ {p_i}$) steps, the  next organ $(\vec{\delta},r)$ to be fetched will be the first not-yet-fetched organ of $B$; 
otherwise such a $(\vec{\delta},r)$ will be  the  first not-yet-fetched organ of $A$. In either case, the simulation of ${\mathcal H}_n$ rolls back to the point $w$ at which ${\mathcal H}_n$ made its last move, and continues from there for additional $r$ steps in the scenario where, at step $w$, ${\mathcal H}_n$'s imaginary adversary responded by the moves $\vec{\delta}$ in the corresponding component (consequent if $\nu\not=\seq{}$ and antecedent if $\nu=\seq{}$) of the game. 
The overall procedure ends when it tries to fetch the next not-yet-fetched organ of   $A$ (resp. $B$) but finds that there are no such organs remaining. Then the  $S$ part of the output $(S,u)$ of $\simm_n(A,B)$ is stipulated to be $-(\vec{\sigma},r)$ (resp. $+(\vec{\sigma},r)$), where 
$\vec{\sigma}$ is the sequence of moves made by ${\mathcal H}_n$ in the antecedent (resp. consequent) of $F'(n-1)\mli F'(n)$  since the last organ of $A$ (resp. $B$) was fetched, and $r$ is the scale of that organ. As for the $u$ part of the output $(S,u)$, in either case it is simply the maximum number of non-blank cells on any (any {\em one}) work tape of ${\mathcal H}_n$ at the end of the simulated episode. 

The case of $\simm_0((),B)$ is similar but simpler. In fact, $\simm_0((), B)$ is a special case of $\simm_n(A,B)$ if we  
think of $F'(0)$ as the implication $F'(-1)\mli F'(0)$ with the dummy antecedent $F'(-1)=\twg$.

In order to be able to define $\simm_0$ or $\simm_n$ ($1\leq n\leq  k$) more formally, we need a couple of notational conventions.

When $\vec{\alpha}=\seq{\alpha_1,\ldots,\alpha_s}$ is a sequence of moves,  $\omega$ is a string over the keyboard alphabet (such as, say, ``$0.$'', ``$1.$'' or the empty string) and $\xx$ is one of the players $\pp$ or $\oo$, we shall write $\xx\omega\vec{\alpha}$\label{xder} for the run  $\seq{\xx\omega\alpha_1,\ldots,\xx\omega\alpha_s}$. 

Next, when $W$ is a configuration of ${\mathcal H}_n$ ($0\leq n\leq k$) and $\Theta$ is a finite sequence of labmoves, we shall write $W\oplus\Theta$\label{xopl} to denote the configuration that results from $W$ by appending $\Theta$ to the (description of the) run-tape content component   of $W$.

In precise terms, this is how the {\bf procedure  $\simm_0((),B)$} works.
 It creates two  
integer-holding variables $b$ and $u$, with $b$ initialized to $1$ and $u$ to $0$. It further creates a variable $\vec{\nu}$ to hold move sequences, initialized to the empty sequence $\seq{}$. It further creates a configuration-holding variable $W$, initialized to the start configuration of ${\mathcal H}_0$ where the run tape  is empty (and, of course, so are the work tapes and the buffer). Finally, it creates  two signed-organ-holding variables $S$ and $R$, with $S$ having no initial value and $R$ initialized to $+O$, where $O$ is the first organ of $B$ (remember that $B$ is required to be nonempty).\footnote{The presence of the variable $S$ may seem redundant at this point, as $\simm_0((),B)$ (and likewise  $\simm_n(A,B)$ with $n\geq 1$) could be defined in a simpler way without it. The reason why we want to have $S$ will become clear in Section \ref{stf}. Similarly, in the present case we could have done without the variable $R$ as well --- it merely serves the purpose of ``synchronizing'' the cases of $n=0$ and $n\geq 1$.}   After this initialization step, the procedure goes into the following loop $\loopp_0$. Each iteration of the latter  simulates a certain number of steps of $H_0$ starting from a certain configuration (namely, the then-current value of $W$) in the scenario where ${\mathcal H}_0$'s imaginary adversary makes no moves other than those already present in configuration $W$ (i.e., already made by the time $W$ was reached).\vspace{10pt}  

{\bf Procedure $\loopp_0$}:  Let $+(\vec{\omega},p)$ be the value of $R$ ($R$ never takes negative values when $n=0$). Change the value of $W$ to $W\oplus\oo\vec{\omega}$. Then simulate/trace $p$ steps  of ${\mathcal H}_0$ starting from  configuration $W$.   
 While performing this simulation, keep track of the maximum number of non-blank cells on any (any one) of the work tapes of ${\mathcal H}_0$, and increment $u$ to that number every time the latter exceeds $u$. Also, every time ${\mathcal H}_0$ makes a move $\mu$, update $\vec{\nu}$ by adding $\mu$ at the end of it, and, additionally, update $W$ to the configuration in which such a move $\mu$ was made. Once the simulation of $p$ steps is complete, do the following. If $\vec{\nu}$  is empty, set the value of $S$ to $-(\vec{\nu},p)$ and return $(S,u)$. Suppose now  $\vec{\nu}$  is nonempty. In this case 
set the value of $S$ to $+(\vec{\nu},p)$. Then, if $b$ equals the size of $B$, return $(S,u)$. Otherwise, increment $b$ to $b+ 1$,  set $R$ to the $b$th organ of $B$ prefixed with ``$+$'', and repeat $\loopp_0$.\vspace{10pt}

Next, this is  how the {\bf procedure $\simm_n(A,B)$}\label{xsim} exactly works when  $n\geq 1$. It creates three integer-holding variables $a,b,u$, with $b$ initialized to $1$ and  $a,u$ to $0$.\footnote{Intuitively, $b$ keeps track of how many organs of $B$ have been fetched so far, and $a$ does the same for $A$. }  It further creates two move-sequence-holding variables $\vec{\psi}$ and $\vec{\nu}$, both initialized to the empty sequence $\seq{}$.  It further creates a configuration-holding variable $W$, initialized to the start configuration of ${\mathcal H}_n$ where the run tape  is empty. Finally, it creates  two signed-organ-holding variables $S$ and $R$, with $S$ having no initial value and $R$ initialized to $+O$, where $O$ is the first organ of $B$.
 After this initialization step, the procedure goes into the following loop $\loopp_n$.
As before, each iteration of the latter  simulates a certain number of steps of $H_n$ starting from a certain configuration (namely, $W$) in the scenario where the imaginary adversary makes no new moves.\vspace{10pt}  

{\bf Procedure $\loopp_n$}:  Let $+(\vec{\omega},p)$ (resp. $-(\vec{\omega},p)$) be the value of $R$. Change the value of $W$ to 
$W\oplus\oo 1.\vec{\omega}$ (resp. $W\oplus\oo 0.\vec{\omega}$).
Then simulate/trace $p$ steps  of ${\mathcal H}_n$ starting from  configuration $W$.   
 While performing this simulation, keep track of the maximum number of non-blank cells on any of the work tapes  of ${\mathcal H}_n$, and increment $u$ to that number every time the latter exceeds $u$. Also, every time ${\mathcal H}_n$ makes a move $\mu$ in the antecedent (resp. consequent) of the game, update $\vec{\psi}$ (resp. $\vec{\nu}$) by adding $\mu$ at the end of it, and, additionally, update $W$ to the configuration in which such a move $\mu$ was made. Once the simulation of $p$ steps is complete, do the following.
\begin{itemize}
  \item If $\vec{\nu}$ is nonempty, set the value of $S$ to $+(\vec{\nu},p)$. Then, if $b$ equals the size of $B$, return $(S,u)$; otherwise, increment $b$ to $b+ 1$, set $R$ to the $b$th organ of $B$ prefixed with ``$+$'', reset $\vec{\nu}$ to $\seq{}$, and repeat $\loopp_n$. 
  \item If $\vec{\nu}$ is empty, set the value of $S$ to $-(\vec{\psi},p)$. Then, if $a$ equals the size of $A$, return $(S,u)$. Otherwise, increment $a$ to $a+1$, set $R$ to the $a$th organ of $A$ prefixed with ``$-$'', reset $\vec{\psi}$ to $\seq{}$, and repeat $\loopp_n$.\vspace{10pt}
\end{itemize}

\noindent For a $\simm$-appropriate triple $(A,B,n)$, we shall write \[\simm_{n}^{\bullet}(A,B)\label{xsb}\] 
to refer to the signed organ $S$ such that $\simm_{n}(A,B)=(S,u)$ for some (whatever) $u$.

Later, we may write $\simm_{n}(A,B)$ to refer to either the {\em procedure} $\simm_{n}$ applied to arguments $A$ and $B$, or to the {\em output} $(S,u)$ of that procedure on the same arguments. It will be usually clear from the context which of these two is meant.
The same applies to  $\simm_{n}^{\bullet}(A,B)$   which, seen as a procedure,  runs exactly like $\simm_n(A,B)$, and only differs from the latter in that it just outputs $S$ rather than $(S,u)$. 

Consider any two bodies  $B=(O_1,\ldots,O_t)$ and $B'=(O'_1,\ldots,O'_{t'})$. We say that $B'$ is an {\bf extension}\label{xbe} of $B$, and that $B$ is a {\bf restriction}\label{xbr} of $B'$, iff  $t\leq t'$ and $O_1=O'_1,\ldots,O_t=O'_t$. As expected, ``{\bf proper extension}'' means ``extension but not restriction''. Similarly for ``{\bf proper restriction}''. 
 
\begin{lem}\label{golemma}
%\marginpar{golemma}
Consider any $\simm$-appropriate triple $(A,B,n)$. 
\begin{enumerate}[label=\arabic*.]
\item If $\simm_{n}^{\bullet}(A,B)$ is negative, then, for every extension $B'$ of $B$, $\simm_n(A,B')=\simm_n(A,B)$.

\item If $\simm_{n}^{\bullet}(A,B)$ is positive and $n\not=0$, then, for every extension $A'$ of $A$,  $\simm_n(A',B)=\simm_n(A,B)$.

\item Whenever  $\simm_{n}^{\bullet}(A,B)$ is positive, the size of $B$ does not exceed $\mathfrak{e}_\top$.
\end{enumerate}
\end{lem}

\begin{proof} Clauses 1-2 can be verified through a straightforward analysis of the work of $\simm_n$. For clause 3, assume $\simm_n(A,B)=+(\vec{\omega},p)$, and let $s$ be the size of $B$. Observe that, in the process of computing $\simm_n(A,B)$, the payloads of all positive values that the variable $S$ ever takes, including its last value $+(\vec{\omega},p)$, are nonempty. All such payloads consist of moves made by ${\mathcal H}_n$ in the consequent of $F'(n-1)\mli F'(n)$.  From the work of $\simm_n$ we can see that altogether there are $s$ positive values taken by $S$. Now, remembering our assumption that ${\mathcal H}_n$ plays quasilegally, implying that it does not make more than $\mathfrak{e}_\top$ moves in the consequent of $F'(n-1)\mli F'(n)$, it is clear that $s$ cannot exceed $\mathfrak{e}_\top$. \end{proof}

By a {\bf saturated}\label{xsaturated} triple we shall mean a $\simm$-appropriate triple   $(A,B,n)$ such that:   
%\begin{quote} \ \vspace{-8pt}
\begin{enumerate}[label=\arabic*.]
\item If $\simm_{n}^{\bullet}(A,B)$ is negative, then, for every nonempty proper restriction  $B'$ of $B$,  $\simm_{n}^{\bullet}(A,B')$ is positive. 

\item If $\simm_{n}^{\bullet}(A,B)$ is positive, then, for every  proper restriction  $A'$ of $A$,  $\simm_{n}^{\bullet}(A',B)$ is negative.%\vspace{-8pt}
\end{enumerate}\bigskip %%%JK
%\ \end{quote}
 
\noindent For a body $B=(O_1,\ldots,O_s)$, we will write 
%%%JK \[\mbox{\em $B^{\odd}$\label{xbodd} (resp. $B^{\even}$)}\] 
$B^{\odd}$\label{xbodd} (resp. $B^{\even}$)
to denote the body $(O_1,O_3,\ldots)$ (resp. $(O_2,O_4,\ldots)$) obtained from $B$ by deleting each $O_i$ with an even (resp. odd) $i$.  

\begin{defi}\label{ap20a}
%\marginpar{ap20a}
Consider a saturated triple $(A,B,n)$. Let $A=(A_1,\ldots,A_a)$ and $B=(B_1,\ldots,B_b)$. Further let   $-P_1,\ldots,-P_p $ be the (sequence of the) negative values that the variable $S$ of the procedure $\simm_n$ goes through when computing $\simm_{n}(A,B)$, and let $+Q_1,\ldots,+Q_q $ be the (sequence of the) positive values that  $S$ goes through. Observe that $a\leq p\leq a+1$ and $q\leq b\leq q+1$. 
\begin{enumerate}[label=\arabic*.]
\item We define $\simm_{n}^{\leftarrow}(A,B)$\label{xslar} as the body $(P_1,A_1,P_2,A_2,\ldots)$ --- that is, the (unique) body $C$ such that $C^{\odd}=(P_1,\ldots,P_p)$ and $C^{\even}=(A_1,\ldots,A_a)$.

\item We define $\simm_{n}^{\rightarrow}(A,B)$\label{xsrar} as the
  body $(B_1,Q_1,B_2,Q_2,\ldots)$ --- that is, the (unique) body $C$
  such that $C^{\odd}=(B_1,\ldots,B_b)$ and
  $C^{\even}=(Q_1,\ldots,Q_q)$.
\end{enumerate}
\end{defi}\bigskip%%%JK

\noindent Let $B=\bigl((\vec{\alpha}_1,p_1),\ldots,(\vec{\alpha}_s,p_s)\bigr)$ be a body. 
  We define  $\overline{B}\label{xoverl}$ %%%JK
 as   the run $\seq{\oo\vec{\alpha}_1,\pp\vec{\alpha}_2,\ldots}$ obtained from $\seq{\vec{\alpha}_1,\ldots,\vec{\alpha}_s}$ by replacing each $\vec{\alpha}_i$ ($1\leq i\leq s$) with $\oo\vec{\alpha}_i$ if $i$ is odd, and with $\pp\vec{\alpha}_i$ if $i$ is even.

Some more notation and terminology. When $\Gamma$ and $\Delta$ are runs, we write $\Gamma\preceq \Delta$\label{xx85} to mean that $\Gamma$ is a (not necessarily proper) initial segment of $\Delta$. Next, as always in CoL, $\gneg \Gamma$\label{xnki} means the result of changing in $\Gamma$ each label $\pp$ to $\oo$ and vice versa.  $\Gamma^{0.}$\label{xgnol}  means the result of deleting from $\Gamma$ all moves (together with their labels, of course) except those of the form $0.\alpha$, and then further deleting the prefix ``$0.$''  in the remaining moves. Similarly for $\Gamma^{1.}$. Intuitively, when $\Gamma$ is a play of a parallel disjunction $G_0\mld G_1$ or conjunction $G_0\mlc G_1$ of games,  $\Gamma^{0.}$ (resp. $\Gamma^{1.}$) is the play that has taken place --- according to the scenario of $\Gamma$ --- in the $G_0$ (resp. $G_1$) component.

\begin{lem}\label{ap20b} Consider any saturated $\simm$-appropriate triple $(A,B,n)$. Let $\simm_{n} (A,B )= \bigl(\pm(\vec{\omega},v),u\bigr)$, where $\pm\in\{+,-\}$.
%\marginpar{ap20b}
\begin{enumerate}[label=\arabic*.]
\item   The case of $n=0$ (and hence $A=()$):
\begin{enumerate} 
\item There is a run $\Upsilon$ generated by ${\mathcal H}_0$ such that  $\overline{\simm_{0}^{\rightarrow }  ((),B )}\preceq \Upsilon$. 

\item Furthermore, if $\overline{\simm_{0}^{\rightarrow }  ((),B )}$ is a reasonable run of $F'(0)$ and    $v\geq \mathfrak{L}(\mathfrak{l},u)$, then, for such an $\Upsilon$, we simply have $\overline{\simm_{0}^{\rightarrow }  ((),B )}= \Upsilon$.
\end{enumerate}

\item The case of $1\leq n\leq k$:
\begin{enumerate}
\item There is a run $\Upsilon$ generated by ${\mathcal H}_n$ such that  $\overline{\simm_{n}^{\rightarrow }  (A,B )}\preceq \Upsilon^{1.}$ and $\gneg\overline{\simm_{n}^{ \leftarrow}  (A,B )}\preceq \Upsilon^{0.}$. 

\item Furthermore, if $\simm_{n}^{\bullet}  (A,B )$ is negative, $\overline{\simm_{n}^{\rightarrow }  (A,B )}$     is a reasonable run of $F'(n)$, $\overline{ \simm_{n}^{ \leftarrow}  (A,B )}$   is a reasonable run of $F'(n-1)$ and    $v\geq \mathfrak{L}(\mathfrak{l},u)$, then, for such an $\Upsilon$, we simply have $\overline{\simm_{n}^{\rightarrow }  (A,B )}= \Upsilon^{1.}$ and $\gneg\overline{\simm_{n}^{\leftarrow }  (A,B )}=\Upsilon^{0.}$.
\end{enumerate}
\end{enumerate}
\end{lem}

\begin{proof} Assume the conditions of the lemma. Let $A=\bigl((\vec{\alpha}_1,p_1),\ldots,(\vec{\alpha}_a,p_a)\bigr)$ and $B=\bigl((\vec{\beta}_1,q_1),\ldots,(\vec{\beta}_b,q_b)\bigr)$. Further let   $-(\vec{\gamma}_1,r_1)$, \ldots, $-(\vec{\gamma}_c,r_c)$ be the  negative values that the variable $S$ of the procedure $\simm_n$ goes through when computing $\simm_{n}(A,B)$, and let 
$+(\vec{\delta}_1,s_1)$, \ldots, $+(\vec{\delta}_d,s_d)$ be the  positive values that  $S$ goes through.

\begin{enumerate}[label=\arabic*.]
\item Assume $n=0$, and thus $A=()$, i.e.,  $a=0$. 
 Analyzing the definitions of $\simm_0$ and $\simm_{0}^{\rightarrow}$ and   taking into account that $\bigl((),B,0\bigr)$ is saturated, we see that, what the procedure $\simm_0\bigl((),B\bigr)$
does is that it simulates the first $t$  steps of a certain computation branch $C$ of ${\mathcal H}_0$ for a certain  
$t$ with $v=q_b\leq t\leq q_1+\ldots +q_b$, and the position spelled on ${\mathcal H}_0$'s imaginary run tape by the end of this episode (without counting the initial moves $\oo \vec{c}$ --- see Remark \ref{apology}) is nothing but $\overline{\simm_{0}^{\rightarrow }  \bigl((),B \bigr)}$. Let $\Upsilon$ be the run spelled by $C$. Then $\Upsilon$ satisfies the promise of clause 1(a) of the lemma. For clause 1(b), additionally assume that $\overline{\simm_{0}^{\rightarrow }  \bigl((),B \bigr)}$ is a reasonable run of $F'(0)$ and    $v\geq \mathfrak{L}(\mathfrak{l},u)$.  We may assume that, in the above branch $C$, ${\mathcal H}_0$'s adversary makes no moves after (beginning from) time $t-v$. Then, by Lemma \ref{m29a}, ${\mathcal H}_0$ makes no moves after (beginning from) time $t$. Thus, the run $\Upsilon$ contains no labmoves in addition to those that are in $\simm_{n}^{\rightarrow }  (A,B )$, meaning that 
 $\simm_{n}^{\rightarrow }  (A,B )= \Upsilon$, as desired.\vspace{4pt}

\item Assume $1\leq n\leq k$. Again, taking into account that $(A,B,n)$ is saturated, we can see that, what the procedure $\simm_n\bigl(A,B\bigr)$ 
does is that it simulates the first $t$  steps of a certain computation branch $C$ of  ${\mathcal H}_n$ for a certain number $t$ with $v\leq t\leq p_1+\ldots+p_a+q_1+\ldots +q_b$. Note that here $v$ is either $p_a$ or $q_b$. Let $\Phi$ be the position spelled on ${\mathcal H}_n$'s imaginary run tape by the end of this episode. It is not hard to see that $\Phi^{1.}=\overline{\simm_{n}^{\rightarrow }  \bigl(A,B \bigr)}$. Further, if  $\simm_{n}^{\bullet}  (A,B )$ is negative, then we also have $\Phi^{0.}=\gneg\overline{\simm_{n}^{\leftarrow }  \bigl(A,B \bigr)}$. Otherwise, if $\simm_{n}^{\bullet}  (A,B )$ is positive, $\Phi^{0.}$   is a (not necessarily proper) extension of $\gneg\overline{\simm_{n}^{\leftarrow }  \bigl(A,B \bigr)}$ through some $\pp$-labeled moves. Let $\Upsilon$ be the run spelled by $C$. Then, in view of the observations that we have just made, $\Upsilon$ satisfies the promise of clause 2(a) of the lemma.

For clause 2(b), additionally assume that 
$\simm_{n}^{\bullet}  (A,B )$ is negative, $\overline{\simm_{n}^{\rightarrow }  (A,B )}$     
is a reasonable run of $F'(n)$, $\gneg \overline{\simm_{n}^{ \leftarrow}  (A,B )}$   is a reasonable run of 
$F'(n-1)$, and    $v\geq \mathfrak{L}(\mathfrak{l},u)$.  As observed in the preceding paragraph, on our 
present assumption of $\simm_{n}^{\bullet}  (A,B )$'s being negative, we have 
$\Phi^{0.}=\gneg\overline{\simm_{n}^{\leftarrow }  \bigl(A,B \bigr)}$ and $\Phi^{1.}=\overline{\simm_{n}^{\rightarrow }  \bigl(A,B \bigr)}$.    
We may assume that, in the above branch $C$, ${\mathcal H}_n$'s adversary makes no moves after (beginning from) time 
$t-v$. Then, by Lemma \ref{m29a}, ${\mathcal H}_n$ makes no 
moves after (beginning from) time $t$. Thus, the run $\Upsilon$ contains 
no labmoves in addition to those that are (after removing the prefixes ``$0.$'' and ``$1.$'') 
in $\gneg\overline{\simm_{n}^{\leftarrow }  (A,B )}$ and $\overline{\simm_{n}^{\rightarrow }  (A,B )}$, meaning that 
 $\gneg\overline{\simm_{n}^{\leftarrow }  (A,B )}= \Upsilon^{0.}$ and
 $\overline{\simm_{n}^{\rightarrow }  (A,B )}= \Upsilon^{1.}$, as
 desired.\qedhere
\end{enumerate}
\end{proof}

\subsection{Aggregations}\label{saggr}
%\marginpar{saggr} 
By an {\bf entry}\label{xentry} we shall mean a pair $E=[n,B]$, where $n$, called the {\bf index}\label{xindex} of $E$, is an element of $\{0,\ldots,k\}$, and $B$, called the {\bf body of $E$},\label{xbodyof} is a body. The {\bf size} of an entry $E$ should be understood as the size of its body. By saying that an entry is {\bf $n$-indexed} we shall mean that $n$ is the index of that entry.

An {\bf aggregation} is a nonempty finite sequence  
$\vec{E}$ of entries such that:
\begin{enumerate}[label=\bf(\roman*)]
  \item The last entry of $\vec{E}$ is $k$-indexed, and its body is odd-size. We call it the 
{\bf master entry}\label{xheadentry} of $\vec{E}$, and call all other entries  (if there are any) 
{\bf common entries}.\label{xtailentry} 
  \item  The indices of the entries of $\vec{E}$ are strictly increasing. That is, the index of any given entry is strictly smaller than the index of any entries to the right of it.
\item Each even-size entry (if there are such entries) is to the left of each odd-size entry. 
\item The sizes of the even-size entries are strictly decreasing. That is, the size of any even-size entry is strictly smaller than the size of any (even-size) entry to the left of it.
\item The sizes of the odd-size common entries are strictly increasing. That is, the size of any odd-size common entry is strictly smaller than the size of any (odd-size) common entry to the right of it.
\item There are no entries of size $0$.
\end{enumerate}

\noindent The {\bf central triple}\label{xcentral} of an  aggregation $\vec{E}$ is $(L,R,n)$, where: 
\begin{enumerate}
  \item $n$ is the index of the leftmost odd-size entry of $\vec{E}$. 
  \item $R$ is the body of the above $n$-indexed entry of $\vec{E}$.
  \item If $\vec{E}$ does not have an entry whose index is $n-1$,\footnote{This condition is always automatically satisfied when $n=0$.}   then $L$ is the empty body $()$. Otherwise, $L$ is the body of the $(n-1)$-indexed  entry of $\vec{E}$.  
\end{enumerate}\medskip%%%JK

\noindent Consider any aggregation $\vec{E}$. The {\bf master body} of $\vec{E}$ is the body of the master entry of $E$; the {\bf master organ}\label{xmaster} of $\vec{E}$ is the last organ of the master body of $\vec{E}$; and the {\bf master payload} (resp. {\bf master scale})  of $\vec{E}$ is the payload (resp. scale) of the  master organ of $\vec{E}$.

\subsection{The procedure {\sc Main}}\label{sn}
%\marginpar{sn}

We are now ready to finalize our description of the work of ${\mathcal M}_k$. This is a machine that creates an aggregation-holding variable $\vec{E}$ and an integer-holding variable $U$, with $\vec{E}$ initialized to the aggregation $\seq{[k,((\seq{},1))]}$\footnote{I.e., the single-entry aggregation where the master body is of size $1$, the master payload is empty and the master scale is $1$.}  and $U$ initialized to $0$. After this initialization step, ${\mathcal M}_k$ goes into the below-described loop \mainn. As already noted, our description of $\mathcal M$ and hence of \mainn\ and our subsequent analysis of its work relies on the  Clean Environment Assumption.\medskip

{\bf Terminology}: In our description of \mainn, whenever we say \repeatt,\label{xrepeat} it is to be understood as repeating (going to) \mainn\ without changing the values of $U$ and $\vec{E}$.\label{xeu} On the other hand, whenever we say \restart,\label{xrestart} it is to be understood as resetting $U$ to $0$, modifying $\vec{E}$ by deleting all common entries in it (but leaving the master entry unchanged), and then repeating \mainn. Finally, when we say ``{\em Environment has made a new move}'', we mean that the run tape of ${\mathcal M}_k$ contains a $(q+1)$th $\oo$-labeled move 
(which we refer to as ``{\em the new move}''), where $q$ is the total number of moves in (all moves in the payloads  of the organs of) $B^{\odd}$, where $B$ is the master body  of $\vec{E}$.\medskip

{\bf Procedure} \mainn.\label{xmain} Let $(L,R,n)$ be the central triple of $\vec{E}$. Start running the procedure $\simm_n$ on $(L^{\even},R^{\odd})$ while, in parallel, at some constant rate,  polling the run tape to see if Environment has made a new move.\footnote{Clarifying: the polling routine is called, say, after every $1000$ steps of performing $\simm_n$; such a call --- which, itself, may take more than a constant amount of time --- interrupts $\simm_n$, saves its state, checks the run tape to see if a new move is made and, if not, returns control back to the caller.}    
Then act depending on which of the following two cases is ``the case'': 
\begin{description}
%{\em Case 1}: 
\item[Case 1] Before  $\simm_n$ terminates, one of the calls of the polling routine detects a new move $1.\theta$ (i.e., the move $\theta$ in the consequent of $k\leq \mathfrak{b}|\vec{d}|\mli F'(k)$) by Environment.  
Let $\theta'$ be the $F'(k)$-prudentization 
of $\theta$.    
 Modify $\vec{E}$ by  adding $\theta'$ to its master payload, and resetting the master scale to $1$. Then \restart.

%{\em Case 2}:
\item[Case 2] $\simm_n$ terminates without any of the calls of the polling routine meanwhile detecting a new move by Environment.
Let  $(S,u)$ be the value computed/returned by $\simm_n(L^{\even},R^{\odd})$. Update $U$ to $\max(u,U)$.  Then 
act depending on whether $S$ is positive or negative.

%{\em Subcase 2.1}: 
\item[Subcase 2.1] $S$ is positive, namely, $S=+(\vec{\omega},s)$. Let $B$ be the body of the $n$-indexed entry of $\vec{E}$.
Act depending on whether $n<k$ or not. 

%{\em Subsubcase 2.1.1}: 
\item[Subsubcase 2.1.1] $n<k$.  Update $\vec{E}$ by adding $(\vec{\omega},s)$ as a new organ  to $B$. Further modify $\vec{E}$ by deleting all $(<n)$-indexed entries whose size does not exceed  that of the $n$-indexed entry, if such entries exist.   Then \repeatt.

%{\em Subsubcase 2.1.2}: 
\item[Subsubcase 2.1.2] $n=k$.  Update $\vec{E}$ by adding $(\vec{\omega},s)$ and $(\seq{},s)$ as two new organs  to $B$.  Then make the moves $\vec{\omega}$ in the consequent of (the real play of) $k\leq \mathfrak{b}|\vec{d}|\mli F'(k) $. Finally, \repeatt. 

%{\em Subcase 2.2}: 
\item[Subcase 2.2] $S$ is negative, namely, $S=-(\vec{\omega},s)$. Act depending on whether $n>0$ or not. 

%{\em Subsubcase 2.2.1}: 
\item[Subsubcase 2.2.1] $n>0$. Then,
if $\vec{E}$ has an $(n-1)$-indexed entry $E$,  modify $\vec{E}$ by adding $(\vec{\omega},s)$ as a new  organ to the body of $E$; otherwise 
  modify $\vec{E}$ by inserting into it the entry $E=[n-1,((\vec{\omega},s))]$ immediately on the left of the $n$-indexed entry. In either case, further modify $\vec{E}$ by deleting all $\geq n$-indexed common entries whose size does not exceed  that of the $(n-1)$-indexed entry, if such entries exist.   
After that  \repeatt.

%{\em Subsubcase 2.2.2}: 
\item[Subsubcase 2.2.2] $n=0$. Let $v$ be the master scale of $\vec{E}$. Act depending on whether $v< \mathfrak{L}(\mathfrak{l},U)$ or not.\footnote{For $\mathfrak{L}$, remember clause 4 of Notation \ref{not1}.}

%{\em Subsubsubcase 2.2.2.1}: 
\item[Subsubsubcase 2.2.2.1] $v<\mathfrak{L}(\mathfrak{l},U)$. Then modify $\vec{E}$ by doubling its master scale $v$,  and \restart. 

%{\em Subsubsubcase 2.2.2.2}:
\item[Subsubsubcase 2.2.2.2] $v\geq \mathfrak{L}(\mathfrak{l},U)$.  Keep polling the run tape of ${\mathcal M}_k$ to see if Environment has made a new move $1.\theta$. If and when such a move is detected, 
modify $\vec{E}$ by  adding the $F'(k)$-prudentization $\theta'$ of $\theta$ to the master payload of $\vec{E}$,  and resetting the master scale to $1$. Then \restart.
\end{description}

\subsection{\texorpdfstring{${\mathcal M}$}{M} is a solution of the target game}\label{stf}
%\marginpar{stf}
In this subsection we want to verify that ${\mathcal M}_k$ indeed wins $k\leq \mathfrak{b}|\vec{d}|\mli F'(k)$ and hence ${\mathcal M}$ wins $x\leq \mathfrak{b}|\vec{s}|\mli F(x,\vec{v})$. 
For this purpose, when analyzing the work and behavior of ${\mathcal M}_k$, we will implicitly have some arbitrary but fixed computation branch (``play'') of ${\mathcal M}_k$ in mind. So, for instance, when we say ``the $i$th iteration of \mainn'', it should be understood in the context of that branch.

\begin{nota}\label{not3}
%\marginpar{not3}
{\em In what follows, $\mathbb{I}$\label{xiii} will stand for the set of positive integers $i$ such that  $\mainn$ is iterated at least $i$ times. 
Also, for each $i\in \mathbb{I}$, $\vec{E}_{i}$\label{xnot3} will stand for the value of the aggregation/variable $\vec{E}$ at the beginning  of the $i$th iteration of $\mainn$. 
}\end{nota}

\begin{lem}\label{bei}
%\marginpar{bei}
For any $i\in \mathbb{I}$ and any entry $E$ of $\vec{E}_i$, the size of $E$ 
 does not exceed $2\mathfrak{e}_\top + 1$.
\end{lem}

\begin{proof} For a contradiction, assume $i\in \mathbb{I}$, and  $\vec{E}_i$ has an entry of size greater than $2\mathfrak{e}_\top + 1$. Let $n$ be the index of such an entry. 

First, consider the case $n<k$. 
Let $j\leq i$ be the smallest number in $\mathbb{I}$ such that $\vec{E}_j$ has 
 a $(2\mathfrak{e}_\top +2)$-size, $n$-indexed entry $[n,(O_1,\ldots,O_{2\mathfrak{e}_\top +2})]$ --- it is not hard to see that such a $j$ exists, and $j> 1$ because $\vec{E}_1$ has no common entries.  The only way the above entry could have emerged in $\vec{E}_j$ is that $\vec{E}_{j-1}$ contained the entry $[n,(O_1,\ldots,O_{2\mathfrak{e}_\top +1})]$, and its body ``grew'' into $(O_1,\ldots,O_{2\mathfrak{e}_\top +2})$ on the transition from $\vec{E}_{j-1}$ to $\vec{E}_j$ according to the prescriptions of Subsubcase 2.1.1 of the description of \mainn. 
This in turn means that the central triple of $\vec{E}_{j-1}$ was $(A,(O_1,\ldots,O_{2\mathfrak{e}_\top +1}),n)$ for a certain body $A$, and $\simm_{n}^{\bullet}(A^{\even},(O_1,\ldots,O_{2\mathfrak{e}_\top +1})^{\odd} )=+O_{2\mathfrak{e}_\top +2}$. This, however, is impossible by clause 3 of Lemma \ref{golemma}, because the size of  $(O_1,\ldots,O_{2\mathfrak{e}_\top +1})^{\odd}$ is $\mathfrak{e}_\top +1$, exceeding $\mathfrak{e}_\top$. 

The case $n=k$ is similar, only with ``$k$'' instead of ``$n$'', and ``$2\mathfrak{e}_\top +3$'' instead of ``$2\mathfrak{e}_\top +2$''. 
\end{proof}

\begin{lem}\label{beijing}
%\marginpar{beijing}
There is a bound $\mathfrak{z}(w)\in{\mathcal R}\tim$ such that the cardinality of $\mathbb{I}$ does not exceed 
$\mathfrak{z}(\mathfrak{l})$. 
\end{lem}

\begin{proof} 
In this proof we will be using $\mathfrak{d}$ as an abbreviation of $2\mathfrak{e}_\top +1$. Whenever we say ``$\vec{E}$ always (never, etc.) so and so'', it is to be understood as that, throughout the work of \mainn, the value of the variable $\vec{E}$ always (never, etc.) so and so. Similarly for $U$. ``Case'', ``Subcase'', etc. mean those of the description of \mainn.

According to Lemma \ref{bei}, we have:
%\marginpar{park1}
\begin{equation}\label{park1}
\mbox{\it The size of no entry of $\vec{E}$ ever exceeds $\mathfrak{d}$.}
\end{equation}
  
Our next claim is the following: 
%\marginpar{park2}
\begin{equation}\label{park2}
\begin{array}{l}
\mbox{\it The number of moves in the payload of no organ}\\
\mbox{\it  of the master body of $\vec{E}$ ever exceeds }\max(\mathfrak{e}_\top,\mathfrak{e}_\bot).
\end{array}
\end{equation} 
Indeed, let $(O_1,\ldots,O_a)$ be the master body of $\vec{E}$ at a given stage of the work of \mainn, and consider any organ $O_i=(\vec{\alpha},s)$ ($1\leq i\leq a$) of this body. From an analysis of the work of \mainn\ we can see that, if $i$ is odd, then $\vec{\alpha}$ are moves made by Environment within the $F'(k)$ component in the real play. Therefore, in view of the Clean  Environment Assumption, the number of such moves is at most $\mathfrak{e}_\bot$.  If $i$ is even, then $\vec{\alpha}$ are moves made by ${\mathcal H}_k$ in a certain play simulated through $\simm_k$. As in the preceding case, the number of such moves cannot exceed $\mathfrak{e}_\top$  because, as we have agreed, ${\mathcal H}_k$ plays quasilegally.

Taking into account that each ${\mathcal H}_n$ ($\mathcal N$ and $\mathcal K$, that is) plays unconditionally prudently and that Environment's moves in $F'(k)$ are also prudentized when copied by \mainn\ according to the prescriptions of Case 1 or Subsubsubcase 2.2.2.2  (and 
that every move that emerges in $\vec{E}$ originates either from Environment or from one of ${\mathcal H}_i$), one can see that the run tape of any 
simulated machine does not contain moves whose magnitude is greater than $\mathfrak{G}(\mathfrak{l})$ where, as we remember,  $\mathfrak{G}$ is the superaggregate bound of  $F(x,\vec{v})$. 
Since the ${\mathcal H}_n$s ($\mathcal N$ and $\mathcal K$, to be more precise) play in unconditional space $\mathfrak{s}$, we then find that  the value of the variable $U$ of \mainn\  never exceeds  $\mathfrak{s}\bigl(\mathfrak{G}(\mathfrak{l})\bigr)$.  
Thus, the maximum value of $\mathfrak{L}(\mathfrak{l},U)$ is bounded by 
$\mathfrak{L}\bigl(\mathfrak{l},\mathfrak{s}(\mathfrak{G}(\mathfrak{l}))\bigr)$. The master scale $v$ of  $\vec{E}$ increases --- namely, doubles --- only according to the prescriptions of Subsubsubcase 2.2.2.1, and such an increase happens only when $v$  is smaller than $\mathfrak{L}(\mathfrak{l},U)$. For this reason, we have:
%\marginpar{park3}
\begin{equation}\label{park3}
\mbox{\it The master scale of $\vec{E}$ is always smaller than $2\mathfrak{L}\bigl(\mathfrak{l},\mathfrak{s}(\mathfrak{G}(\mathfrak{l}))\bigr)$.}
\end{equation} 

Let $f$ be the unarification of the bound $\mathfrak{b}\in{\mathcal R}\tim$ from (\ref{r22}). Note that, since $k\leq \mathfrak{b}|\vec{d}|$, we have $k\leq f(\mathfrak{l})$. 

Let  $\mathfrak{K}(w)$ be the unary function defined by
%\marginpar{ka}
\begin{equation}\label{ka}
\mathfrak{K}(w)\ =\ \max\Bigl(|\mathfrak{L}\bigl(w,\mathfrak{s}(\mathfrak{G}(w))\bigr)|,\ f(w), \ \mathfrak{d},\ \mathfrak{e}_{\bot} \Bigr)+1,
\end{equation}  
and let $\mathfrak{k}$ be an abbreviation of $\mathfrak{K}(\mathfrak{l})$. 

With each element $i$ of $\mathbb{I}$ we now associate an integer $\mbox{\it Rank}(i)$ defined as follows:

\[\mbox{\it Rank}(i)=c_0\times\mathfrak{k}^0+ c_1\times\mathfrak{k}^1 + c_2\times\mathfrak{k}^2+\ldots +c_{\mathfrak{d}}\times\mathfrak{k}^{\mathfrak{d}}+
c_{\mathfrak{d}+1}\times\mathfrak{k}^{\mathfrak{d}+1}+ c_{\mathfrak{d}+2}\times\mathfrak{k}^{\mathfrak{d}+2}+ c_{\mathfrak{d}+3}\times\mathfrak{k}^{\mathfrak{d}+3},\]
where:

\begin{itemize}
  \item $c_0=0$. Take a note of the fact  that $c_0<\mathfrak{k}$. 
  \item For each even $j\in\{1,\ldots,\mathfrak{d}\}$: If $\vec{E}_i$ contains a common entry of size $j$, then $c_j$ is $n+1$, where $n$ is the index of that entry; otherwise $c_j=0$. Thus, $c_j$ cannot exceed $k$ and, since $k\leq f(\mathfrak{l})$,  from (\ref{ka}) we can see that $c_j<\mathfrak{k}$. 
  \item For each odd $j\in\{1,\ldots,\mathfrak{d}\}$: If $\vec{E}_i$ contains a common entry of size $j$, then $c_j$ is $k-n$, where $n$ is the index of that entry; otherwise $c_j=0$. Again, we have $c_j<\mathfrak{k}$. 
  \item $c_{\mathfrak{d}+1}$ is $|v|$, where $v$ is the master scale of  $\vec{E}_i$. In view of (\ref{park3}), we find $c_{\mathfrak{d}+1}<\mathfrak{k}$.
  \item $c_{\mathfrak{d}+2}$ is the number of moves in the master payload of  $\vec{E}_i$. From (\ref{park2}), we see that $c_{\mathfrak{d}+2}<\mathfrak{k}$.
  \item $c_{\mathfrak{d}+3}$ is the size of the master body of $\vec{E}_i$. The fact (\ref{park1}) guarantees that $c_{\mathfrak{d}+3}<\mathfrak{k}$.
\end{itemize}

\noindent As we have observed in each case above, all of the factors $c_0,c_1,\ldots,c_{\mathfrak{d}+3}$ from $\mbox{\it Rank}(i)$ are smaller than $\mathfrak{k}$. This allows us to think of   $\mbox{\it Rank}(i)$ as a $\mathfrak{k}$-ary numeral of length $\mathfrak{d}+4$, with the least significant digit being $c_{0}$ and the most significant digit being $c_{\mathfrak{d}+3}$. 

With some analysis of the work of \mainn, which we here leave to the reader, one can see that 

%\marginpar{ff2}
\begin{equation}\label{ff2}
\mbox{\it For any $i$ with $(i+ 1)\in \mathbb{I}$, \ $\mbox{\it Rank}(i)<\mbox{\it Rank}(i+1)$.}
\end{equation}
But, by our observation $c_0,c_1,\ldots,c_{\mathfrak{d}+3}<\mathfrak{k}$,  no rank can exceed the (generously taken) number 
\[ (\mathfrak{k}-1)\times\mathfrak{k}^0+ (\mathfrak{k}-1)\times\mathfrak{k}^1 + (\mathfrak{k}-1)\times\mathfrak{k}^2+\ldots +  (\mathfrak{k}-1)\times\mathfrak{k}^{\mathfrak{d}+3},\]
i.e., the number $\mathfrak{M}(\mathfrak{l})$, where $\mathfrak{M}(w)$ is the unary function 
 \[(\mathfrak{K}(w)-1)\times(\mathfrak{K}(w))^0+ (\mathfrak{K}(w)-1)\times(\mathfrak{K}(w))^1 + (\mathfrak{K}(w)-1)\times(
\mathfrak{K}(w))^2+\ldots +  (\mathfrak{K}(w)-1)\times
(\mathfrak{K}(w))^{\mathfrak{d}+3}.\]
 Thus:
%\marginpar{iin}
\begin{equation}\label{iin}
 \mbox{\it For any $i\in \mathbb{I}$, \ $\mbox{\it Rank}(i)\leq \mathfrak{M}(\mathfrak{l})$.}   
\end{equation}
By the conditions of ${\mathcal R}$-Induction, $F(x,\vec{v})$ is ${\mathcal R}\spa$-bounded. Hence, by Lemma \ref{lagg},  $\mathfrak{G}(w)\preceq{\mathcal R}\spa$.  This, by condition 4 of Definition 2.2 of \cite{AAAI},  implies  $\mathfrak{s}(\mathfrak{G}(w))\preceq {\mathcal R}\spa$. The two facts $\mathfrak{G}(w)\preceq{\mathcal R}\spa$ and $\mathfrak{s}(\mathfrak{G}(w))\preceq {\mathcal R}\spa$, 
 by condition 5 of Definition 2.2 of \cite{AAAI}, further yield   $\mathfrak{G}(w)\preceq{\mathcal R}\tim$ and
$\mathfrak{s}(\mathfrak{G}(w))\preceq {\mathcal R}\tim$.
Looking back at our definition of $\mathfrak{L}$ in Notation \ref{not1}(4), we see that 
%\marginpar{2244}
\begin{equation}\label{2244}
|\mathfrak{L}(w,u)|=O(|w|+|\mathfrak{G}(w)|+u)
\end{equation}
 and thus $|\mathfrak{L}(w,\mathfrak{s} (\mathfrak{G}(w) ))|=O(|w|+|\mathfrak{G}(w)|+\mathfrak{s} (\mathfrak{G}(w) ))$. This, together with  $\mathfrak{G}(w)\preceq{\mathcal R}\tim$ and $\mathfrak{s}(\mathfrak{G}(w))\preceq {\mathcal R}\tim$, by the linear closure of ${\mathcal R}\tim$, 
implies 
%\marginpar{sevv1}
\begin{equation}\label{sevv1}
|\mathfrak{L}(w,\mathfrak{s} (\mathfrak{G}(w) ))|\preceq {\mathcal R}\tim.
\end{equation}
Since $f$ is the unarification of $\mathfrak{b}\in{\mathcal R}\tim$, we obviously have $f\preceq {\mathcal R}\tim$. This, together with 
(\ref{sevv1}), (\ref{ka}) and the fact of ${\mathcal R}\tim$'s being linearly closed, implies that $\mathfrak{K}\preceq {\mathcal R}\tim$. 
The latter, in turn, in view of ${\mathcal R}\tim$'s being polynomially closed, implies that $\mathfrak{M}\preceq {\mathcal R}\tim$. 
So,  there is a bound $\mathfrak{z}(w)$ in 
${\mathcal R}\tim$ with $\mathfrak{M}(w)\preceq \mathfrak{z}(w)$ and hence $\mathfrak{M}(\mathfrak{l})\leq \mathfrak{z}(\mathfrak{l})$. In view of (\ref{iin}),   no rank can thus ever exceed 
$\mathfrak{z}(\mathfrak{l})$. 
 But, by (\ref{ff2}), different elements of $\mathbb{I}$ have different ranks. Hence, the cardinality of $\mathbb{I}$ does not exceed $\mathfrak{z}(\mathfrak{l})$ either, as desired. 
 \end{proof}
 
For a number $h\in \mathbb{I}$, we define the set $\mathbb{I}^h$\label{xih} by 
\[\mathbb{I}^h\ =\ \{i\ |\ i\in \mathbb{I} \ \mbox{\it and } i\leq h\}=\{1,\ldots,h\}.
\] 

We say that a given iteration of \mainn\ is {\bf restarting}\label{xring} (resp. {\bf repeating}) iff it terminates and calls the next iteration through \restart\ (resp. \repeatt). The repeating iterations are exactly those that proceed according to Subcase 2.1 or Subsubcase 2.2.1 of \mainn; and the restarting iterations are those that proceed according to Case 1 or Subsubsubcase 2.2.2.1, as well as the terminating iterations that proceed according to Subsubsubcase 2.2.2.2.   
Next, we say that a given iteration of \mainn\ is {\bf locking}\label{xlocking} iff it proceeds according to Subsubcase 2.1.2 of \mainn. 

Consider any $h\in \mathbb{I}$ and any $i\in \mathbb{I}^h$. We say that the $i$th iteration of \mainn\  is {\bf $\mathbb{I}^h$-transient}\label{xtransient} iff there is a $j\in \mathbb{I}^h$  such that  the following three conditions are satisfied:
\begin{itemize} 
  \item $i\leq j<h$.
  \item The $j$th iteration of \mainn\ restarting.
\item There is no $e$ with $i\leq e< j$ such that the $e$th iteration  of \mainn\ is locking. 
\end{itemize}

For a number $h\in \mathbb{I}$, we define \[\mathbb{I}_{!}^{h}\label{xihe} = \{i\ |\ i\in \mathbb{I}^h\mbox{\em and the $i$'th iteration of \mainn\ is not $\mathbb{I}^h$-transient}\}.\]

We say that two bodies are {\bf consistent}\label{xconsis} with each other iff one is an extension of the other. This, of course, includes the case of their being simply equal.

\begin{lem}\label{jinan}
%\marginpar{jinan}
Consider any $n\in\{0,\ldots,k\}$, $h\in \mathbb{I}$ and  $i,j\in \mathbb{I}_{!}^{h}$.  Suppose $\vec{E}_i$ has an entry $[n,B_i]$, and $\vec{E}_j$ has an entry $[n,B_j]$.  Then $B_i$ and $B_j$ are consistent with each other. 
\end{lem}

\begin{proof} Assume the conditions of the lemma. The case $i=j$ is trivial, so we shall assume that $i<j$. 

First, consider the case $n=k$. We thus want to show that the master bodies of $\vec{E}_i$ and $\vec{E}_j$ are consistent with each other.   Notice that only those iterations of \mainn\ affect the master body of (the evolving) $\vec{E}$ that are either restarting or locking. So, if no iteration between\footnote{Here and later, we may terminologically identify iterations with their ordinal numbers.} $i$ and $j$ is either restarting or locking, then the master entry of $\vec{E}_j$ is the same as that of $\vec{E}_i$, and we are done. Now suppose there is an  $e$ with  $i\leq e\leq j$ such that the $e$th iteration is restarting or locking. We may assume that $e$ is the smallest such number. Then the $e$th iteration cannot be  restarting, because this would make  the $i$th iteration  $\mathbb{I}^h$-transient. Thus,  the $e$th iteration  is locking. Such an iteration ``locks'' the master body of $\vec{E}_i$, in the sense that no later iterations   can destroy what is already there --- such iterations will only extend the   master body. So, the master body of $\vec{E}_j$ is an extension of that of $\vec{E}_i$, implying that, as desired, the two bodies are consistent with each other. 

Now, for the rest of this proof, assume $n<k$. Note that $i,j>1$, because $\vec{E}_1$ has no common ($n$-indexed) entries. Further note that the $(i-1)$th and $(j-1)$th iterations  are not restarting ones, because \restart\ erases all common entries. Hence, obviously, both $i-1$ and $j-1$ are in $\mathbb{I}_{!}^{h}$. 

The case of either $B_i$ or $B_j$ being empty is  trivial, because the empty body is consistent with every body. Thus, we shall assume that 
\[\mbox{\it $B_i$ looks like $\bigl((\vec{\alpha}_1,p_1),\ldots,(\vec{\alpha}_a,p_a),(\vec{\alpha},p)\bigr)$ and $B_j$ looks like 
$\bigl((\vec{\beta}_1,q_1),\ldots,(\vec{\beta}_b,q_b),(\vec{\beta},q)\bigr)$}\] for some $a,b\geq 0$. In what follows, we will be 
using $\vec{P}$ and $\vec{Q}$ as abbreviations of ``$(\vec{\alpha}_1,p_1)$, \ldots, $(\vec{\alpha}_a,p_a)$'' and 
``$(\vec{\beta}_1,q_1)$, \ldots, $(\vec{\beta}_b,q_b)$'', respectively. Thus, 
$B_i=\bigl(\vec{P},(\vec{\alpha},p)\bigr)$ and $B_j=\bigl(\vec{Q},(\vec{\beta},q)\bigr)$.

We prove the lemma by complete  induction on $i+ j$. Assume the aggregation $\vec{E}_{i-1}$ contains the entry  $[n,B_i]$. Since $(i-1)+j<i+j$ and (as we established just a while ago) $(i-1)\in \mathbb{I}_{!}^{h}$, the induction hypothesis applies, according to which $B_i$ is consistent with $B_j$, as desired. The case of $\vec{E}_{j-1}$ containing  the entry  $[n,B_j]$ is  similar. Now, for the rest of the present proof, we assume that 
%\marginpar{hunqian}
\begin{equation}\label{hunqian}
\mbox{\it $\vec{E}_{i-1}$ does not have the entry $[n,B_i]$, and  $\vec{E}_{j-1}$ does not have the entry $[n,B_j]$.}
\end{equation}
Assume $a<b$. Then $b\geq 1$. In view of this fact and (\ref{hunqian}), it is easy to see that $\vec{E}_{j-1}$ contains an $n$-indexed entry whose body is $(\vec{Q})$. By the induction hypothesis, $\bigl(\vec{P},(\vec{\alpha},p)\bigr)$ is consistent with $(\vec{Q})$, meaning (as $a+1\leq b$) that the latter is an extension of the former. Hence, $\bigl(\vec{P},(\vec{\alpha},p)\bigr)$ is just as well consistent with $\bigl(\vec{Q},(\vec{\beta},q)\bigr)$, as desired. The case of $b<a$ will be handled in a similar way. 

Now, for the rest of this proof, we further assume that $a=b$.  
We claim that 
%\marginpar{hun}
\begin{equation}\label{hun}
\vec{P}=\vec{Q},\mbox{\it \ i.e., \ }\bigl((\vec{\alpha}_1,p_1),\ldots,(\vec{\alpha}_a,p_a) \bigr)=\bigl((\vec{\beta}_1,q_1),\ldots,(\vec{\beta}_b,q_b)\bigr).
\end{equation}
Indeed, the case of $a,b=0$ is trivial. Otherwise, if $a,b\not=0$, in view of (\ref{hunqian}),  obviously   $\vec{E}_{i-1}$ contains the entry 
$[n,(\vec{P}) ]$ and  $\vec{E}_{j-1}$ contains the entry $[n,(\vec{Q})]$. Hence, by the induction hypothesis, the two bodies $(\vec{P})$ and 
$(\vec{Q})$ are consistent, which, as $a=b$, simply means that they are identical. (\ref{hun}) is thus verified. 
In view of (\ref{hun}), all that now remains to show is that $(\vec{\alpha},p)=(\vec{\beta},q)$. 

Assume $a$ is odd. Analyzing the work of $\mainn$ and keeping (\ref{hunqian}) in mind, we see that the $(i-1)$th iteration of \mainn\ proceeds according to Subsubcase 2.1.1, where the central triple of $\vec{E}_{i-1}$ is $\bigl(C,(\vec{P}),n\bigr)$ for a certain even-size body $C$, and  $
\simm_{n}^{\bullet}\bigl(C^{\even},(\vec{P})^{\odd}\bigr)=+(\vec{\alpha},p)$. 
Similarly, the $(j-1)$th iteration of \mainn\ proceeds according to Subsubcase 2.1.1, where the central triple of $\vec{E}_{j-1}$ is $\bigl(D,(\vec{Q}),n\bigr)$ --- which, by (\ref{hun}), is the same as $(D,(\vec{P}),n)$ --- for a certain even-size body $D$, and 
$\simm_{n}^{\bullet}\bigl(D^{\even},(\vec{P})^{\odd}\bigr)=+(\vec{\beta},q)$. Here, if one of the bodies $C,D$ is empty, the two bodies are consistent with each other.  Otherwise obviously $n>0$, $\vec{E}_{i-1}$ contains the entry $[n-1,C]$, and $\vec{E}_{j-1}$ contains the entry $[n-1,D]$. Then, by the induction hypothesis, again, $C$ is consistent with $D$. Thus, in either case, $C$ and $D$ are consistent. Then clause 2 of Lemma \ref{golemma} implies that 
$\simm_{n}\bigl(C^{\even},(\vec{P})^{\odd}\bigr)=\simm_{n}\bigl(D^{\even},(\vec{P})^{\odd}\bigr)$. Consequently,  $(\vec{\alpha},p)=(\vec{\beta},q)$, as desired. 

The case of $a$ being even is rather similar. In this case, the $(i-1)$th iteration of \mainn\ deals with Subsubcase 2.2.1, where the central triple of $\vec{E}_{i-1}$ is $\bigl((\vec{P}),C,n+1\bigr)$ for a certain odd-size body $C$, with $\simm_{n+1}^{\bullet}\bigl((\vec{P})^{\even},C^{\odd}\bigr)=-(\vec{\alpha},p)$. 
And the $(j-1)$th iteration of \mainn\ also deals with Subsubcase 2.2.1, where the central triple of $\vec{E}_{j-1}$ is $\bigl((\vec{P}),D,n+1\bigr)$ for a certain odd-size body $D$, with  
$\simm_{n+1}^{\bullet}\bigl((\vec{P})^{\even},D^{\odd}\bigr)=-(\vec{\beta},q)$. So, $\vec{E}_{i-1}$ contains the entry $[n+1,C]$ and $\vec{E}_{j-1}$ contains the entry $[n+1,D]$. Therefore,   by the induction hypothesis, $C$ and $D$ are again consistent. Then clause 1 of Lemma \ref{golemma} implies that 
$\simm_{n+1}\bigl((\vec{P})^{\even},C^{\odd}\bigr)=\simm_{n+1}\bigl((\vec{P})^{\even},D^{\odd}\bigr)$, meaning that, as desired,   $(\vec{\alpha},p)=(\vec{\beta},q)$.   
\end{proof}

Consider any $n\in\{0,\ldots,k\}$ and $h\in \mathbb{I}$.  We define  
%\marginpar{eqdef}
\begin{equation}\label{eqdef}
 \mathbb{B}_{n}^{h} 
\end{equation}  
as the smallest-size body such that, for every $i\in \mathbb{I}_{!}^{h}$, whenever 
$\vec{E}_i$ has an $n$-indexed entry, $\mathbb{B}_{n}^{h}$ is a (not necessarily proper) extension of that entry's body. In view of Lemma   \ref{jinan}, such a $\mathbb{B}_{n}^{h}$ always exists. 
We further define the bodies   $\mathbb{B}_{n}^{h}\hspace{-2pt}\uparrow$ and   $\mathbb{B}_{n}^{h}\hspace{-2pt}\downarrow$ as follows. Let $\mathbb{B}_{n}^{h}=(O_1,\ldots,O_s)$. We agree that below and later, where $t$ is $0$ or a negative integer, the denotation of an expression like $(P_1,\ldots,P_t)$ should be understood as the empty tuple $()$. Then: 
\[
\mathbb{B}_{n}^{h}\hspace{-2pt}\uparrow=\left\{\begin{array}{ll}
(O_1,\ldots,O_s) & \mbox{\it if $s$ is even;}\\
(O_1,\ldots,O_{s-1}) & \mbox{\it if $s$ is odd.}
\end{array}\right.\label{x235}
\]

\[\mathbb{B}_{n}^{h}\hspace{-2pt}\downarrow=\left\{\begin{array}{ll}
(O_1,\ldots,O_s) & \mbox{\it if $s$ is odd;}\\
(O_1,\ldots,O_{s-1}) & \mbox{\it if $s$ is even.} 
\end{array}\right. \]

Assume $h\in \mathbb{I}$, $n\in\{0,\ldots,k\}$, and $(P_1,\ldots,P_p)$ is a nonempty, not necessarily proper, restriction of the  body $\mathbb{B}_{n}^{h}$. By the {\bf $(h,n)$-birthtime}\label{xbirth} of $(P_1,\ldots,P_p)$ 
we shall mean  the smallest number $i\in \mathbb{I}^{h}_{!}$ such that,   for some (not necessarily proper) extension $B$ of 
$(P_1,\ldots,P_p)$, $\vec{E}_i$ has the entry $[n,B]$. We extend this concept to the case  $p=0$  by stipulating that the $(h,n)$-birthtime of the empty body $()$ is always $0$. In informal discourses  we may say ``$(O_1,\ldots,O_p)$  
was {\bf $(h,n)$-born}\label{xbrno} at time $i$'' to mean that $i$ is the $(h,n)$-birthtime of $(O_1,\ldots,O_p)$.  When $h$ and $n$ are fixed or clear from the context, we may omit a reference to $(h,n)$ and simply say ``birthtime'' or ``born''. 

\begin{lem}\label{wuhano} Consider any $h\in \mathbb{I}$ and $n\in\{1,\ldots,k\}$. Let $\mathbb{B}_{n-1}^{h}\hspace{-2pt}\downarrow =(P_1,\ldots,P_p)$ and 
$\mathbb{B}_{n}^{h}\hspace{-2pt}\uparrow =(Q_1,\ldots,Q_q)$, where $q>0$. Further let $i_P$  be the $(h,n-1)$-birthtime of 
$(P_1,\ldots,P_p)$ and $i_Q$ be the $(h,n)$-birthtime of $(Q_1,\ldots,Q_q)$. 
%\marginpar{wuhano}
\begin{enumerate}[label=\arabic*.]
\item If  $i_Q>i_P$, then we have:
%\marginpar{aba0-aba3}
\begin{eqnarray}
& & \simm_{n}^{\bullet}\bigl((\mathbb{B}_{n-1}^{h})^{\even}, (Q_1,\ldots,Q_{q-1})^{\odd}\bigr)=+Q_q;\label{aba0}\\
 & & \mbox{\it The triple $\bigl((\mathbb{B}_{n-1}^{h})^{\even}, (Q_1,\ldots,Q_{q-1})^{\odd},n\bigr)$ is saturated;}\label{aba1}\\
& & \simm_{n}^{\rightarrow} \bigl((\mathbb{B}_{n-1}^{h})^{\even}, (Q_1,\ldots,Q_{q-1})^{\odd}\bigr)=\mathbb{B}^{h}_{n}\hspace{-2pt}\uparrow;\label{aba3}\\
& & \simm_{n}^{\leftarrow} \bigl((\mathbb{B}_{n-1}^{h})^{\even}, (Q_1,\ldots,Q_{q-1})^{\odd}\bigr)=\mathbb{B}^{h}_{n-1}.\label{aba2}
\end{eqnarray}
 
\item If  $i_P>i_Q$,  then we have:
%\marginpar{caca0-caca3}
\begin{eqnarray}
& & \simm_{n}^{\bullet}\bigl((P_1,\ldots,P_{p-1})^{\even}, (\mathbb{B}_{n}^{h})^{\odd}\bigr)=-P_p;\label{caca0}\\
 & & \mbox{\it The triple $\bigl((P_1,\ldots,P_{p-1})^{\even}, (\mathbb{B}_{n}^{h})^{\odd},n\bigr)$ is saturated;}\label{caca1}\\
& & \simm_{n}^{\rightarrow} \bigl((P_1,\ldots,P_{p-1})^{\even}, (\mathbb{B}_{n}^{h})^{\odd}\bigr)=\mathbb{B}^{h}_{n};\label{caca3}\\
& & \simm_{n}^{\leftarrow} \bigl((P_1,\ldots,P_{p-1})^{\even}, (\mathbb{B}_{n}^{h})^{\odd}\bigr)=\mathbb{B}^{h}_{n-1}\hspace{-2pt}\downarrow.\label{caca2}
\end{eqnarray}
\end{enumerate}
\end{lem}

\begin{proof} Assume the conditions of the lemma. Take a note of the fact that $i_P,i_Q\in\mathbb{I}^{h}_{!}$. 
\begin{enumerate}[label=\arabic*.]
\item%{\em Clause 1}. 
  Assume $i_Q>i_P$. Note that, by the definition of $\mathbb{B}_{n}^{h}\hspace{-2pt}\uparrow$, \ $q$ is even. 

Since ($q>0$ and) $q$ is even, at time $i_Q$ the body $(Q_1,\ldots,Q_q)$ obviously must have been ``born'' --- i.e., the transition from the $(i_Q-1)$th iteration to the $i_Q$th iteration must have happened --- according to the scenario of Subcase 2.1 of \mainn.  Namely, in that scenario,  the central triple of $\vec{E}_{i_Q-1}$ was $\bigl(C,(Q_1,\ldots,Q_{q-1}), n\bigr)$ for a certain even-size body $C$, and 
$\simm_{n}^{\bullet}\bigl(C^{\even},(Q_1,\ldots,Q_{q-1})^{\odd}\bigr)=+Q_{q}$. Since the 
$(i_Q-1)$th iteration of \mainn\ was not a restarting one, $i_Q-1$ is in $\mathbb{I}^{h}_{!}$ just like $i_Q$ is. Therefore, by the definition (\ref{eqdef}) of $\mathbb{B}_{n-1}^{h}$, \ $\mathbb{B}_{n-1}^{h}$ is an extension of $C$. Now, (\ref{aba0}) holds by clause 2 of Lemma \ref{golemma}. 

To verify  claim (\ref{aba1}), deny it for a contradiction. That is, assume there is a proper restriction $D$ of 
$(\mathbb{B}_{n-1}^{h})^{\even}$ such that $\simm_{n}^{\bullet}\bigl(D, (Q_1,\ldots,Q_{q-1})^{\odd}\bigr)$ is positive.  Since $(\mathbb{B}_{n-1}^{h})^{\even}$ has a proper restriction, the size of $\mathbb{B}_{n-1}^{h}$ is at least $2$, and therefore, by the definition of $\mathbb{B}_{n-1}^{h}\hspace{-2pt}\downarrow$, \ $p$ is an odd positive integer. 
Since $D$ is a proper restriction of $(\mathbb{B}_{n-1}^{h})^{\even}$, it is also a (not necessarily proper) restriction of $(P_1,\ldots,P_p)^{\even}$. Furthermore, since $p$ is odd, $(P_1,\ldots,P_p)^{\even}=(P_1,\ldots,P_{p-1})^{\even}$. Consequently, $D=(P_1,\ldots,P_{r})^{\even}$ for some $r$ strictly smaller than $p$. We may assume that $r$ is even, for otherwise $(P_1,\ldots,P_r)^{\even}=(P_1,\ldots,P_{r-1})^{\even}$ and we could have taken $r-1$ instead of $r$. Thus, for the nonnegative even integer $r$ with $r<p$, 
%\marginpar{555}
\begin{equation}\label{555}
\mbox{\em $\simm_{n}^{\bullet}\bigl((P_1,\ldots,P_{r})^{\even},(Q_1,\ldots,Q_{q-1})^{\odd}\bigr)$ is positive}.
\end{equation}
Let $j$ be the $(h,n-1)$-birthtime of $(P_1,\ldots,P_{r+1})$. Note that $j\leq i_P$, and hence $j<i_Q$. Since 
${r+1}$ is odd, $(P_1,\ldots,P_{r+1})$ must have been born according to the scenario of Subsubcase 2.2.1 of \mainn. Namely, in that scenario, ($j>1$ and) $j-1\in\mathbb{I}^{h}_{!}$, the central triple of $\vec{E}_{j-1}$ is $\bigl((P_1,\ldots,P_{r}),A,n\bigr)$ for some odd-size body $A$, and 
$\simm_{n}^{\bullet} \bigl((P_1,\ldots,P_{r})^{\even},A^{\odd}\bigr)=-P_{r+1}$. By definition (\ref{eqdef}), the  body $\mathbb{B}_{n}^{h}$ is an extension of $A$. But, 
since $j< i_Q$,  $(Q_1,\ldots,Q_q)$  was not yet $(h,n)$-born at time $j$. So, we must have $A=(Q_1,\ldots,Q_s)$ for some $s\leq q-1$.  Therefore, by clause 1 of Lemma \ref{golemma},  
\[\simm_{n}^{\bullet} \bigl((P_1,\ldots,P_{r})^{\even},(Q_1,\ldots,Q_{q-1})^{\odd}\bigr)=-P_{r+1}.\] 
The above, however, contradicts (\ref{555}). Claim (\ref{aba1}) is thus proven. 

To justify  (\ref{aba3}),  assume 
$\simm_{n}^{\rightarrow}\bigl((\mathbb{B}_{n-1}^{h})^{\even},(Q_1,\ldots,Q_{q-1})^{\odd} \bigr)=(U_1,\ldots,U_u)$. We want to show that $(U_1,\ldots,U_u)=(Q_1,\ldots,Q_{q})$. With (\ref{aba0}) and the evenness of $q$ in mind, we can see directly from the definition of $\simm_{n}^{\rightarrow}$ that $u=q$, and that 
$(U_1,\ldots,U_q)^{\odd}=(Q_1,\ldots,Q_{q-1})^{\odd}$. $q$'s being even  further implies that $(Q_1,\ldots,Q_{q-1})^{\odd}=(Q_1,\ldots,Q_{q})^{\odd}$. So, $(U_1,\ldots,U_u)^{\odd}=(Q_1,\ldots,Q_{q})^{\odd}$.
It remains to show that we have
$(U_1,\ldots,U_u)^{\even}=(Q_1,\ldots,Q_{q})^{\even}$ as well, i.e., $(U_1,\ldots,U_q)^{\even}=(Q_1,\ldots,Q_{q})^{\even}$. 
Consider any even $r\in\{1,\ldots,q\}$. Let $j$ be the $(h,n)$-birthtime of 
$(Q_1,\ldots,Q_r)$. Obviously this body must have been born according to the scenario of Subcase 2.1 of \mainn\ in which  $ j-1\in\mathbb{I}^{h}_{!}$,   
 $\vec{E}_{j-1}$ has the entry $[n,(Q_1,\ldots,Q_{r-1})]$ and, 
with $\bigl(C,(Q_1,\ldots,Q_{r-1}),n\bigr)$ being the central triple of $\vec{E}_{j-1}$ for some even-size body $C$, we have 
$\simm_{n}^{\bullet}\bigl(C^{\even},(Q_1,\ldots,Q_{r-1})^{\odd}\bigr)=+Q_r$. By definition (\ref{eqdef}), $\mathbb{B}_{n-1}^{h}  $ is an extension of $C $. So, by clause 2 of Lemma \ref{golemma}, 
%\marginpar{29aq}
\begin{equation}\label{29aq}
\simm_{n}^{\bullet}\bigl((\mathbb{B}_{n-1}^{h})^{\even},(Q_1,\ldots,Q_{r-1})^{\odd}\bigr)=+Q_r.
\end{equation}
But how does the computation of (\ref{29aq}) 
differ from the computation of (\ref{aba0})?  
 The two computations proceed in exactly the same ways, with the variable $S$ of $\simm_{n}^{\bullet}$ going through exactly the same values in both cases, with the only difference  that, while the computation of  (\ref{29aq}) stops after $S$ takes its $(r/2)$th positive value $+U_r$ and returns that value as $+Q_r$, the computation of (\ref{aba0}) 
 continues further (if $r\not=q$) until the value of $S$ becomes $+U_q$. As we see, we indeed have $U_r=Q_r$ as desired. Claim (\ref{aba3}) is now verified. 

Claim (\ref{aba2}) can be verified in a rather similar way. Assume 
\[\simm_{n}^{\leftarrow}\bigl((\mathbb{B}_{n-1}^{h})^{\even},(Q_1,\ldots,Q_{q-1})^{\odd} \bigr)=(V_1,\ldots,V_v).\] We want to show that $(V_1,\ldots,V_v)=\mathbb{B}^{h}_{n-1}$.
By the 
definition of $\simm_{n}^{\leftarrow}$, \ $(V_1,\ldots,V_v)^{\even}= (\mathbb{B}^{h}_{n-1})^{\even}$. 
It remains to show that we also have  $(V_1,\ldots,V_v)^{\odd}=(\mathbb{B}^{h}_{n-1})^{\odd}$. Notice that 
$(\mathbb{B}^{h}_{n-1})^{\odd}=(\mathbb{B}_{n-1}^{h}\hspace{-2pt}\downarrow)^{\odd}=(P_1,\ldots,P_p)^{\odd}$, and that ($p\leq v$ and) 
$(V_1,\ldots,V_v)^{\odd}=(V_1,\ldots,V_p)^{\odd}$.
 So, what we want to show is $(V_1,\ldots,V_p)^{\odd}=(P_1,\ldots,P_p)^{\odd}$. 
 Consider any odd $r\in\{1,\ldots,p\}$. 
Let $j$ be the $(h,n-1)$-birthtime of 
$(P_1,\ldots,P_r)$. Note that $j\leq i_P$ and hence $j<i_Q$. The birth of $(P_1,\ldots,P_r)$ should have occurred according to Subsubcase 2.2.1 of \mainn, in a situation where $1\leq j-1\in\mathbb{I}^{h}_{!}$, the central triple of $\vec{E}_{j-1}$ is 
$\bigl((P_1,\ldots,P_{r-1}),C,n\bigr)$  for some odd-size body $C$, and  
$\simm_{n}^{\bullet}\bigl((P_1,\ldots,P_{r-1})^{\even},C^{\odd}\bigr)=-P_r$. But  $(Q_1,\ldots,Q_{q})$ is an extension of $C$ because so is $\mathbb{B}^{h}_{n}$. In fact, it is a proper extension, because $(Q_1,\ldots,Q_q)$  was not yet $(h,n)$-born at time $j$. So,
 $(Q_1,\ldots,Q_{q-1})^{\odd}$ is a (not necessarily proper) extension of $C^{\odd}$. Hence, by clause 1 of Lemma \ref{golemma}, 
%\marginpar{29z}
\begin{equation}\label{29z}
\simm_{n}^{\bullet}\bigl((P_1,\ldots,P_{r-1})^{\even},(Q_1,\ldots,Q_{q-1})^{\odd}\bigr)=-P_r.
\end{equation}
But how does the computation of (\ref{29z}) 
differ from the computation of (\ref{aba0})?  
 The two computations proceed in exactly the same ways, with the variable $S$ of $\simm_{n}^{\bullet}$ going through exactly the same values in both cases, with the only difference that, while the computation of  (\ref{29z}) stops after $S$ takes its $((r+1)/2)$th negative value $-V_r$ and returns that value as $-P_r$, the computation of (\ref{aba0}) 
 continues further until the value of $S$ becomes $+Q_q$. As we see, we indeed have $V_r=P_r$ as desired. This completes our proof of clause 1 of the lemma.\medskip

\item%{\em Clause 2}. 
  Assume $i_P>i_Q$.  Note that $p$ is odd, $q$ is even and $q\not=0$. 

The way $(P_1,\ldots,P_p)$ was born is that the central triple of
$\vec{E}_{i_P-1}$ had the form 
$\bigl((P_1,\ldots,P_{p-1}),C, n\bigr)$ for a certain odd-size body $C$, and 
$\simm_{n}^{\bullet}\bigl((P_1,\ldots,P_{p-1})^{\even},C^{\odd}\bigr)=-P_{p}$. But $\mathbb{B}_{n}^{h}$ is an extension of $C$. Therefore (\ref{caca0}) holds by  Lemma \ref{golemma}.   

To verify  (\ref{caca1}), deny it for a contradiction:  assume there is a proper restriction $D$ of 
$(\mathbb{B}_{n}^{h})^{\odd}$ such that $\simm_{n}^{\bullet}\bigl((P_1,\ldots,P_{p-1})^{\even},D\bigr)$ is negative.  $D$'s being  a proper restriction of $(\mathbb{B}_{n}^{h})^{\odd}$ implies that  $D=(Q_1,\ldots,Q_{r})^{\odd}$ for some odd $r$ --- fix it --- strictly smaller than $q$. Thus, 
%\marginpar{555a}
\begin{equation}\label{555a}
\mbox{\em $\simm_{n}^{\bullet}\bigl((P_1,\ldots,P_{p-1})^{\even},(Q_1,\ldots,Q_{r})^{\odd}\bigr)$ is negative}.
\end{equation}
Let $j$ be the birthtime of $(Q_1,\ldots,Q_{r+1})$. Note that $j\leq
i_Q$, and hence $j<i_P$. $(Q_1,\ldots,Q_{r+1})$'s birth must have
happened in a situation where $1\leq j-1\in\mathbb{I}^{h}_{!}$,  the
central triple of $\vec{E}_{j-1}$ happens to be $\bigl(A,(Q_1,\ldots,Q_{r}),n\bigr)$ for some even-size body $A$, and 
$\simm_{n}^{\bullet} \bigl(A^{\even},(Q_1,\ldots,Q_{r})^{\odd}\bigr)=+Q_{r+1}$.  $\mathbb{B}_{n-1}^{h}$ is an extension of $A$. But 
since $j< i_P$,  $(P_1,\ldots,P_p)$ was not yet born at time $j$. So,  $A=(P_1,\ldots,P_s)$ for some $s\leq p-1$.  Therefore, by Lemma \ref{golemma},  
\[\simm_{n}^{\bullet} \bigl((P_1,\ldots,P_{p-1})^{\even},(Q_1,\ldots,Q_{r})^{\odd}\bigr)=+Q_{r+1}.\] 
The above, however, contradicts (\ref{555a}). Claim (\ref{caca1}) is thus proven. 

For (\ref{caca3}),  assume 
$\simm_{n}^{\rightarrow} \bigl((P_1,\ldots,P_{p-1})^{\even},
(\mathbb{B}_{n}^{h})^{\odd}\bigr)=(U_1,\ldots,U_u)$. We want to show $(U_1,\ldots,U_u)=\mathbb{B}^{h}_{n}$.
Directly from the 
definition of $\simm_{n}^{\rightarrow}$, \  
$(U_1,\ldots,U_u)^{\odd}=(\mathbb{B}_{n}^{h})^{\odd}$. 
It remains to show that  $(U_1,\ldots,U_u)^{\even}=(\mathbb{B}_{n}^{h})^{\even}$. Note that $(\mathbb{B}_{n}^{h})^{\even}=(\mathbb{B}_{n}^{h}\hspace{-2pt}\uparrow)^{\even}  =(Q_1,\ldots,Q_q)^{\even}$, and that ($q\leq u$ and) $(U_1,\ldots,U_u)^{\even}= (U_1,\ldots,U_q)^{\even}$. So, what we want to show is 
$(U_1,\ldots,U_q)^{\even}=(Q_1,\ldots,Q_q)^{\even}$.   
For this purpose, consider any even $r\in\{1,\ldots,q\}$.  Let $j$ be the $(h,n)$-birthtime of 
$(Q_1,\ldots,Q_r)$. Obviously $(Q_1,\ldots,Q_r)$ must have been born according to the scenario of Subcase 2.1 of \mainn\ in which  $ j-1\in\mathbb{I}^{h}_{!}$,   
 $\vec{E}_{j-1}$ has the entry $[n,(Q_1,\ldots,Q_{r-1})]$ and, 
with $\bigl(C,(Q_1,\ldots,Q_{r-1}),n\bigr)$ being the central triple of $\vec{E}_{j-1}$ for some even-size body $C$, 
$\simm_{n}^{\bullet}\bigl(C^{\even},(Q_1,\ldots,Q_{r-1})^{\odd}\bigr)=+Q_r$. By definition (\ref{eqdef}), $\mathbb{B}_{n-1}^{h}  $ is an extension of $C $. Therefore, by  Lemma \ref{golemma}, $\simm_{n}^{\bullet}\bigl((\mathbb{B}_{n-1}^{h})^{\even},(Q_1,\ldots,Q_{r-1})^{\odd}\bigr)=+Q_r$. However, $(\mathbb{B}_{n-1}^{h})^{\even}=(\mathbb{B}_{n-1}^{h}\hspace{-2pt}\downarrow)^{\even} =(P_1,\ldots,P_p)^{\even}$; further, since $p$ is odd, $(P_1,\ldots,P_p)^{\even}=(P_1,\ldots,P_{p-1})^{\even}$, and hence $(\mathbb{B}_{n-1}^{h})^{\even}=(P_1,\ldots,P_{p-1})^{\even}$. Thus we have:
%\marginpar{29aqqq}
\begin{equation}\label{29aqqq}
\simm_{n}^{\bullet}\bigl((P_1,\ldots,P_{p-1})^{\even},(Q_1,\ldots,Q_{r-1})^{\odd}\bigr)=+Q_r.
\end{equation}
Comparing the computations of (\ref{caca0}) and  
 (\ref{29aqqq}), we see that  
 the two computations proceed in exactly the same ways,  with the only difference that, while the computation of  (\ref{29aqqq}) stops after variable $S$ of $\simm_{n}^{\bullet}$ takes its $(r/2)$th positive value $+U_r$ and returns that value as $+Q_r$, the computation of (\ref{caca0}) 
 continues further until the value of $S$ becomes $-P_p$. As we see, we indeed have $U_r=Q_r$ as desired. Claim (\ref{caca3}) is verified. 

For (\ref{caca2}), assume 
 $\simm_{n}^{\leftarrow} \bigl((P_1,\ldots,P_{p-1})^{\even}, (\mathbb{B}_{n}^{h})^{\odd}\bigr)=
(V_1,\ldots,V_v)$. We want to show that $(V_1,\ldots,V_v)=\mathbb{B}^{h}_{n-1}\hspace{-2pt}\downarrow$. With (\ref{caca0}) and the oddness of $p$ in mind, we see from  the 
definition of $\simm_{n}^{\leftarrow}$ that $v=p$, and that  $(V_1,\ldots,V_p)^{\even}= (P_1,\ldots,P_{p-1})^{\even}$. The fact that $p$ is odd additionally implies  
$(P_1,\ldots,P_{p-1})^{\even}=(P_1,\ldots,P_{p})^{\even}$. Consequently, $(V_1,\ldots,V_v)^{\even}=(P_1,\ldots,P_{p})^{\even}=
(\mathbb{B}^{h}_{n-1}\hspace{-2pt}\downarrow)^{\even}$.
So, it remains to show that we also have   
$(V_1,\ldots,V_v)^{\odd}=(\mathbb{B}^{h}_{n-1}\hspace{-2pt}\downarrow)^{\odd}$, i.e., $(V_1,\ldots,V_p)^{\odd}=(P_1,\ldots,P_{p})^{\odd}$. 
Consider any odd $r\in\{1,\ldots,p\}$. Let $j$ be the $(h,n-1)$-birthtime of 
$(P_1,\ldots,P_r)$. Note that $j\leq i_P$. 
The birth of $(P_1,\ldots,P_r)$ should have occurred according to Subsubcase 2.2.1 of \mainn, in a situation where $1\leq j-1\in\mathbb{I}^{h}_{!}$, the central triple of $\vec{E}_{j-1}$ is 
$\bigl((P_1,\ldots,P_{r-1}),C,n\bigr)$  for some odd-size body $C$, and  
$\simm_{n}^{\bullet}\bigl((P_1,\ldots,P_{r-1})^{\even},C^{\odd}\bigr)=-P_r$. But $\mathbb{B}^{h}_{n}$ is an extension of $C$.  Hence, by Lemma \ref{golemma}, 
%\marginpar{2229a}
\begin{equation}\label{2229a}
\simm_{n}^{\bullet}\bigl((P_1,\ldots,P_{r-1})^{\even},(\mathbb{B}^{h}_{n})^{\odd}\bigr)=-P_r.
\end{equation}
Compare the computations of  (\ref{2229a}) 
and (\ref{caca0}). The two computations proceed in exactly the same ways, with the only difference  that, while the computation of  (\ref{2229a}) stops after $S$ takes its $((r+1)/2)$th negative value $-V_r$ and returns that value as $-P_r$, the computation of (\ref{caca0}) 
 continues further (if $r\not=p$) until the value of $S$ becomes
 $-P_p$. Thus $V_r=P_r$, as desired.\qedhere 
\end{enumerate}
\end{proof}

\noindent We agree for the rest of Section \ref{ssind} that 
$\hbar\label{xhbar}$
is the greatest element of $\mathbb{I}$. The existence of such an element is guaranteed by Lemma \ref{beijing}. 

\begin{lem}\label{a21} The following statements are true (with $\hbar$ as above): 
%\marginpar{a21}
\begin{enumerate}[label=\arabic*.]
\item For every $n\in\{0,\ldots,k\}$, the size of  $\mathbb{B}^{\hbar}_{n}$ is odd.

\item For every $n\in\{1,\ldots,k\}$, the $(\hbar,n-1)$-birthtime of $\mathbb{B}_{n-1}^{\hbar}$ is greater than the $(\hbar,n)$-birthtime of $\mathbb{B}_{n}^{\hbar}$. 

\item For every $n\in\{0,\ldots,k\}$, the scale of 
the last organ of $\mathbb{B}_{n}^{\hbar}$ is the same as the master
scale of $\vec{E}_{\hbar}$.
\end{enumerate}
\end{lem}

\begin{proof}\hfill 
\begin{enumerate}[label=\arabic*.]
\item%{\em Clause 1}. 
Assume $n\in\{0,\ldots,k\}$, $\mathbb{B}^{\hbar}_{n}=(Q_1,\ldots,Q_q)$, and $i_Q$ is the $(\hbar,n)$-birthtime of $\mathbb{B}^{\hbar}_{n}$.
 If $n\geq 1$, further assume that $\mathbb{B}_{n-1}^{\hbar}=(P_1,\ldots,P_p)$, and   $i_P$ is the $(\hbar,n-1)$-birthtime of $\mathbb{B}_{n-1}^{\hbar}$. 

We first verify that 
%\marginpar{hhh}
\begin{equation}\label{hhh}
\mbox{\it If $n=0$, then $q$ is odd.}
\end{equation}
Indeed, assume $n=0$.  Consider the last, i.e., $\hbar$th, iteration of $\mainn$. 
This must be an iteration that proceeds according to Subsubsubcase 2.2.2.2, because all other sorts of iterations always either \repeatt\ or \restart. Namely, the central triple of $\vec{E}_\hbar$ is   $\bigl((),B,0\bigr)$ for some odd-size body $B$, and $\simm_{0}^{\bullet}\bigl((),B^{\odd}\bigr)$ is negative. Of course  the $\hbar$th iteration  is not $\hbar$-transient, so  $\hbar\in \mathbb{I}_{!}^{\hbar}$.  By definition (\ref{eqdef}),    $\mathbb{B}_{0}^{\hbar}$ is an extension of $B$. So,
$B=(Q_1,\ldots,Q_a)$ for some odd $a$ with $a\leq q$.  
  Suppose $a<q$. Let $i$ be the $(\hbar,0)$-birthtime of $(Q_1,\ldots,Q_{a+1})$. Obviously the birth of $(Q_1,\ldots,Q_{a+1})$ must have occurred according to the scenario of Subcase 2.1 of \mainn\ in which 
$i>1$, $\vec{E}_{i-1}$ contains the entry  $[0, (Q_1,\ldots,Q_a)]$, i.e., $[0,B]$, and $\simm_{0}^{\bullet}\bigl((),B^{\odd}\bigr)=+Q_{a+1}$. This, however, contradicts with our earlier observation that $\simm_{0}^{\bullet}\bigl((),B^{\odd}\bigr)$ is negative. From this contradiction we conclude that $a=q$. If so, (\ref{hhh}) holds, because, as already noted, $a$ is odd.

We next verify that
%\marginpar{hhhh}
\begin{equation}\label{hhhh}
\mbox{\it If $n\in\{1,\ldots,k\}$, then $q$ is odd.}
\end{equation}
  Our proof of (\ref{hhhh}) is, in fact, by induction on $n\geq 1$. Assume $n\in\{1,\ldots,k\}$.   By (\ref{hhh}) if $n=1$ (i.e., if we are dealing with the basis of induction), and by the induction hypothesis if $n>1$ (i.e., if we are dealing with the inductive step), 
 we have: 
%\marginpar{jjj}
\begin{equation}\label{jjj}
\mbox{\it $p$ (the size of $\mathbb{B}_{n-1}^{\hbar}$) is odd.}
\end{equation}
 Obviously, (\ref{jjj}) implies that  $(P_1,\ldots,P_{p})$ was born according to the scenario of Subsubcase 2.2.1 of \mainn\ in which $i_P-1\in\mathbb{I}^{\hbar}_{!}$, the central triple of 
$\vec{E}_{i_P-1}$  is $\bigl((P_1,\ldots,P_{p-1}),C,n\bigr)$
for 
a certain odd-size body $C$, and \[\simm_{n}^{\bullet}\bigl((P_1,\ldots,P_{p-1})^{\even},C^{\odd}\bigr)=-P_p.\] By  definition (\ref{eqdef}), 
$\mathbb{B}_{n}^{\hbar}$ is an extension of $C$. Hence, by clause 1 of Lemma \ref{golemma}, 
%\marginpar{chenghai}
\begin{equation}\label{chenghai} 
\simm_{n}^{\bullet}\bigl((P_1,\ldots,P_{p-1})^{\even},(Q_1,\ldots,Q_{q})^{\odd}\bigr)=-P_p.
\end{equation}

For a contradiction suppose (\ref{hhhh}) is false, i.e., assume $q$ is even. Then $q\geq 2$, because $\mathbb{B}_{n}^{\hbar}=(Q_1,\ldots,Q_q)$ is an extension of the odd-size $C$. 
Remember that $i_Q$ is the $(\hbar,n)$-birthtime of $(Q_1,\ldots,Q_q)$. Since $q$ is even, $(Q_1,\ldots,Q_q)$ must have been born according to the scenario of Subcase 2.1 of \mainn\  in which $i_Q-1\in\mathbb{I}^{\hbar}_{!}$, $\vec{E}_{i_Q-1}$ contains the entry $[n,(Q_1,\ldots,Q_{q-1})]$ and, with 
$\bigl(D,(Q_1,\ldots,Q_{q-1}),n\bigr)$ being the central triple of $\vec{E}_{i_Q-1}$ for some even-size restriction $D$ of   $(P_1,\ldots,P_{p})$, 
\[\simm_{n}^{\bullet}\bigl(D^{\even},(Q_1,\ldots,Q_{q-1})^{\odd}\bigr)=+Q_q .\] But  since --- by (\ref{jjj}) --- $p$ is odd,  $(P_1,\ldots,P_{p-1})$ is an extension of $D$. Hence, by clause 2 of Lemma \ref{golemma},  
\[\simm_{n}^{\bullet}\bigl((P_1,\ldots,P_{p-1})^{\even},(Q_1,\ldots,Q_{q-1})^{\odd}\bigr)=+Q_q,\]
which, as $q$ is even and hence $(Q_1,\ldots,Q_{q})^{\odd}=(Q_1,\ldots,Q_{q-1})^{\odd}$, is the same as to say that 
%\marginpar{jj15}
\begin{equation}\label{jj15}
 \simm_{n}^{\bullet}\bigl((P_1,\ldots,P_{p-1})^{\even},(Q_1,\ldots,Q_{q})^{\odd}\bigr)=+Q_q. 
\end{equation}
Comparing (\ref{chenghai})  with (\ref{jj15}), we see a desired contradiction. This completes our proof of (\ref{hhhh}) and hence of clause 1 of the lemma, because the latter is nothing but (\ref{hhh}) and (\ref{hhhh}) put together.\vspace{7pt}

\item%{\em Clause 2}. 
Assume $n\in\{1,\ldots,k\}$, $\mathbb{B}_{n-1}^{\hbar}=(P_1,\ldots,P_p)$,   $i_P$ is the $(\hbar,n-1)$-birthtime of $\mathbb{B}_{n-1}^{\hbar}$, 
$\mathbb{B}^{\hbar}_{n}=(Q_1,\ldots,Q_q)$, and $i_Q$ is the $(\hbar,n)$-birthtime of $\mathbb{B}^{\hbar}_{n}$.
  For a contradiction, further assume $i_P\leq i_Q$. From the already verified clause 1 of the present lemma, we know that both $p$ and $q$ are odd. The oddness of $p$ implies that, at time $i_P$,   
$(P_1,\ldots,P_p)$  was  born  according to the scenario of Subsubcase 2.2.1 of \mainn\ in which $i_P-1\in \mathbb{I}^{\hbar}_{!}$,  the central triple of $\vec{E}_{i_P-1}$ is $\bigl((P_1,\ldots,P_{p-1}),C,n\bigr)$ for some odd-size body $C$, and $\simm_{n}^{\bullet}\bigl((P_1,\ldots,P_{p-1})^{\even},C^{\odd}\bigr)=-P_p$. By definition (\ref{eqdef}),  $(Q_1,\ldots,Q_{q})$ is an extension of $C$. Further, since $q$ is odd and the body $(Q_1,\ldots,Q_{q})$ was not yet  born  at time $i_P-1$, we have $q\geq 3$, with    $(Q_1,\ldots,Q_{q-2})$ being an extension of $C$. Then, by clause 1 of Lemma \ref{golemma},   
%\marginpar{30a}
\begin{equation}\label{30a}
\simm_{n}^{\bullet}\bigl((P_1,\ldots,P_{p-1})^{\even},(Q_1,\ldots,Q_{q-2})^{\odd}\bigr)=-P_p.
\end{equation}
Let $j$ be the $(\hbar,n)$-birthtime of $(Q_1,\ldots,Q_{q-1})$.  
The birth of $(Q_1,\ldots,Q_{q-1})$ should have occurred according to the scenario Subcase 2.1 of \mainn\ in which $j-1\in\mathbb{I}^{\hbar}_{!}$, the central triple of $\vec{E}_{j-1}$ is $\bigl(D^{\even},(Q_1,\ldots,Q_{q-2})^{\odd}, n\bigr)$ for some even-size body $D$, and   
\[\simm_{n}^{\bullet}\bigl(D^{\even},(Q_1,\ldots,Q_{q-2})^{\odd}\bigr)=+Q_{q-1}.\] By definition (\ref{eqdef}),   $(P_1,\ldots,P_p)$ is an extension of $D$. So, by clause 2 of Lemma \ref{golemma}, 
$\simm_{n}^{\bullet}\bigl((P_1,\ldots,P_p)^{\even},(Q_1,\ldots,Q_{q-2})^{\odd}\bigr)=+Q_{q-1}$. But, since $p$ is odd, we have \[(P_1,\ldots,P_p)^{\even}=(P_1,\ldots,P_{p-1})^{\even}.\] Thus, 
\[\simm_{n}^{\bullet}\bigl((P_1,\ldots,P_{p-1})^{\even},(Q_1,\ldots,Q_{q-2})^{\odd}\bigr)=+Q_{q-1}.\]
The above is in contradiction with (\ref{30a}). 
\vspace{7pt} 

\item%{\em Clause 3}. 
We start with the following claim:\medskip

%\begin{quote}
\begin{cLa}\label{CL:1}%description}
%\item[Claim 1] 
%{\em 
Consider any $n\in\{0,\ldots,k\}$. Assume
$\mathbb{B}^{\hbar}_{n}=(Q_1,\ldots,Q_{q+1})$, and  $t$ is an even
number with $2\leq t \leq q$. Then the scale of $Q_t$ is the same as
that of $Q_{t-1}$.
\end{cLa}%}\end{description}\medskip
%\end{quote}

\noindent To verify this claim, assume its conditions. 
We proceed by induction on $n=0,1,\ldots,k$. 

For the basis of induction, consider the case of $n=0$. Let $i$ be the $(\hbar,0)$-birthtime of  $(Q_1,
\ldots,Q_t)$. Obviously  the $(i-1)$th iteration of \mainn\ follows the scenario of Subcase 2.1 where  $i-1\in\mathbb{I}^{\hbar}_{!}$, 
  the central triple of $\vec{E}_{i-1}$ is  $\bigl((),(Q_1,\ldots,Q_{t-1}),0\bigr)$, and 
%\marginpar{qq1}
\begin{equation}\label{qq1}
\simm_{0}^{\bullet}\bigl(( ),(Q_1,\ldots,Q_{t-1})^{\odd}\bigr)=+Q_{t}.
\end{equation}
Looking back at the description of the procedure $\simm_{0}^{\bullet}$, we see that, in computing (\ref{qq1}), 
the procedure simply lets the scale  of the output $+Q_t$ be a copy of the scale of the ``last-fetched'' organ $Q_{t-1}$.  Done.

For the inductive step,  assume $n\geq 1$. Let  $\mathbb{B}^{\hbar}_{n-1}=(P_1,\ldots,P_p)$. From clause 1 of the present lemma we know that both $p$ and $q+1$ are odd. Note that, for this reason, $\mathbb{B}_{n-1}^{\hbar}\hspace{-2pt}\downarrow=(P_1,\ldots,P_p)$ and $\mathbb{B}_{n}^{\hbar}\hspace{-2pt}\uparrow=(Q_1,\ldots,Q_q)$. Let $i_P$ be the $(\hbar,n-1)$-birthtime of $\mathbb{B}_{n-1}^{\hbar}\hspace{-2pt}\downarrow$, and $i_Q$ be the $(\hbar,n)$-birthtime of $\mathbb{B}_{n}^{\hbar}\hspace{-2pt}\uparrow$. Clause 2 of the present lemma implies that $i_P>i_Q$. Hence the statements (\ref{caca0})-(\ref{caca2}) of Lemma \ref{wuhano}, with $\hbar$ in the role of $h$, are true. 
Let us again remember the work of $\simm^\bullet$ and imagine the computation of (\ref{caca0}) (with $h=\hbar$). With some thought and  with 
(\ref{caca0})-(\ref{caca2}) in mind, we can see the following scenario. At some point --- by the end of one of the iterations of $\loopp_n$, to be more specific --- the variable $R$ of $\simm_{n}^{\bullet}$ takes the value $+Q_{t-1}$.   
Let $g$ be the scale of $Q_{t-1}$. 
By the end of the next iteration of $\loopp_n$, the variable $S$ of $\simm_{n}^{\bullet}$ becomes either $+Q_t$, or $-P_{j-1}$ for some even $j\in\{1,\ldots,p\}$, with the scale of $S$ in either case being the same as the scale $g$ of the latest (by that time) value of $R$. Thus, if $S$ becomes $+Q_t$, the scale of $Q_t$ is the same as that of $Q_{t-1}$, and  we are done. If $S$ becomes  $-P_{j-1}$, then, immediately after that (on the same iteration of $\loopp_n$), $R$ takes the value $-P_{j}$. By the induction hypothesis, the scale of $P_{j}$ is the same as the scale $g$ of $P_{j-1}$. On the iterations of $\loopp_n$ that follow, $S$ and $R$ may take 
several (possibly zero) consecutive values from the series $-P_{j+1},-P_{j+3},\ldots$ and $-P_{j+2},-P_{j+4},\ldots$, respectively, and the  scales of all these values will remain to be $g$ for the same reasons as above. Sooner or later, after this series of negative values, $S$ becomes $+Q_t$. The scale of this signed organ, as before, will be the same as the scale $g$ of the latest  value of $R$.  The scale of $Q_t$ is thus the same as that of $Q_{t-1}$, which ends  
 our proof of Claim \ref{CL:1}.\medskip

Now, we prove clause 3 of the lemma by induction on $k-n$.  Let $m$ be the master scale of 
$\vec{E}_{\hbar}$. The basis case of $k-n=0$, i.e., $n=k$,  is straightforward. 
Next, consider any $n\in\{1,\ldots,k\}$. By the induction hypothesis, the scale of the last organ of $\mathbb{B}_{n}^{\hbar}$ is $m$. Let,
as in the inductive step of the above proof of Claim \ref{CL:1}, $\mathbb{B}^{\hbar}_{n-1}=(P_1,\ldots,P_p)$ and $\mathbb{B}^{\hbar}_{n}=(Q_1,\ldots,Q_{q+1})$. Arguing as in that proof --- with $q+1$ in the role of $t-1$, $m$ in the role of $g$ and relying on Claim \ref{CL:1} itself where the proof of the inductive step of the proof of Claim \ref{CL:1} relied on its induction hypothesis --- we find that, in the process of computing (\ref{caca0}) (with $h=\hbar$), at some point, the variable $R$ of the procedure $\simm^{\bullet}_{n}$ takes the value $+Q_{q+1}$ (its last positive value) and that, beginning from that point, the scale $m$ will be inherited by all subsequent negative values that the variables $S$ and $R$ assume, which (in the present case) include the final value $-P_p$ assumed by $S$.
 Thus, as desired, the scale of the last organ $P_p$ 
 of $\mathbb{B}_{n-1}^{\hbar}$ is the same as the master scale $m$ of $\vec{E}_{\hbar}$.\qedhere
\end{enumerate}
 \end{proof}

\begin{lem}\label{wuhanu} Consider an arbitrary member $h$ of $\mathbb{I}$. 
%\marginpar{wuhanu}
\begin{enumerate}[label=\arabic*.]
\item
\begin{enumerate}
\item There is a run $\Gamma_{0}^{h}$ generated by ${\mathcal H}_0$
  such that $\overline{\mathbb{B}_{0}^{h}}\preceq \Gamma_{0}^{h}$.  
\item Furthermore, if $h$ is the greatest element of $\mathbb{I}$ and
  $\overline{\mathbb{B}_{0}^{h}}$ is a reasonable run of $F'(0)$,
  then, for such a $\Gamma_{0}^{h}$, we simply have
  $\overline{\mathbb{B}_{0}^{h}}=\Gamma_{0}^{h}$.
\end{enumerate}

\item Consider any $n\in\{1,\ldots,k\}$. 
\begin{enumerate}
\item There is a 
run $\Gamma_{n}^{h}$ generated by ${\mathcal H}_n$ such that   $\overline{\mathbb{B}_{n}^{h}\hspace{-2pt}\uparrow}\preceq (\Gamma^{h}_{n})^{1.}$ and $\gneg \overline{\mathbb{B}_{n-1}^{h}\hspace{-2pt}\downarrow}\preceq 
( \Gamma^{h}_{n})^{0.}$. 
\item Furthermore, if $h$ is the greatest element of $\mathbb{I}$,
  $\overline{\mathbb{B}_{n }^{h}}$ is a reasonable run of $F'(n )$ and
  $\overline{\mathbb{B}_{n-1}^{h}}$ is a reasonable run of $F'(n-1)$,
  then, for such a $\Gamma_{n}^{h}$, we simply have
  $\overline{\mathbb{B}_{n}^{h}}= (\Gamma^{h}_{n})^{1.}$ and
  $\gneg\overline{\mathbb{B}_{n-1}^{h}}= (\Gamma^{h}_{n})^{0.}$. 
\end{enumerate}
\end{enumerate}
\end{lem}

\begin{proof} Fix an arbitrary $h\in \mathbb{I}$.\vspace{4pt}
\begin{enumerate}[label=\arabic*.]
\item Let $\mathbb{B}_{0}^{h}=(T_1,\ldots,T_t)$.  

If $t=0$,  then the position $\overline{\mathbb{B}_{0}^{h}}$ is empty, and is thus an initial segment of any run. So, an arbitrarily selected run  $\Gamma^{h}_{0}$  generated by ${\mathcal H}_0$ --- such as, for instance, the run in which Environment made no moves at all --- satisfies subclause (a). As for subclause (b), it is trivially satisfied because, by clause 1 of Lemma \ref{a21}, $h$ is not the greatest element of $\mathbb{I}$, for otherwise $t$ would have to be odd. 

 Now, for the rest of our proof of clause 1, assume $t\geq 1$. This automatically makes $\bigl((),(\mathbb{B}_{0}^{h})^{\odd},0\bigr)$ a $\simm$-appropriate triple. We first claim that
%\marginpar{z27a}
\begin{equation}\label{z27a}
\mbox{\it For any nonempty proper restriction $C$ of  $(\mathbb{B}_{0}^{h})^{\odd}$, 
$\simm_{0}^{\bullet}\bigl((),C\bigr)$ is positive.}
\end{equation}
For a contradiction suppose (\ref{z27a}) if false, and assume that, for some nonempty proper restriction $C$ of $(\mathbb{B}_{0}^{h})^{\odd}$,  
$\simm_{0}^{\bullet}\bigl((),C\bigr)$ is negative. Obviously $C=(T_1,\ldots,T_s)^{\odd}$ for some odd $s$ with $s<t$. Fix such an $s$.
 Thus,
%\marginpar{27aaa}
\begin{equation}\label{27aaa}
\mbox{\it $\simm_{0}^{\bullet}\bigl((),(T_1,\ldots,T_s)^{\odd}\bigr)$ is negative.}
\end{equation}
Let $i$ be the $(h,0)$-birthtime of 
$(T_1,\ldots,T_{s+1})$. This means that $i-1\in \mathbb{I}^{h}_{!}$, the $(i-1)$th iteration of \mainn\ proceeds according to Subcase 2.1, $\vec{E}_{i-1}$ contains the entry $[0,(T_1,\ldots,T_{s})]$ and, with $\bigl((),(T_1,\ldots,T_{s}),0\bigr)$ being the central triple of $\vec{E}_{i-1}$, we have 
$\simm_{0}^{\bullet}\bigl((),(T_1,\ldots,T_{s})^{\odd}\bigr)=+T_{s+1}$. This, however, contradicts (\ref{27aaa}). 
 Claim (\ref{z27a}) is thus verified.  
 
Now we observe that 
%\marginpar{27a}
\begin{equation}\label{27a}
\mbox{\it The triple $\bigl((),(\mathbb{B}_{0}^{h})^{\odd},0\bigr)$ is saturated}.
\end{equation}
Indeed, 
if  $\simm_{0}^{\bullet}\bigl((),(\mathbb{B}_{0}^{h})^{\odd}\bigr)$ is positive, then (\ref{27a}) automatically holds because the empty body $()$ has no proper restrictions; and if 
$\simm_{0}^{\bullet}\bigl((),(\mathbb{B}_{0}^{h})^{\odd}\bigr)$ is negative, then (\ref{27a}) is an immediate consequence of (\ref{z27a}). 

Our next claim is that
%\marginpar{27b}
\begin{equation}\label{27b}
\mbox{\em $\simm_{0}^{\rightarrow}\bigl((),(\mathbb{B}_{0}^{h})^{\odd} \bigr)$ is an extension of $ \mathbb{B}_{0}^{h}$.}
\end{equation}
To justify this claim, assume 
$\simm_{0}^{\rightarrow}\bigl((),(\mathbb{B}_{0}^{h})^{\odd} \bigr)=(W_1,\ldots,W_w)$. From 
the definition of $\simm_{0}^{\rightarrow}$, we have 
$(W_1,\ldots,W_w)^{\odd}=(\mathbb{B}_{0}^{h})^{\odd}$.
So, we only need to show that $(W_1,\ldots,W_w)^{\even}$ is an extension  of $(\mathbb{B}_{0}^{h})^{\even}$, i.e., of $(T_1,\ldots,T_t)^{\even}$.
But indeed, consider any even $r\in\{1,\ldots,t\}$. Let $i$ be the $(h,0)$-birthtime of 
$(T_1,\ldots,T_r)$. This means that $i-1\in\mathbb{I}^{h}_{!}$, the $(i-1)$th iteration of \mainn\ proceeds according to the scenario of Subcase 2.1 where $\vec{E}_{i-1}$ has the entry $[0,(T_1,\ldots,T_{r-1})]$ and, 
with $\bigl((),(T_1,\ldots,T_{r-1}),0\bigr)$ being the central triple of $\vec{E}_{i-1}$,   
$\simm_{0}^{\bullet}\bigl((),(T_1,\ldots,T_{r-1})^{\odd}\bigr)=+T_r$. But how does the computation of 
$\simm_{0}^{\bullet}\bigl((),(T_1,\ldots,T_{r-1})^{\odd}\bigr)$ differ from the computation of 
$\simm_{0}^{\bullet}\bigl((),(T_1,\ldots,T_{t})^{\odd}\bigr)$ (from which the value $(W_1,\ldots,W_w)$ of $\simm_{0}^{\rightarrow}\bigl((),(\mathbb{B}_{0}^{h})^{\odd} \bigr)$ is extracted)? Both computations proceed in exactly the same way, with the variable $S$ of $\simm_{0}^{\bullet}$ going through exactly the same values, with the only difference that, while the computation of  $\simm_{0}^{\bullet}\bigl((),(T_1,\ldots,T_{r-1})^{\odd}\bigr)$ stops after $S$ takes its $(r/2)$th value $+W_r$ and returns that value as $+T_r$, the computation of 
$\simm_{0}^{\bullet}\bigl((),(T_1,\ldots,T_{t})^{\odd}\bigr)$ continues further until that value becomes $+W_w$  (if the output is positive) or $-\bigl((),s\bigr)$ for some $s$ (if the output is negative). Thus $W_r=T_r$, which completes our proof of  claim (\ref{27b}). 

Since, by (\ref{27a}), the triple $\bigl((),(\mathbb{B}_{0}^{h})^{\odd},0\bigr)$ is  saturated, clause 1(a) of Lemma \ref{ap20b} guarantees that there is a run $\Upsilon$ --- let us rename it into $\Gamma_{0}^{h}$ --- generated by ${\mathcal H}_0$ such that 
$\overline{\simm_{0}^{\rightarrow}\bigl((),(\mathbb{B}_{0}^{h})^{\odd} \bigr)}\preceq \Gamma_{0}^{h}$. This, by (\ref{27b}), implies that $\overline{\mathbb{B}_{0}^{h}}\preceq \Gamma_{0}^{h}$, as promised in clause 1(a) of the present lemma. 

For clause 1(b) of the present lemma, let us additionally assume that $h$ is the greatest element of $\mathbb{I}$ and  $\overline{\mathbb{B}_{0}^{h}}$ is a reasonable run of $F'(0)$. Note that the last, $h$th iteration of \mainn\ deals with Subsubsubcase 2.2.2.2, for any other case causes a next iteration to occur. Let $\bigl((),B,0\bigr)$ be the central triple of $\vec{E}_h$. So,   
%\marginpar{iuy}
\begin{equation}\label{iuy}
\simm_{0}\bigl((),B^{\odd}\bigr)=\bigl(-\bigl((),v\bigr),u\bigr)
\end{equation}
 for some numbers $v,u$. Fix these numbers. By definition (\ref{eqdef}), $B$ is a restriction of $\mathbb{B}_{0}^{h}$. And, by clause 1 of Lemma \ref{a21}, the size of $\mathbb{B}_{0}^{h}$ is odd. Consequently, $B$ is not a proper restriction of $\mathbb{B}_{0}^{h}$, because otherwise  $B^{\odd}$ would be a proper restriction of $(\mathbb{B}_{0}^{h})^{\odd}$,  making the statements (\ref{z27a}) and (\ref{iuy}) contradictory. We thus find that $B=\mathbb{B}_{0}^{h}$, which allows us to re-write (\ref{iuy}) as 
%\marginpar{iiuy}
\begin{equation}\label{iiuy}
\simm_{0}\bigl((),(\mathbb{B}_{0}^{h})^{\odd}\bigr)=\bigl(-\bigl((),v\bigr),u\bigr).
\end{equation}
In view of $\simm_{0}^{\bullet}\bigl((),(\mathbb{B}_{0}^{h})^{\odd}\bigr)$'s being negative, one can see immediately from the definition of $\simm_{0}^{\rightarrow}$ that 
the size of $\simm_{0}^{\rightarrow}((),(\mathbb{B}_{0}^{h})^{\odd})$ does not exceed the size of $\mathbb{B}_{0}^{h}$. 
This, in combination with   (\ref{27b}), means that
%\marginpar{fgh}
\begin{equation}\label{fgh}
\simm_{0}^{\rightarrow}\bigl((),(\mathbb{B}_{0}^{h})^{\odd} \bigr)= \mathbb{B}_{0}^{h}.
\end{equation}

Imagine the work of $\simm_{0}$ when computing (\ref{iiuy}). Taking (\ref{27a}) into account, we can see that  $v$ is just a copy of the scale of 
the last organ of $(\mathbb{B}_{0}^{h})^{\odd}$ and hence, by clause 1 of Lemma \ref{a21}, of the last organ of $(\mathbb{B}_{0}^{h})$. Consequently, by clause 3 of Lemma \ref{a21}, $v$ is the master scale of $\vec{E}_h$. Then, since the $h$th iteration of \mainn\ proceeds according to Subsubsubcase 2.2.2.2,   
we have $v\geq \mathfrak{L}(\mathfrak{l},U_h)$, where $U_h$ is the value that the variable $U$ of \mainn\ assumes on the $h$th iteration as a result of updating the old value to $\max(u,U)$. We thus have $u\leq U_h$. And the function  
$\mathfrak{L}$ is, of course, monotone. Consequently, from the fact $v\geq \mathfrak{L}(\mathfrak{l},U_h)$, we find that  
$v\geq \mathfrak{L}(\mathfrak{l},u)$. But then, by (\ref{fgh}) and clause 1(b) of Lemma \ref{ap20b}, there is a run $\Upsilon$ generated by ${\mathcal H}_0$ --- let us rename it into $\Gamma^{h}_{0}$ --- such that 
 $\overline{\mathbb{B}_{0}^{h}}= \Gamma_{0}^{h}$. Done.\vspace{4pt}

\item%{\em Clause 2}.  
Fix any $n\in\{1,\ldots,k\}$, and assume  
\[\begin{array}{ll}
\mathbb{B}_{n-1}^{h}\hspace{-2pt}\downarrow =(P_1,\ldots,P_p); \ \ & \ \ \mathbb{B}_{n}^{h}\hspace{-2pt}\uparrow =(Q_1,\ldots,Q_q);\\
 \mathbb{B}_{n-1}^{h} =(P_1,\ldots,P_{p'}); \ \  \ & \ \  \mathbb{B}_{n}^{h} =(Q_1,\ldots,Q_{q'}).
\end{array}\]

For clause 2(a),  we  want to show the existence of a run $\Gamma_{n}^{h}$ generated by ${\mathcal H}_n$ such that 
%\marginpar{cx1}
\begin{equation}\label{cx1}
\mbox{\it  $\overline{(Q_1,\ldots,Q_{q})}\preceq (\Gamma^{h}_{n})^{1.}$ \ and \ $\gneg \overline{(P_1,\ldots,P_{p})}\preceq (\Gamma^{h}_{n})^{0.}$}
\end{equation} 

It is not hard to see that, if $q$ is $0$, then so is $p$, because there is no way for $(P_1)$ to be ever $(h,n-1)$-born. Then 
 the runs $\overline{(P_1,\ldots,P_{p})}$ and $\overline{(Q_1,\ldots,Q_{q})}$ are empty and, therefore, any run $\Gamma_{n}^{h}$ generated by ${\mathcal H}_n$ satisfies  (\ref{cx1}). Now, for the rest of this proof, assume $q$ is non-zero, which, in view of $q$'s being even, means that $q\geq 2$. 
In what follows, we use $i_P$  to denote the $(h,n-1)$-birthtime of $(P_1,\ldots,P_{p})$  and  $i_Q$ to denote the $(h,n)$-birthtime of $(Q_1,\ldots,Q_{q})$. We claim that
%\marginpar{14aa}
\begin{equation}\label{14aa}
i_P\not=i_Q.
\end{equation}
Indeed, it is easy to see that two bodies have identical birthtimes only if they  are both empty (and hence their birthtimes are both $0$). However, as we have already agreed,  $(Q_1,\ldots,Q_{q})$ is nonempty. 
In view of (\ref{14aa}), it is now sufficient to consider the  two cases $i_Q>i_P$ and $i_P>i_Q$.

{\em Case of $i_Q>i_P$}:  In this case, according to clause 1 of Lemma \ref{wuhano},  the triple 
\[\bigl((P_1,\ldots,P_{p'})^{\even}, (Q_1,\ldots,Q_{q-1})^{\odd},n\bigr)\] is  saturated, and we have: 
\[\begin{array}{l}
\simm_{n}^{\rightarrow} \bigl((P_1,\ldots,P_{p'})^{\even}, (Q_1,\ldots,Q_{q-1})^{\odd}\bigr)=(Q_1,\ldots,Q_q);\\
\simm_{n}^{\leftarrow} \bigl((P_1,\ldots,P_{p'})^{\even}, (Q_1,\ldots,Q_{q-1})^{\odd}\bigr)=(P_1,\ldots,P_{p'}).
\end{array}\]
 Therefore, by clause 2(a) of Lemma \ref{ap20b}, there is a run $\Upsilon$ --- let us rename it into $\Gamma^{h}_{n}$ ---  generated by ${\mathcal H}_n$ such that  
$\overline{(Q_1,\ldots,Q_{q})}\preceq (\Gamma^{h}_{n})^{1.}$ and $\gneg\overline{(P_1,\ldots,P_{p'})}\preceq (\Gamma^{h}_{n})^{0.}$. Of course, 
$\gneg\overline{(P_1,\ldots,P_{p'})}\preceq (\Gamma^{h}_{n})^{0.}$ implies $\gneg\overline{(P_1,\ldots,P_{p})}\preceq (\Gamma^{h}_{n})^{0.}$. So, (\ref{cx1}) holds, which takes care of clause 2(a) of the present lemma. 
 As for clause 2(b), it is satisfied vacuously because $h$ is not the greatest element of $\mathbb{I}$. To see why 
$h$ is not the greatest element of $\mathbb{I}$, assume the opposite. Let $i_{P'}$ be the $(h,n-1)$-birthtime of $(P_1,\ldots,P_{p'})$ and $i_{Q'}$ be the $(h,n)$-birthtime of $(Q_1,\ldots,Q_{q'})$. By clause 1 of Lemma \ref{a21}, $p$ is odd, implying that $p'=p$ and hence $i_{P'}=i_{P}$. Next, the fact $q'\geq q$ obviously implies that $i_{Q'}\geq i_{Q}$. So, the condition $i_Q>i_P$ of the present case implies $i_{Q'}>i_{P'}$. But this is in contradiction with clause 2 of Lemma \ref{a21}. 

{\em Case of $i_P>i_Q$}:  In this case, according to clause 2 of Lemma \ref{wuhano}, we have:
%\marginpar{tod1-tod4}
\begin{eqnarray}
& & \simm_{n}^{\rightarrow} \bigl((P_1,\ldots,P_{p-1})^{\even}, (Q_1,\ldots,Q_{q'})^{\odd}\bigr)=-P_p;\label{tod1}\\
& & \mbox{\em The triple $\bigl((P_1,\ldots,P_{p-1})^{\even}, (Q_1,\ldots,Q_{q'})^{\odd},n\bigr)$
 is saturated;}\label{tod2}\\
& & \simm_{n}^{\rightarrow} \bigl((P_1,\ldots,P_{p-1})^{\even}, (Q_1,\ldots,Q_{q'})^{\odd}\bigr)=(Q_1,\ldots,Q_{q'});\label{tod3}\\
& & \simm_{n}^{\leftarrow} \bigl((P_1,\ldots,P_{p-1})^{\even}, (Q_1,\ldots,Q_{q'})^{\odd}\bigr)=(P_1,\ldots,P_{p}).\label{tod4}
\end{eqnarray}
From (\ref{tod2})-(\ref{tod4}),  by clause 2(a) of Lemma \ref{ap20b} with $(P_1,\ldots,P_{p-1})^{\even}$ in the role of $A$ and $(Q_1,\ldots,Q_{q'})^{\odd}$ in the role of $B$, there is a run $\Upsilon$ --- let us rename it into $\Gamma^{h}_{n}$ --- generated by ${\mathcal H}_n$ such that  
$\overline{(Q_1,\ldots,Q_{q'})}\preceq (\Gamma^{h}_{n})^{1.}$ and $\gneg\overline{(P_1,\ldots,P_{p})}\preceq (\Gamma^{h}_{n})^{0.}$. But 
$\overline{(Q_1,\ldots,Q_{q'})}\preceq (\Gamma^{h}_{n})^{1.}$ implies $\overline{(Q_1,\ldots,Q_{q})}\preceq (\Gamma^{h}_{n})^{1.}$. So, (\ref{cx1}) holds, which takes care of clause 2(a) of the present lemma. 

For clause 2(b), let us assume moreover that $h$ is the greatest element of $\mathbb{I}$,    $\overline{(Q_1,\ldots,Q_{q'})}$ is a reasonable run of $F'(n )$, and  $\gneg \overline{(P_1,\ldots,P_{p'})}$ is a reasonable run of $F'(n-1)$.  
 By clause 1 of Lemma \ref{a21}, $p'$ is odd, implying that $p=p'$. 
So, (\ref{tod1})-(\ref{tod4}) can be re-written as
%\marginpar{todd1-todd4}
\begin{eqnarray}
& & \simm_{n}^{\bullet} \bigl((P_1,\ldots,P_{p'-1})^{\even}, (Q_1,\ldots,Q_{q'})^{\odd}\bigr)=-P_{p'};\label{todd1}\\
& & \mbox{\em The triple $\bigl((P_1,\ldots,P_{p'-1})^{\even}, (Q_1,\ldots,Q_{q'})^{\odd},n\bigr)$
 is saturated;}\label{todd2}\\
& & \simm_{n}^{\rightarrow} \bigl((P_1,\ldots,P_{p'-1})^{\even}, (Q_1,\ldots,Q_{q'})^{\odd}\bigr)=(Q_1,\ldots,Q_{q'});\label{todd3}\\
& & \simm_{n}^{\leftarrow} \bigl((P_1,\ldots,P_{p'-1})^{\even}, (Q_1,\ldots,Q_{q'})^{\odd}\bigr)=(P_1,\ldots,P_{p'}).\label{todd4}
\end{eqnarray}
Let $P_{p'}=(\vec{\omega},v)$. In view of (\ref{todd1}), there is a number $u$ (fix it) such that  
%\marginpar{toddd}
\begin{equation}\label{toddd}
\simm_{n} \bigl((P_1,\ldots,P_{p'-1})^{\even}, (Q_1,\ldots,Q_{q'})^{\odd}\bigr)=\bigl(-(\vec{\omega},v),u\bigr).
\end{equation}

As observed earlier when verifying clause 2(b) of the lemma in the case of $i_Q>i_P$, we have $p=p'$, meaning that $i_P$ is the $(h,n-1)$-birthtime of $\mathbb{B}_{n-1}^{h}=(P_1,\ldots,P_{p'})$. In addition,  let $i_L$ be the $(h,k)$-birthtime of $\mathbb{B}_{k}^{h}$. By clause 2 of Lemma \ref{a21}, $i_P>i_L$. This means that, for any $j\in\{i_P,\ldots,h\}$, the $j$th iteration of \mainn\ is not locking, because a locking iteration always gives birth to a new, ``bigger''  master body.   But the absence of locking iterations between $i_P$ and $h$ implies the following, because otherwise $i_P$ would be $h$-transient:
%\marginpar{fgy}
\begin{equation}\label{fgy}
\mbox{\em For any $j\in\{i_P,\ldots,h\}$, the $j$'th iteration of \mainn\ is not restarting.}
\end{equation}

Since $h$ is the greatest element of $\mathbb{I}$, according to clause 3 of Lemma \ref{a21}, $v$ is the master scale of 
$\vec{E}_{h}$. Also, 
 as observed earlier in the proof of clause 1(b), the $h$th iteration of \mainn\ deals with Subsubsubcase 2.2.2.2, implying that $v\geq \mathfrak{L}(\mathfrak{l},U_h)$, where $U_h$ is the final value of the variable $U$ of \mainn\ (assumed on the $h$th iteration). But note that $U_{i_P}$ --- the value of $U$ assumed on the $i_P$th iteration of \mainn\ --- does not exceed $U_h$. That is because only restarting iterations of \mainn\ can decrease the value of $U$, but, by (\ref{fgy}), there are no such iterations between $i_P$ and $h$. Also, it is clear that, on the $i_P$th iteration, $(P_1,\ldots,P_{p'})$ was  born  according to the scenario of Subsubcase 2.2.1 due to (\ref{toddd}), implying that $U_{i_P}\geq u$, because, at the beginning of that iteration, the variable $U$ was updated to $U_{i_P}=\max(u,U)$.  Thus, $U_h\geq u$ and hence, due to the monotonicity of $\mathfrak{L}$ and the earlier-established fact $v\geq \mathfrak{L}(\mathfrak{l},U_h)$, we have 
%\marginpar{tod0}
\begin{equation}\label{tod0}
v\geq \mathfrak{L}(\mathfrak{l},u).
\end{equation}
  From (\ref{todd2}), (\ref{toddd}),  (\ref{todd3}), (\ref{todd4}) and (\ref{tod0}),  using clause 2(b) of Lemma \ref{ap20b}, with $(P_1,\ldots,P_{p'-1})^{\even}$ in the role of $A$ and $(Q_1,\ldots,Q_{q'})^{\odd}$ in the role of $B$, there is a run $\Upsilon$ --- let us rename it into $\Gamma^{h}_{n}$ --- such that  
  $\overline{(Q_1,\ldots,Q_{q'})}= (\Gamma^{h}_{n})^{1.}$ and
  $\gneg\overline{(P_1,\ldots,P_{p'})}= (\Gamma^{h}_{n-1})^{0.}$, as
  desired.\qedhere
\end{enumerate}
\end{proof}

\begin{lem}\label{shan}
%\marginpar{shan}
For every $n\in\{0,\ldots,k\}$ and every $h\in \mathbb{I}$, $\overline{\mathbb{B}_{n}^{h}}$ is a reasonable run of $F'(n)$.
\end{lem}

\begin{proof} Fix an $n\in\{0,\ldots,k\}$ and an $h\in \mathbb{I}$. Below, whenever we say that a player $\xx$ has made --- or is responsible for making --- a given run unreasonable, it is to be (or, at least, can be) understood as that the last move of the shortest unreasonable initial segment of the run is $\xx$-labeled. 

First, consider the case  $n=0$. For a contradiction, assume $\overline{\mathbb{B}_{0}^{h}}$ is not a reasonable run of $F'(0)$. By clause 1(a) of Lemma \ref{wuhanu},  $\overline{\mathbb{B}_{0}^{h}}$ is an initial segment of a certain run $\Gamma_{0}^{h}$ generated by ${\mathcal H}_0$. Therefore, in view of our assumption that ${\mathcal H}_0$ plays $F'(0)$ reasonably, the only way $\overline{\mathbb{B}_{0}^{h}}$ could be unreasonable    is if $\oo$ (${\mathcal H}_0$'s adversary) made it so. But, according to clause 2(a) of Lemma \ref{wuhanu}, a certain extension $(\Gamma_{1}^{h})^{0.}$ of 
$\gneg \overline{\mathbb{B}_{0}^{h}\hspace{-2pt}\downarrow}$ is a run generated by ${\mathcal H}_{1}$ (with ${\mathcal H}_{1}$ playing as $\pp$) in the component $\gneg F'(0)$ of $\gneg F'(0)\mld F'(1)$. Therefore, as --- by our assumption --- ${\mathcal H}_1=\top$ plays reasonably, $\top$ cannot be responsible for making $\gneg \overline{\mathbb{B}_{0}^{h}\hspace{-2pt}\downarrow}$ unreasonable. Then $\pp$ cannot be responsible for $\gneg \overline{\mathbb{B}_{0}^{h}}$'s being unreasonable either, because   
$\gneg \overline{\mathbb{B}_{0}^{h}}$ differs from  $\gneg \overline{\mathbb{B}_{0}^{h}\hspace{-2pt}\downarrow}$  only in that the former perhaps has some additional $\oo$-labeled moves at the end. Contradiction.

Next, consider the case $0< n< k$. It is rather similar to the preceding one. For a contradiction, assume $\overline{\mathbb{B}_{n}^{h}}$ is not a reasonable run of $F'(n)$. 
By clause 2(a) of Lemma \ref{wuhanu}, there is a run $\Gamma^{h}_{n}$ generated by ${\mathcal H}_n$ such that  $\overline{\mathbb{B}_{n}^{h}\hspace{-2pt}\uparrow}$ is an initial segment of  $(\Gamma_{n}^{h})^{1.}$.   $\overline{\mathbb{B}_{n}^{h}}$ only differs from $\overline{\mathbb{B}_{n}^{h}\hspace{-2pt}\uparrow}$ in that the former perhaps has some additional $\oo$-labeled moves at the end. For this reason, as ${\mathcal H}_n$ plays $F'(n-1)\mli F'(n)$ reasonably, 
the only way $\overline{\mathbb{B}_{n}^{h}}$ could be unreasonable  is if $\oo$ (${\mathcal H}_n$'s adversary) made it so. Then $\gneg\overline{\mathbb{B}_{n}^{h}}$ is an unreasonable   run of $\gneg F'(n)$, with player $\pp$ 
 being responsible for making it so. But, (again) according to clause 2(a) of Lemma \ref{wuhanu}, a certain extension $(\Gamma_{n+1}^{h})^{0.}$ of $\gneg \overline{\mathbb{B}_{n}^{h}\hspace{-2pt}\downarrow}$ is a run generated by ${\mathcal H}_{n+1}$  in the component $\gneg F'(n)$ of $\gneg F'(n)\mld F'(n+1)$. As ${\mathcal H}_{n+1}=\pp$  plays this game reasonably, it cannot be responsible for making  $\gneg \overline{\mathbb{B}_{n}^{h}\hspace{-2pt}\downarrow}$ an unreasonable run of $\gneg F'(n)$. Then $\pp$ cannot be responsible for making 
$\gneg\overline{\mathbb{B}_{n}^{h}}$ unreasonable either, because $\gneg \overline{\mathbb{B}_{n}^{h}}$ only differs from   $\gneg \overline{\mathbb{B}_{n}^{h}\hspace{-2pt}\downarrow}$   
 in that the former perhaps has some additional $\oo$-labeled moves at the end.  Contradiction.

Finally, consider the case $n=k$. Just as in the preceding cases, ${\mathcal H}_k$ cannot be responsible for making $\overline{\mathbb{B}_{k}^{h}}$ an unreasonable run of $F'(k)$. Looking at Case 1, Subsubcase 2.1.2 and Subsubsubcase 2.2.2.2 of the description of $\mainn$, it is clear that ${\mathcal H}_k$'s imaginary adversary does not make $\overline{\mathbb{B}_{k}^{h}}$  unreasonable either. This is so because, in $F'(k)$,  $\mainn$ lets ${\mathcal H}_k$'s  adversary mimic  ${\mathcal M}_k$'s real environment's play. The latter, by 
the Clean Environment Assumption,   
 plays  (legally and hence) quasilegally. And even if it does not play prudently,  
$\mainn$ prudentizes  ${\mathcal M}_k$'s environment's moves  before copying and adding them to $\overline{\mathbb{B}_{k}^{h}}$ as ${\mathcal H}_k$'s imaginary adversary's moves. 
\end{proof}

Recall that $\hbar$ is the greatest element of $\mathbb{I}$.

\begin{lem}\label{shantou}
%\marginpar{shantou}
For every $n\in\{0,\ldots,k\}$, $\overline{\mathbb{B}_{n}^{\hbar}}$ is  a $\pp$-won run of $F'(n)$.
\end{lem}

\begin{proof} Induction on $n$. According to clause 1(b) of Lemma \ref{wuhanu}, in conjunction with Lemma \ref{shan}, $\overline{\mathbb{B}_{0}^{\hbar}}$ is a run generated by ${\mathcal H}_0$. So, since ${\mathcal H}_0$ wins $F'(0)$, $\overline{\mathbb{B}_{0}^{\hbar}}$ is a $\pp$-won run of $F'(0)$.

Next, consider any $n$ with $0< n\leq k$. According to clause 2(b) of Lemma \ref{wuhanu}, in conjunction with Lemma \ref{shan}, there is a run $\Gamma_{n}^{\hbar}$ generated by ${\mathcal H}_n$ such that $(\Gamma_{n}^{\hbar})^{0.}=\gneg \overline{\mathbb{B}_{n-1}^{\hbar}}$ and $(\Gamma_{n}^{\hbar})^{1.}=\overline{\mathbb{B}_{n}^{\hbar}}$. Note that, since ${\mathcal H}_n$ plays quasilegally, every move of $\Gamma_{n}^{\hbar}$ has one of the two prefixes ``$0.$'' or ``$1.$''. But we know that ${\mathcal H}_n$ wins $\gneg F'(n-1)\mld F'(n)$. So, $\Gamma_{n}^{\hbar}$ has to be a $\pp$-won run of  $\gneg F'(n-1)\mld F'(n)$, meaning that either $(\Gamma_{n}^{\hbar})^{0.}$, i.e.   $\gneg \overline{\mathbb{B}_{n-1}^{\hbar}}$, is a $\pp$-won run of $\gneg F'(n-1)$, or $(\Gamma_{n}^{\hbar})^{1.}$, i.e. $\overline{\mathbb{B}_{n}^{\hbar}}$, is a $\pp$-won run of $F'(n)$.  But, by the induction hypothesis, $\overline{\mathbb{B}_{n-1}^{\hbar}}$ is a $\pp$-won run of $F'(n-1)$. This obviously means that  $\gneg \overline{\mathbb{B}_{n-1}^{\hbar}}$ is a $\oo$-won (and thus not $\pp$-won) run of $\gneg F'(n-1)$. Therefore, $\overline{\mathbb{B}_{n}^{\hbar}}$ is a $\pp$-won run of $F'(n)$. 
\end{proof}

According to Lemma \ref{shantou},  $\overline{\mathbb{B}_{k}^{\hbar}}$ is a $\pp$-won run of $F'(k)$. Therefore, by the known property of static games and delays (see the end of Section 3 of \cite{cl12}) we have:
%\marginpar{trr}
\begin{equation}\label{trr}
\mbox{\it Whenever a run $\Pi$ is a $\pp$-delay of $\overline{\mathbb{B}_{k}^{\hbar}}$, $\Pi$ is a $\pp$-won run of $F'(k)$.}
\end{equation}

Let $\Theta$ be the run generated by ${\mathcal M}_k$ that took place in the real play of $k\leq \mathfrak{b}|\vec{d}|\mli F'(k)$. How does $\Theta^{1.}$ relate to $\overline{\mathbb{B}_{k}^{\hbar}}$? As promised earlier, the real play in the consequent of $k\leq \mathfrak{b}|\vec{d}|\mli F'(k)$ --- that is, the run $\Theta^{1.}$ --- would be ``essentially synchronized'' with the play $\overline{\mathbb{B}_{k}^{\hbar}}$  by ${\mathcal H}_k$ in the consequent of $F'(k-1)\mli F'(k)$, meaning that $\Theta^{1.}$ is ``essentially the same'' as $\overline{\mathbb{B}_{k}^{\hbar}}$. The qualification ``essentially'' implies that the two runs, while being similar, may not necessarily be strictly identical. 

One reason why $\overline{\mathbb{B}_{k}^{\hbar}}$ may differ from $\Theta^{1.}$ is that, as seen from Case 1 and Subsubsubcase 2.2.2.2 of the description of $\mainn$, if $\Theta^{1.}$ contains a (legal but)  imprudent (with respect to $F'(k)$) move by $\oo$, then such a move appears in 
$\overline{\mathbb{B}_{k}^{\hbar}}$ in the prudentized  form. Namely, if ${\mathcal H}_k$'s adversary chose some ``oversized'' constant $a$ for $z$ in a subcomponent $\ada z G$ of $F'(k)$, then the same move will appear in $\overline{\mathbb{B}_{k}^{\hbar}}$ as if $a'$ was chosen instead of $a$, where $a'$ is a certain ``small'' constant. Note, however, that having made the above imprudent choice makes $\oo$ lose in the $\ada z G$ component. So, prudentizing $\oo$'s imprudent moves can only increase rather than decrease $\oo$'s chances to win the overall game. That is, if $\pp$ (i.e., ${\mathcal M}_k$) wins the game even after such moderation of the adversary's imprudent moves,  it would also win (``even more so'') without moderation.  For this reason, we can and will safely assume that ${\mathcal M}_k$'s environment plays not only legally, but also prudently.

But even if ${\mathcal M}_k$'s adversary has played $\Theta^{1.}$ prudently, there is another reason  that could make $\overline{\mathbb{B}_{k}^{\hbar}}$ ``somewhat'' different from $\Theta^{1.}$. Namely, with some thought, one can see that $\Theta^{1.}$ may be a proper $\pp$-delay of (rather than equal to) $\overline{\mathbb{B}_{k}^{\hbar}}$. Luckily, however, by (\ref{trr}), $\Theta^{1.}$ is still a $\pp$-won run of $F'(k)$.

Thus, as desired, ${\mathcal M}_k$ wins $k\leq \mathfrak{b}|\vec{d}|\mli F'(k)$, and hence $\mathcal M$ wins the conclusion of (\ref{r22}).  

\subsection{\texorpdfstring{${\mathcal M}$}{M} runs in target amplitude}\label{sla}
%\marginpar{sla}

$\mathcal M$ plays  $x\leq\mathfrak{b}|\vec{s}|\mli F(x,\vec{v})$ prudently, and the latter is an 
 ${\mathcal R}\spa$-bounded formula. By condition 5 of Definition 2.2 of \cite{AAAI}, ${\mathcal R}\spa\preceq {\mathcal R}\amp$. This, of course, implies that   $\mathcal M$  runs in amplitude ${\mathcal R}\amp$, as desired. 

\subsection{\texorpdfstring{${\mathcal M}$}{M} runs in target space}\label{sls}
%\marginpar{sls}
As we agreed earlier, $(\mathfrak{a},\mathfrak{s},\mathfrak{t})\in{\mathcal R}\amp\times {\mathcal R}\spa\times {\mathcal R}\tim$ is a  common tricomplexity in which the machines $\mathcal N$ and $\mathcal K$ --- and hence the ${\mathcal H}_n$s --- run. All three bounds are unary. 

Remember from Section \ref{sagree}  that $\mathfrak{l}$ is the size of the greatest of the constants chosen by  ${\mathcal M}$'s environment for the free variables of $x\leq \mathfrak{b}|\vec{s}|\mli F(x,\vec{v})$. This, of course, means that the background of any clock cycle of ${\mathcal M}_k$ in any scenario of its work 
 is at least 
$\mathfrak{l}$. For this reason and with Remark 2.4 of \cite{AAAI} in mind,  in order to show that ${\mathcal M}$ runs in  space ${\mathcal R}\spa$, it is sufficient to show that    the spacecost of  any clock cycle of ${\mathcal M}_k$ is  bounded by $O\bigl(\mathfrak{p}(\mathfrak{l})\bigr)$ for some $\mathfrak{p}(z)\in{\mathcal R}\spa$. In what follows, we shall write ${\mathcal R}\spa(\mathfrak{l})$ as an abbreviation of the phrase ``$O\bigl(\mathfrak{p}(\mathfrak{l})\bigr)$ for some $\mathfrak{p}(z)\in{\mathcal R}\spa$''. Similarly for 
${\mathcal R}\tim(\mathfrak{l})$.

In asymptotic terms, the space consumed by ${\mathcal M}_k$ --- namely, by any given $h$th ($h\in \mathbb{I}$) iteration of $\mainn$ --- is the sum of the following two quantities:
%\marginpar{spc1,spc2}
\begin{eqnarray}
& \mbox{\it the space needed to hold the value of the aggregation $\vec{E}$;} & \label{spc1}\\
& \mbox{\it the space needed to update $\vec{E}=\vec{E}_h$ to $\vec{E}=\vec{E}_{h+1}$  (if $(h+ 1)\in \mathbb{I}$).}  & \label{spc2}
\end{eqnarray}
Here we did not mention the space needed to hold the value of the variable $U$,  and  to update it to its next value.
 That is because, as it is easy to see,  the space taken by $U$ or its updates does not exceed the maximum possible value of the quantity (\ref{spc2}) (in fact, the logarithm of the latter). So, this component of ${\mathcal M}_k$'s space consumption, being superseded by another component,  
 can be safely ignored in an asymptotic analysis. Consider any $h\in\mathbb{I}$.

In verifying that  (\ref{spc1}) is bounded by ${\mathcal R}\spa(\mathfrak{l})$, 
  we  observe that, by conditions  (iv) and (v) of Section \ref{saggr}, an aggregation cannot contain two 
same-size entries. Next, by Lemma \ref{bei}, the size of an entry never exceeds $2\mathfrak{e}_\top+ 1$. Thus, the 
number of entries in $\vec{E}_h$  is bounded by the constant $2\mathfrak{e}_\top+ 1$. For this reason, it is  sufficient 
for us to just show  that  any given entry $[n,C]$ of $\vec{E}_h$ can be held with ${\mathcal R}\spa(\mathfrak{l})$ space. This is done in the following two paragraphs.

The component $n$ of an entry $[n,C]$ never exceeds $k$. As observed in the proof of Lemma \ref{beijing}, we have 
$k\leq f(\mathfrak{l})$, where $f(z)$ is the unarification of $\mathfrak{b}$.  As further observed near the end of the same proof,  $f(z)\preceq{\mathcal R}\tim$. This, by condition 5 
of Definition 2.2 of \cite{AAAI}, implies that $|f(z)|\preceq {\mathcal R}\spa$. So, $|n|$, which asymptotically is the amount of space needed 
to hold $n$, is bounded by ${\mathcal R}\spa(\mathfrak{l})$.

 As for the component $C$ of an entry $[n,C]$, it is a restriction of (and hence not ``bigger'' than) $\mathbb{B}_{n}^{h}$, so let us talk about $\mathbb{B}_{n}^{h}$ instead. Let $\mathbb{B}_{n}^{h}=\bigl((\vec{\alpha}_1,p_1),\ldots,(\vec{\alpha}_m,p_m)\bigr)$.  By Lemma \ref{shan},  $\overline{\mathbb{B}_{n}^{h}}$  is a reasonable run of $F'(n)$. Consequently, the overall number of moves in it cannot exceed the constant bound $\mathfrak{e}$. Remembering the work of $\simm_{n}^{\bullet}$, we see that only negative values of this procedure's  output may have empty payloads. With this fact in mind, a look back at the work of \mainn\ --- its Subcase 2.1 in particular --- easily reveals that, for each even $i\in\{2,\ldots,m\}$, $\vec{\alpha}_i$ is nonempty. Therefore $m\leq 2\mathfrak{e}+1$. That is, the number of organs in $\mathbb{B}_{n}^{h}$ is bounded by a constant.   So, asymptotically,
$\mathbb{B}_{n}^{h}$ does not take more space than any organ $(\vec{\alpha}_i,p_i)$ of it, which allows us now to just focus on 
$(\vec{\alpha}_i,p_i)$. Due to  $\overline{\mathbb{B}_{n}^{h}}$'s being reasonable, there is only a constant ($\leq \mathfrak{e}$) number of moves  in the payload $\vec{\alpha}_i$ of 
$(\vec{\alpha}_i,p_i)$, and the size of no such move exceeds $O\bigl(\mathfrak{G}(\mathfrak{l})\bigr)$, where $\mathfrak{G}$, as we remember,  is the superaggregate bound of the formula $F(x,\vec{v})$ and hence, by Lemma \ref{lagg}, $\mathfrak{G}\preceq{\mathcal R}\spa$. Thus, ${\mathcal R}\spa(\mathfrak{l})$ space is sufficient to record $\vec{\alpha}_i$.  It now remains to show that the same holds for  $p_i$. 
An analysis of  \mainn\  reveals that, during its work, the only case when a new scale (as opposed to an old, inherited scale) greater than $1$ of whatever organ of whatever entry is ever created is Subsubsubcase 2.2.2.1, and when such a creation takes place, the new scale is smaller than $2\mathfrak{L}(\mathfrak{l},U)$. As observed earlier in this proof when we agreed to ignore $U$, the value of $U$ is bounded by $\mathfrak{s}'(\mathfrak{l})$ for some $\mathfrak{s}'\in {\mathcal R}\spa$. So,  
$p_i<2\mathfrak{L}(\mathfrak{l},\mathfrak{s}'(\mathfrak{l}))$ and thus $|p_i|\leq |2\mathfrak{L}(\mathfrak{l},\mathfrak{s}'(\mathfrak{l}))|$. In view of our earlier observation (\ref{2244}), 
$|2\mathfrak{L}(\mathfrak{l},\mathfrak{s}'(\mathfrak{l}))|=O(|\mathfrak{l}|+|\mathfrak{G}(\mathfrak{l})|+\mathfrak{s}'(\mathfrak{l}))$.   This fact, in conjunction with $\mathfrak{G}\in{\mathcal R}\spa$ and condition  2   of Definition 2.2 of \cite{AAAI}, implies that $|p_i|$, which asymptotically is the amount  of memory needed to hold $p_i$, does not exceed 
${\mathcal R}\spa(\mathfrak{l})$.

Now about the quantity (\ref{spc2}).  Let us only consider the case $n>0$, with the case $n=0$ being similar but simpler. Updating $\vec{E}_h$ to $\vec{E}_{h+1}$ happens through running $\simm_n(A^{\even},B^{\odd})$, where  $(A,B,n)$ is the central triple of $\vec{E}_h$. So, we just need to show that the space consumed by $\simm_n(A^{\even},B^{\odd})$ is bounded by ${\mathcal R}\spa(\mathfrak{l})$.  
This quantity, with asymptotically irrelevant technicalities suppressed,  is the sum of (I)  the space needed for simulating ${\mathcal H}_n$, and (II) the space needed for maintaining (the contents of) the variables $a,b,u,\vec{\psi},\vec{\nu},W,S,R$ of $\simm_n$, as well as the space needed to keep track of how many steps of ${\mathcal H}_n$ have been simulated so far within the present iteration of 
$\loopp_n$.

(I):  There are two groups of moves on the simulated ${\mathcal H}_n$'s run tape. The first group, that we here shall refer to as the 
{\bf early moves},\label{xem} comprises the $\oo$-labeled moves  signifying the initial choices of the constants $n-1$ and $\vec{c}$ for the free variables $x$ and $\vec{v}$ of $F(x,\vec{v})\mli F(x\successor,\vec{v})$. All other moves constitute the second group, which we shall refer to as the {\bf late moves}.\label{xlm} The  information that ${\mathcal M}_k$ needs to keep track of (and whose size is asymptotically relevant) in order to simulate ${\mathcal H}_n$ consists of  the contents (here also including the scanning head locations) of ${\mathcal H}_n$'s run and work tapes, and the content of ${\mathcal H}_n$'s buffer.  
Since $(A,B,n)$ is the central triple of $\vec{E}_h$, $A$ is a restriction of $\mathbb{B}_{n-1}^{h}$ and $B$ is a restriction of $\mathbb{B}_{n}^{h}$. This, in view of Lemma \ref{shan} and in view of ${\mathcal H}_n$'s playing reasonably, obviously implies that 
 the run spelled by the simulated  ${\mathcal H}_n$'s run tape is reasonable. As a result,  there is only a constant number of late moves, and the magnitude of each such move is obviously bounded by  $\mathfrak{G}(\mathfrak{l})$. In view of Lemma \ref{lagg}, this means that the combined size of all late moves is bounded by ${\mathcal R}\spa(\mathfrak{l})$. Since ${\mathcal H}_n$ is unconditionally provident, everything written in its buffer will sooner or later mature into a late move, so, whatever we said about the sizes of the late moves, also applies to the maximum possible size of ${\mathcal H}_n$'s buffer content.   As for the early moves,
they reside on ${\mathcal M}_k$'s own run tape, and hence ${\mathcal M}_k$ does not need to utilize any of its work-tape space to keep track of them. Thus, keeping track of the contents of ${\mathcal H}_n$'s imaginary run tape and buffer does not take ${\mathcal M}_k$ beyond the target ${\mathcal R}\spa(\mathfrak{l})$ space limits. It remains to see that the same holds for the contents of ${\mathcal H}_n$'s work tapes. But indeed, the magnitude of no (early or late) move on ${\mathcal H}_n$'s imaginary run tape exceeds 
$\max(\mathfrak{l},\mathfrak{G}(\mathfrak{l}))$ and hence (as ${\mathcal R}\amp$ is linearly closed and $\mathfrak{G}\in{\mathcal R}\spa\preceq{\mathcal R}\amp$) $\mathfrak{a}'(\mathfrak{l})$ for some $\mathfrak{a}'\in {\mathcal R}\amp$.  But then, since ${\mathcal H}_n$ runs in unconditional space $\mathfrak{s}\in{\mathcal R}\spa$, it consumes at most $\mathfrak{s}(\mathfrak{a}'(\mathfrak{l}))$ space of its work tapes. ${\mathcal M}_k$ can  keep track of the contents of those tapes with asymptotically the same amount 
 $\mathfrak{s}(\mathfrak{a}'(\mathfrak{l}))$ of its own work-tape space. And the latter, by condition 4 of  Definition 2.2 of \cite{AAAI}, does not exceed ${\mathcal R}\spa(\mathfrak{l})$. 

(II): The sizes of the variables $a$ and $b$ of $\simm$ are bounded by a constant (namely, $|2\mathfrak{e}+1|$). As for the sizes of the remaining variables $u,\vec{\psi},\vec{\nu},W,S,R$, as well as the space  needed to keep track of how many steps of ${\mathcal H}_n$ have been simulated so far within the present iteration of 
$\loopp_n$, can be easily seen to be superseded by  (\ref{spc1}) or (I).

\subsection{\texorpdfstring{${\mathcal M}$}{M} runs in target time}\label{slt}
%\marginpar{slt}

We agree that, for an  $h\in \mathbb{I}$, \ $\mathbb{I}^{h}_{\bullet}$\label{xdea} denotes the set of all numbers $i\in\mathbb{I}^{h}$ satisfying the condition that  there is no $j$ with $i\leq j<h$ such that the $j$th iteration of \mainn\ proceeds according to the scenario of Case 1 or Subsubsubcase 2.2.2.2. Next, $\mathbb{I}^{h}_{\bullet\bullet}$\label{xdea7} 
denotes the set of all numbers $i\in\mathbb{I}^{h}_{\bullet}$ (additionally) satisfying the condition that  there is no $j$ with $i\leq j<h$ such that the $j$th iteration of \mainn\ proceeds according to the scenario of   Subsubsubcase 2.2.2.1. Finally, 
$\mathbb{I}^{h}_{\bullet\bullet\bullet}$\label{xdea8} 
denotes the set of all numbers $i\in\mathbb{I}^{h}_{\bullet\bullet}$ (additionally) satisfying the condition that  there is no $j$ with $i\leq j<h$ such that the $j$th iteration of \mainn\ proceeds according to the scenario of   Subsubsubcase 2.1.2.

\begin{lem}\label{29a}
%\marginpar{29a}
Consider any $h\in\mathbb{I}$ such that the $h$'th iteration of \mainn\ is locking. Then the master scale of $\vec{E}_h$ is bounded by ${\mathcal R}\tim(\mathfrak{l})$.
\end{lem}

\begin{proof} Throughout this proof, $\mathfrak{w}$ will be an abbreviation of the constant $\mathfrak{e}_\top +1$.
Consider any $h\in\mathbb{I}$ such that the $h$th iteration of \mainn\ is locking. 
Let $m$ be the master scale of 
$\vec{E}_h$. We claim that $m$ is smaller than $2^{\mathfrak{w}-1} \mathfrak{t}(\max(\mathfrak{l},\mathfrak{G}(\mathfrak{l})))$ and hence, by Lemma \ref{lagg} and conditions 2, 3 and 4 of Definition 2.2 of \cite{AAAI}, $m$ is bounded by ${\mathcal R}\tim(\mathfrak{l})$. Indeed, 
for a contradiction, assume $m\geq 2^{\mathfrak{w}-1} \mathfrak{t}(\max(\mathfrak{l},\mathfrak{G}(\mathfrak{l})))$. 
We (may) additionally assume that $\mathfrak{t}(\max(\mathfrak{l},\mathfrak{G}(\mathfrak{l})))\not=0$. Let $b_1$ be 
the smallest element of $\mathbb{I}^{h}_{\bullet\bullet}$. So, there are no restarting 
iterations between $b_1$ (including) and $h$ (not including). But only restarting iterations 
of \mainn\ modify the master scale of $\vec{E}$. Thus, the master scale of $\vec{E}_{b_1}$ 
is the same $m$ as that of $\vec{E}_h$. Since $m>1$ 
and $b_1$ is the smallest element of $ \mathbb{I}^{h}_{\bullet\bullet}$, the $(b_1-1)$th iteration of \mainn\ (exists and) is restarting. Besides, that iteration does not proceed 
by the scenario of Case 1 or Subsubsubcase 2.2.2.2 of \mainn, because in either case the
 master scale of the resulting aggregation $\vec{E}_{b_1}$ would be reset to $1$. Hence, the $(b_1-1)$th iteration of \mainn\ proceeds according to the scenario of (the master-scale-doubling) Subsubsubcase 2.2.2.1. This means that the master scale of $\vec{E}_{b_1-1}$ is $m/2$. 
Let $b_2$ be the smallest element of $ \mathbb{I}^{b_1-1}_{\bullet\bullet}$. 
Reasoning as above but this time with $b_2$ and $b_{1}-1$ instead of $b_1$ and $h$, respectively, 
 we find that the master scale of $\vec{E}_{b_2}$ is $m/2$ and the master scale of $\vec{E}_{b_2-1}$ is $m/4$. Continuing this 
pattern, we further define $b_3>b_4>\ldots>b_{\mathfrak{w}}$ in the same style as we defined $b_1,b_2$ and find that the master scales of $\vec{E}_{b_3}$, $\vec{E}_{b_3-1}$, $\vec{E}_{b_4}$, $\vec{E}_{b_4-1}$, \ldots, $\vec{E}_{b_{\mathfrak{w}}}$, $\vec{E}_{b_{\mathfrak{w}}-1}$ are $m/4$, $m/8$, $m/8$, $m/16$, \ldots, $m/2^{\mathfrak{w}-1}$, $m/2^{\mathfrak{w}}$, respectively. Each iteration of \mainn\ that proceeds according to Subsubcase 2.1.2 results in ${\mathcal M}_k$ making a move in the real play of $h\leq \mathfrak{b}|\vec{d}|\mli F'(k)$.  
Since ${\mathcal M}_k$ plays (quasi)legally, altogether it makes fewer than $\mathfrak{w}$ moves.  This means that, altogether, there are fewer than $\mathfrak{w}$ iterations of \mainn\ that proceed according to Subsubcase 2.1.2. Besides, one of such iterations is the $h$th iteration.  Therefore there is at least one $i$ with $1\leq i< \mathfrak{w}$ such that 
$\mathbb{I}^{b_{i}-1}_{\bullet\bullet}=\mathbb{I}^{b_{i}-1}_{\bullet\bullet\bullet}$ and hence 
$b_{i+1}\in\mathbb{I}^{b_{i}-1}_{\bullet\bullet\bullet}$. Pick the smallest such $i$ (fix it!), and let us rename $b_i$ into $c$ and $b_{i+1}$ into $a$. Further, let $d$ be the smallest element of $\mathbb{I}^{h}$ such that $c\leq d$ and the $d$th iteration of \mainn\ is locking. It is not hard to see that such a $d$ exists. Namely, $d\in\{b_1,\ldots,h\}$ if $i=1$ and $d\in\{b_i,\ldots,b_{i-1}-1\}$ if $1<i<\mathfrak{w}$.

In what follows, we shall say that two organs $(\vec{\alpha},p)$ and $(\vec{\beta},q)$ are {\bf essentially the same}\label{xestsm}
 iff $\vec{\alpha}=\vec{\beta}$ and either $p=q$ or $p,q\in\{m/2^{i},m/2^{i-1}\}$, where $i$ is as above. This extends to all pairs $X,Y$ of organ-containing objects/structures (such as signed organs, bodies, aggregations, etc.) by stipulating that  $X$ and $Y$ are essentially the same iff they only differ from each other in that where $X$ has an organ $P$, $Y$ may have an essentially the same organ $Q$ instead. For instance, two signed organs are essentially the same iff they are both in $\{+P,+Q\}$  or both in $\{-P,-Q\}$  for some essentially the same organs $P$ and $Q$;  
two bodies $(P_1,\ldots,P_s)$ and $(Q_1,\ldots,Q_t)$ are essentially the same iff $s=t$ and, for each $r\in\{1,\ldots,s\}$, the organs $P_r$ and $Q_r$ are essentially the same; etc.  

\begin{cLa}\label{CL:2}
%\begin{quote}{\em {\bf Claim 2}: 
For any $j\in\{0,\ldots,d-c+1\}$, the aggregations $\vec{E}_{a+j}$ and $\vec{E}_{c+j}$ are essentially the same.
\end{cLa}%}\end{quote}

This claim can be proven by induction on $j$. We give an outline of such a proof, leaving more elaborate details to the reader. 
For the basis of induction, we want to show that the aggregations $\vec{E}_{a}$ and $\vec{E}_{c}$ are essentially the same. To see that this is so, observe that the master entry is the only entry of both aggregations. Also, the only iteration of \mainn\ between $a$ (including) and $c$ that modifies the master entry of $\vec{E}$ is the $(c-1)$th iteration, which proceeds according to Subsubsubcase 2.2.2.1 and the only change that it makes in the master body of $\vec{E}$ is that it doubles $\vec{E}$'s master scale $m/2^{i}$, turning it into  $m/2^{i-1}$. 

For the inductive step, consider any $j\in\{0,\ldots,d-c\}$ and make the following observations.  
Updating $\vec{E}_{c+j}$ to $\vec{E}_{c+j+1}$ happens through running $\simm_{n}^{\bullet}$\footnote{Of course, \mainn\ runs  
$\simm_{n}$ rather than $\simm_{n}^{\bullet}$, but this is only relevant to the value of the variable $U$ of \mainn. The latter may only become relevant to the way the variable $\vec{E}$ is updated when a given iteration of \mainn\ proceeds according to Subsubcase 2.2.2. However, no iterations between (including) $c$ and $d$ proceed according to that Subsubcase. So, it is safe to talk about $\simm_{n}^{\bullet}$ instead of $\simm_{n}$ here.}
(for a certain $n$) on certain arguments $A,B$. The same is the case with 
updating $\vec{E}_{a+j}$ to $\vec{E}_{a+j+1}$, where, by the induction hypothesis, the arguments $A'$ and $B'$ on which ${\mathcal H}_n$ is run are essentially the same as  $A$ and $B$, respectively.   
 So, the only difference between the two computations 
$\simm_{n}^{\bullet}(A,B)$ and $\simm_{n}^{\bullet}(A',B')$ is that, occasionally, one traces $m/2^{i-1}$ steps of ${\mathcal H}_n$'s work beginning from a certain configuration $W$  while the other only traces $m/2^{i}$ steps in otherwise virtually the same scenario. 
This guarantees that the outcomes of the two computations --- and hence the ways $\vec{E}_{c+j}$ and $\vec{E}_{a+j}$ are updated to their next values --- are essentially the same. The point is that, 
 since ${\mathcal H}_n$ runs in time $\mathfrak{t}$ and since --- as observed near the end of the preceding subsection --- the magnitude 
of no move on the simulated ${\mathcal H}_n$'s run tape exceeds $\max(\mathfrak{l},\mathfrak{G}(\mathfrak{l}))$, all moves that 
${\mathcal H}_n$ makes within $m/2^{i-1}\geq 2 \mathfrak{t}(\max(\mathfrak{l},\mathfrak{G}(\mathfrak{l})))$ Steps are in fact made within 
the first $m/2^{i}\geq  \mathfrak{t}(\max(\mathfrak{l},\mathfrak{G}(\mathfrak{l})))$ steps of the simulated interval, so the computations of $\simm_{n}^{\bullet}(A,B)$ and $\simm_{n}^{\bullet}(A',B')$ proceed in 
``essentially the same'' ways, yielding essentially the same outcomes.\vspace{4pt} 

Taking $j=d-c+1$, Claim \ref{CL:2} tells us that the master body of  $\vec{E}_{c+(d-c+1)}$ --- i.e., of $\vec{E}_{d+1}$ --- and the master body of $\vec{E}_{a+(d-c+1)}$ are essentially the same. This is however a contradiction, because the size of the former, as a result of the $d$th iterations' being locking, is greater than the size of the master body of any earlier aggregations $\vec{E}_{1},\ldots,\vec{E}_{d}$.  
\end{proof}

\begin{lem}\label{29b}
%\marginpar{29b}
Consider any $h\in\mathbb{I}$ such that the $h$'th iteration of \mainn\ is locking. Assume $e\in\mathbb{I}^{h}_{\bullet}$, and $(A,B,n)$ is the central triple of $\vec{E}_e$. Then the scales of all organs of $A$ and $B$ are bounded by  ${\mathcal R}\tim(\mathfrak{l})$. 
\end{lem}

\begin{proof} Assume $h$ is an element of $\mathbb{I}$ such that the $h$th iteration of \mainn\ is locking. Let $C$ be the master body of $\vec{E}_h$. It is not hard to see (by induction on $e-e_0$, where $e_0$ is the smallest element of $\mathbb{I}^{h}_{\bullet}$) that, for any $e\in\mathbb{I}^{h}_{\bullet}$, the scale of any organ of the body of any entry of $\vec{E}_e$ is either the same as the scale of one of the organs of  $C$, or one half, or one quarter, or\ldots  of such a scale. Thus,  the scales of the organs of $C$ (at least, the greatest of such scales) are not smaller that the scales of the organs of the entries of $\vec{E}_e$, including the scales of the organs of $A$ and $B$. For this reason, it is sufficient to prove that the scales of all organs of  $C$ are bounded by  ${\mathcal R}\tim(\mathfrak{l})$. 

Let $C=(O_1,\ldots,O_{2m},O_{2m+1})$, and let $p_1,\ldots,p_{2m},p_{2m+1}$ be the corresponding scales. 
Note that, since the $h$th iteration of \mainn\ is locking, we have $h\in\mathbb{I}^{\hbar}_{!}$ and, consequently, $C$ is a restriction of $\mathbb{B}^{\hbar}_{k}$. Therefore, 
according to  Claim \ref{CL:1} from the proof of Lemma \ref{a21}, we have $p_1=p_2$, $p_3=p_4$, \ldots, $p_{2m-1}=p_{2m}$. So, it is sufficient to consider $p_i$ where $i$ is an odd member of $\{1,\ldots,2m+1\}$. 
The case of $i=2m+1$ is immediately taken care of by Lemma \ref{29a}. Now consider any odd member $i$ of $\{1,\ldots,2m-1\}$.
Let $j$ be the $(h,k)$-birthtime of $(O_1,\ldots,O_{i+1})$. Thus, the $(j-1)$th iteration of \mainn\  is locking. But note that $p_i$ is the master scale of $\vec{E}_{j-1}$. Then, according to Lemma \ref{29a}, $p_i$ is bounded by ${\mathcal R}\tim(\mathfrak{l})$.
\end{proof}

Now we are ready to argue that ${\mathcal M}$ runs in target time. We already know from Lemma \ref{beijing} that, for a certain bound $\mathfrak{z}\in {\mathcal R}\tim$, \mainn\ is iterated only $\mathfrak{z}(\mathfrak{l})$ times. In view of $ {\mathcal R}\tim$'s being at least polynomial as well as polynomially closed,  
it remains to see that each relevant iteration takes a polynomial (in $\mathfrak{l}$) amount of time. Here ``relevant'' means an iteration that  is followed (either within the same iteration or in some later iteration) by an ${\mathcal M}_k$-made  move without meanwhile being interrupted by Environment's moves. In other words, this is an $e$th iteration with $e\in\mathbb{I}^{h}_{\bullet}$ for some $h\in\mathbb{I}$ such that the $h$th iteration of \mainn\ is locking. Consider any such $e$.

There are two asymptotically relevant sources/causes of the time consumption of the $e$th (as well as any other) iteration of \mainn:   running $\simm_n(A^{\even},B^{\odd})$, where $(A,B,n)$ is the central triple of $\vec{E}_e$,  and  periodically polling ${\mathcal M}_k$'s run tape to see if Environment has made any new moves.

Running $\simm_n(A^{\even},B^{\odd})$ requires simulating 
the corresponding machine ${\mathcal H}_n$ in the scenario determined by $A^{\even}$ 
and $B^{\odd}$. With asymptotically irrelevant or superseded  details suppressed, simulating a single step of ${\mathcal H}_n$ requires going, a constant number of times,  through ${\mathcal M}_k$'s work and run tapes to collect the information necessary for updating ${\mathcal H}_n$'s ``current'' configuration to the next one, and to actually make such an update. As we already know from Section \ref{sls}, the size of (the non-blank, to-be-scanned portion of) ${\mathcal M}_k$'s work tape is bounded by ${\mathcal R}\spa$.\footnote{Of course, a  tape is infinite in the rightward direction, but in contexts like the present one we treat the leftmost blank cell of a tape as its ``end''.}  And the size of ${\mathcal M}_k$'s run tape is $O(\mathfrak{l})$ (the early moves) plus $O(\mathfrak{G}(\mathfrak{l}))$ (the late moves). Everything together, in view of the linear closure of ${\mathcal R}\tim$ (condition 3 of Definition 2.2 of \cite{AAAI}) and the facts $\mathfrak{G}\in{\mathcal R}\spa$ (Lemma \ref{lagg}) and ${\mathcal R}\spa\preceq {\mathcal R}\tim$ (condition 5 of Definition 2.2 of \cite{AAAI}), is well within the target ${\mathcal R}\tim(\mathfrak{l})$. 

The amount of steps of ${\mathcal H}_n$ to be simulated when running $\simm_n(A^{\even},B^{\odd})$ is obviously at most a constant   times the greatest of the scales of the organs of $A$ and $B$, which, in view of Lemma \ref{29b}, is ${\mathcal R}\tim(\mathfrak{l})$.  

Thus, the time $T$ needed for running $\simm_n(A^{\even},B^{\odd})$ is the product of the two ${\mathcal R}\tim(\mathfrak{l})$   quantities established in the preceding two paragraphs. By the polynomial closure of ${\mathcal R}\tim$, such a product remains ${\mathcal R}\tim(\mathfrak{l})$.    
How much time is added to this by the polling routine? Obviously the latter is repeated at most $T$ times. Any given repetition  does not require more time than it takes to go from one end of the run tape of ${\mathcal M}_k$ to the other 
end. And this quantity, as we found just a while ago, is ${\mathcal R}\tim(\mathfrak{l})$. Thus,  the $e$th iteration of \mainn\ takes ${\mathcal R}\tim(\mathfrak{l})+{\mathcal R}\tim(\mathfrak{l})\times {\mathcal R}\tim(\mathfrak{l})$ time, which, by ${\mathcal R}\tim$'s being polynomially closed,  remains ${\mathcal R}\tim(\mathfrak{l})$ as promised.

\section{Final remarks}\label{sfr}
%\marginpar{sfr}
In writing this paper and its predecessor, the author has tried to keep balance between  generality and simplicity, often sacrificing the former for the sake of the latter. Among the ways that the present results could be strengthened is relaxing the concept of a regular theory $\areleven_{\mathcal A}^{\mathcal R}$. Specifically, the condition of ${\mathcal R}\amp$'s being linearly closed  can be removed as long as  Definition 2.2 of \cite{AAAI}  is correspondingly refined/readjusted. This condition, in fact, amounts to adopting an asymptotic view of amplitude complexity, which significantly simplifies the completeness proofs, allowing us to avoid numerous annoying exceptions and details one would need to otherwise deal with. As noted in \cite{cl12}, however, unlike time and space complexities, we may not always be willing to --- and it is almost never really  necessary to --- settle for merely asymptotic analysis when it comes to amplitude complexity. A non-asymptotic approach to amplitude complexity would make it possible to consider much finer amplitude complexities, such as ``strictly $\ell$'' (``non-size-increasing'', as studied in \cite{bbb1}), ``$\ell$ plus a constant'', etc.

\appendix

\section{Proof of Lemma \ref{vasa}}\label{sap2}
Lemma  \ref{vasa} states:

\begin{quote}{\em There is an effective procedure that takes an arbitrary bounded formula $H(\vec{y})$,
 an arbitrary HPM $\mathcal L$ and constructs an HPM $\mathcal M$ such that, as long as $\mathcal L$ is a provident solution of $H(\vec{y})$, 
 the following conditions are satisfied: 
\begin{enumerate}[label=\arabic*.]
\item
 $\mathcal M$ is a quasilegal and unconditionally provident solution of $H(\vec{y})$.

\item
 If $\mathcal L$ plays $H(\vec{y})$ prudently, then $\mathcal M$ plays $H(\vec{y})$ unconditionally prudently.

\item For any arithmetical functions $\mathfrak{a},\mathfrak{s},\mathfrak{t}$,    
if $\mathcal L$ plays $H(\vec{y})$ in 
tricomplexity $(\mathfrak{a},\mathfrak{s},\mathfrak{t})$, then $\mathcal M$ plays in unconditional tricomplexity $(\mathfrak{a},\mathfrak{s},\mathfrak{t})$.
\end{enumerate}}\end{quote}\medskip %%%JK

\noindent Consider an arbitrary  HPM $\mathcal L$ and an arbitrary bounded formula $H(\vec{y})$ with all free variables displayed. We want to (show how to) construct an HPM $\mathcal M$\label{xmm4} --- with the same number of work tapes as $\mathcal L$--- satisfying the above conditions 1-3. From our construction of $\mathcal M$ it will be immediately clear that $\mathcal M$ is built effectively from $H(\vec{y})$ and $\mathcal L$. As usual, we may not always be very careful about the distinction between $H(\vec{y})$ and $\ada H(\vec{y})$, but which of these two is really meant can always easily be seen from the context.

We agree on the following terminology. 
A {\bf semiposition}\label{xsemiposition} is a string 
$S$ of the form $\xx_1\alpha_1\ldots\xx_n\alpha_n\omega$, where each $\xx_i$  is a  label $\pp$ or $\oo$, each $\alpha_i$ 
is a string over the keyboard alphabet, and $\omega\in \{\epsilon, \blank\}$ (remember that $\epsilon$ stands for the empty string).   When 
$\omega$ is $\blank$, we say that  $S$ is {\bf complete};\label{xcsmp} otherwise $S$ is {\bf incomplete}. We say that a semiposition   $S'$ is a {\bf completion}\label{xcdfd} of $S$ iff (1) either $S$ is complete and $S'=S$, or (2) $S$ is incomplete and $S'=S\beta\blank$ for some  (possibly empty) string $\beta$ over the keyboard alphabet. 
When $S$ is complete --- namely, is $\xx_1\alpha_1\ldots\xx_n\alpha_n\blank$ --- then the {\bf position spelled by $S$}, as expected, is the position 
$\seq{\xx_1\alpha_1,\ldots,\xx_n\alpha_n}$. 
We say that a semiposition $S$ is {\bf legitimate}\label{xdswe} (resp. {\bf quasilegitimate}\label{xdswepp}) iff there is a completion $S'$ of $S$ such that the  position spelled by $S'$ is a legal (resp. quasilegal) position of $\ada H(\vec{y})$.  
The {\bf compression}\label{xcdcfgf}  of a legitimate or quasilegitimate semiposition $S$ is the expression  $\overline{S}$ resulting from $S$ through replacing 
the numer  of every numeric move by the symbol $\star$.  
Note that, while generally there are infinitely many  possible legitimate or quasilegitimate semipositions,  the number of their compressions is finite. The reason is that an infinite variety of legal runs of $\ada H(\vec{y})$ exists only due to numer variations   within numeric moves; in compressions, however, all numers degenerate into $\star$.

In the context of a given step $i$ of a given computation branch of a given HPM, by the  
{\bf so-far-seen   semiposition}\label{xsofar}
we shall mean the semiposition $W$ written at time $i$ on the initial section of the run tape that has ever been visited (at steps $\leq i$) by the run-tape scanning head, except that the last symbol of $W$ should be $\blank$ if the corresponding cell contained a $\blank$ at the time when it was last seen by the scanning head, even if the content of that cell changed (namely, became  
$\pp$ or $\oo$) later. Intuitively, such a $W$ is exactly what the machine knows at time $i$ about its run-tape content based on what it has seen there so far. Next, let $Z$ be the semiposition $\pp\delta_1\ldots\pp\delta_m$, where $\delta_1,\ldots,\delta_m$ are the moves made by the machine so far (at steps $\leq i$). And let $\pi$ be the  string residing in the buffer at time $i$.  Then by the  
{\bf so-far-authored   semiposition}\label{xsofarau}
we shall mean the (complete) semiposition $Z\blank$ if $\pi$ is empty, and the (incomplete) semiposition $Z\pp\pi$ if 
$\pi$ is nonempty. The {\bf windup}\label{xwindup} 
of a quasilegitimate yet incomplete  semiposition  $V$
 of the form $\pp\delta_1\ldots\pp\delta_m\pp\pi$ 
is the lexicographically smallest string $\omega$ such that  
$\seq{\pp\delta_1,\ldots,\pp\delta_m,\pp\pi\omega}$ is a 
$\pp$-quasilegal position of $\ada H(\vec{y})$. 
Note that there is only a constant number of strings that are windups of some incomplete quasilegitimate semipositions. Also note that knowing the compression $\overline{V}$ of an (incomplete quasilegitimate) semiposition $V$ is sufficient to determine $V$'s windup.
   
We let $\mathcal M$ keep partial track of the so-far-authored quasilegitimate semiposition $V$ through remembering its compression $\overline{V}$. 
Similarly, $\mathcal M$ keeps track of the so-far-seen legitimate semiposition $W$ through remembering its compression $\overline{W}$; besides, one of the symbols of $\overline{W}$ is marked to indicate (keep track of) the current location of $\mathcal M$'s run-tape scanning head.\footnote{Namely, a marked symbol of $\overline{W}$ other than $\star$ indicates that the head is looking at the corresponding symbol of $W$, and a marked $\star$ indicates that the head is looking at one of the bits of the corresponding numer.}    With appropriately arranged details that are not worth discussing here, it is possible for $\mathcal M$, this way, to be able to immediately detect if and when $W$ becomes illegitimate. If and when this happens, we let $\mathcal M$ retire; besides, if $V$ is quasilegitimate yet incomplete at the time of this event, then $\mathcal M$ puts $V$'s windup into the buffer and, simultaneously, enters a move state before retiring. We shall refer to a move made this way as a {\bf retirement move}.\label{xretmo}
Maintaining the above $\overline{W}$ (together with its mark) and $\overline{V}$ only requires a  constant  amount of memory, so this can be fully done through $\mathcal M$'s state rather than tape memory. This means that, as long as $W$ remains legitimate,  $\mathcal M$ can follow the work of $\mathcal L$ step-by-step without having any  time or space overhead, and act (reposition heads, put things into the buffer, move, etc.) exactly like $\mathcal L$, with the only difference between the two machines being that $\mathcal M$ has a greater number of states than $\mathcal L$ does, with any given state of $\mathcal L$ being imitated by one of many ``counterpart'' states of $\mathcal M$, depending on the present values of $\overline{V}$ and the marked $\overline{W}$ that each such state ``remembers'' (e.g., is labeled with). 
 
For the rest of this appendix, assume $\mathcal L$ is a provident solution of $H(\vec{y})$. Fix an arbitrary computation branch $B$ of $\mathcal M$, and let $\Gamma^{B}_{\infty}$ be the run spelled by $B$. From now on, whenever a context requires a reference to a computation branch but such a reference is missing, it should be understood as that we are talking about $B$. 
For simplicity, we shall assume that, in $B$, Environment made (legal) initial moves that brought $\ada H(\vec{y})$ down to 
$H(\vec{c})$ for some constants $\vec{c}$. Fix these $\vec{c}$.  The case of $B$ violating this assumption is not worth our attention for the reason of being trivial or, at least, much simpler than the present case. 

We also fix arbitrary arithmetical functions $\mathfrak{a},\mathfrak{s},\mathfrak{t}$. We may assume that all three functions are unary, or else replace them with their unarifications. Since the parameters $B$, $\Gamma_{\infty}^{B}$, $\vec{c}$, $\mathfrak{a}$, $\mathfrak{s}$, $\mathfrak{t}$ are arbitrary, it is sufficient for us to focus on them
 and just show that the three conditions of the lemma are satisfied in the context of these particular parameters. For instance, to show that $\mathcal M$ plays $\ada H(\vec{y})$ quasilegally, it is sufficient to show that $\Gamma^{B}_{\infty}$ is a $\pp$-quasilegal run of $\ada H(\vec{y})$. 

We extend the notation $\Gamma^{B}_{\infty}$\label{xgamm} from $B$ to any computation branch $C$ of either $\mathcal M$ or $\mathcal L$, stipulating that $\Gamma^{C}_{\infty}$ is the run spelled by $C$. We further agree that, for any $i\geq 0$, $\Gamma^{C}_{i}$ stands for the position spelled on the run tape of the corresponding machine at step $i$ of branch $C$, and $\ell^{C}_{i}$\label{xellc} stands for the background of that step. We also agree that $W_{i}^{C}$\label{xdublie} denotes the so-far-seen semiposition at step $i$ of branch $C$, and $V_{i}^{C}$\label{xvublie} denotes the so-far-authored semiposition at step $i$ of  $C$. Finally, since $\ada H(\vec{y})$ and $H(\vec{c})$ are the only formulas/games we deal with in this appendix, without risk of ambiguity  we will often omit references to them 
when saying ``legal'', ``quasilegal'', ``prudent'' etc. 

Consider any $i$ such that $W_{i}^{B}$ is legitimate. The legitimacy of   $W_{i}^{B}$ means it has a completion  $U=\xx_1\alpha_1\ldots\xx_n\alpha_n\beta\blank$  such that 
the position $\Omega=\seq{\xx_1\alpha_1,\ldots,\xx_n\alpha_n\beta}$  spelled by $U$ 
is  legal. Let $k$ be the number of $\oo$-labeled moves in $\Omega$. And let $C$ be the computation branch of $\mathcal M$ in which Environment acts exactly as it does in $B$, with only the following two differences:
(1) Environment stops making any moves after it makes its $k$th move (meaning that, if $k=0$, Environment simply never moves); (2) If $\xx_n=\oo$, Environment's $k$th move (i.e., the $n$th move of the play) is $\alpha_n\beta$. Of course, $C$ spells a legal run. For this reason, in this branch $\mathcal M$ behaves just like $\mathcal L$ in the branch $D$ where the environment makes exactly the same moves, in exactly the same order and at exactly the same times,  as in $C$. We call such a $D$ the 
{\bf $W_{i}^{B}$-induced branch of $\mathcal L$}.\label{xwiin} The following two lemmas are simple observations, hardly requiring any proofs: 

\begin{lem}\label{wha}
%\marginpar{wha} 
Consider any $j\geq 0$ such that $\Gamma_{j}^{B}$ is legitimate, and let $D$ be the $W_{j}^{B}$-induced branch of $\mathcal L$. We have:
\begin{enumerate}[label=\arabic*.]
\item In $D$, $\mathcal L$'s environment  makes no moves at any step $e$ with $e>j$. 

\item $\Gamma_{\infty}^{D}$ is a legal run of $\ada H(\vec{y})$. 

\item The initial segment of $\Gamma_{\infty}^{B}$ that brings $\ada H(\vec{y})$ down to $H(\vec{c})$ is also an initial segment of 
$\Gamma_{\infty}^{D}$.

\item $V_{j+1}^{D}=V_{j+1}^{B}$, and hence also $(\Gamma^{D}_{j+1})^\top=(\Gamma^{B}_{j+1})^\top$. 

\item For any $e\geq 0$, $\ell_{e}^{D}\leq \ell_{e}^{B}$.
\end{enumerate}
\end{lem}

\begin{lem}\label{whaha}
%\marginpar{whaha} 
There is a number $s$ such that, for every $j\geq s$, $W_j=W_s$. The smallest of such numbers $s$ we call the 
{\bf $W$-stabilization point}.\label{xlvw}  
\end{lem}

Having set up the above preliminaries, we prove the lemma clause by clause.\medskip
\begin{description}
\item[{\rm\em CLAUSE 1}] Our goal is to show that: 
%\marginpar{po1-po3}
\begin{eqnarray}
 && \mbox{\em $\Gamma^{B}_{\infty}$ is $\pp$-won (so, $\mathcal M$ is a solution of $\ada H(\vec{y})$);}\label{po1}\\
 && \mbox{\em $\Gamma^{B}_{\infty}$ is  $\pp$-quasilegal (so, $\mathcal M$ plays  quasilegally);}\label{po2}\\ 
 && \mbox{\em $B$ is provident (so, $\mathcal M$ plays unconditionally providently).}\label{po3} 
\end{eqnarray}
\begin{description}
\item[(\ref{po1})] From the description of $\mathcal M$ we can see that, if $\Gamma^{B}_{\infty}$ is  $\oo$-legal and thus the so-far-seen semiposition always remains legitimate, $\mathcal M$ interacts with its environment exactly like $\mathcal L$ interacts with its environment in the corresponding scenario\footnote{Namely, in the computation branch where $\mathcal L$'s environment makes exactly the same moves at exactly the same times and in exactly the same order as in $B$.}  and, since $\mathcal L$ is a solution of $\ada H(\vec{y})$, $\Gamma^{B}_{\infty}$ is $\pp$-won. And if $\Gamma^{B}_{\infty}$ is $\oo$-illegal, then $\Gamma^{B}_{\infty}$ is  automatically $\pp$-won. 

\item[(\ref{po2})]  For a contradiction, suppose $\Gamma^{B}_{\infty}$ is not $\pp$-quasilegal.  Let $i$ be the smallest number such that the position $\Gamma^{B}_{i}$  is not  $\pp$-quasilegal. Let $\phi$ be the (``offending'') move that $\mathcal M$ made at step $i$ of $B$. 

%INDENTATION NEEDED HERE
Assume $W_{i-1}$ is legitimate. Let $D$ be the $W_{i-1}$-induced branch of $\mathcal L$. According to clause 4 of Lemma \ref{wha}, $(\Gamma_{i}^{D})^\top=(\Gamma_{i}^{B})^\top$. So, $\Gamma_{i}^{D}$ is not $\pp$-quasilegal, and then the same holds for the extension 
$\Gamma_{\infty}^{D}$ of $\Gamma_{i}^{D}$. Of course, $\Gamma_{\infty}^{D}$'s not being $\pp$-quasilegal implies that it is simply illegal. But this contradicts clause 2 of Lemma \ref{wha}, according to which  $\Gamma_{\infty}^{D}$ is legal. 

%INDENTATION NEEDED HERE
Now assume $W_{i-1}$ is not legitimate. Note that $i\geq 2$, because,
at the initial step $0$, $\mathcal M$ would not be able to see an illegitimate semiposition (at best, $\mathcal M$ would only see the label $\oo$ in the leftmost cell, nothing else).
Further note that the semiposition 
$W_{i-2}$ is legitimate, because otherwise $\mathcal M$ would have retired right after seeing it and thus would not have moved at step $i$. As soon as $\mathcal M$ sees the illegitimate $W_{i-1}$, it retires. Thus, the move $\phi$ made at step $i$ is a retirement move. Looking back at the conditions under which $\mathcal M$ makes a retirement move, we see that the so-far-authored semiposition $V_{i}^{B}$  should be complete and quasilegitimate. Let $\Sigma_{i}^{B}$ be the position spelled by $V_{i}^{B}$. So, $\Sigma_{i}^{B}$ is $\pp$-quasilegal. But note that $(\Sigma_{i}^{B})^\top=(\Gamma_{i}^{B})^\top$, and thus we are facing a contradiction because, as we remember, $\Gamma_{i}^{B}$ is not $\pp$-quasilegal.

\item[(\ref{po3})] As already noted in the proof of (\ref{po1}), if the run $\Gamma^{B}_{\infty}$ is  $\oo$-legal, $\mathcal M$ and its environment behave exactly like $\mathcal L$ and its environment in the corresponding scenario. Then, since $\mathcal L$ plays providently, 
$B$ is a provident branch. Suppose now  $\Gamma^{B}_{\infty}$ is $\oo$-illegal. 

%INDENTATION NEEDED HERE
First, assume the so-far-seen semiposition in $B$ becomes illegitimate at some step $i$. Note that ($i>0$ and)  $W_{i-1}^{B}$ is legitimate. Let $D$ be the $W_{i-1}^{B}$-induced branch of $\mathcal L$. By clauses 2 and 4 of Lemma \ref{wha}, 
$\Gamma^{D}_{\infty}$  is   $\oo$-legal and $V_{i}^{D}=V_{i}^{B}$. The semiposition $V_{i}^{D}$  must be quasilegitimate because otherwise, as can be seen with a little thought, (the provident) $\mathcal L$ will have to make an illegal move in branch $D$ at some point.  But, in branch $B$, $\mathcal M$ retires immediately after seeing the non-legitimate $W_i$. The only possibility for the buffer content of $\mathcal M$ to remain nonempty after retirement would be if 
$V_{i}^{B}$ was not quasilegitimate. However, as just observed, this is not the case. 

%INDENTATION NEEDED HERE
Now assume  the so-far-seen semiposition in $B$ never becomes illegitimate. Let $i$ be the $W$-stabilization point (which exists according to Lemma \ref{whaha}).  And let $D$ be the $W_{i}^{B}$-induced branch of $\mathcal L$. 
It is not hard to see that, throughout the entire play, $\mathcal M$ behaves --- makes moves, puts strings into the buffer, repositions scanning heads --- the same way  in $B$ as $\mathcal L$ behaves in $D$. From clause 2 of lemma \ref{wha}, we also know that  $D$ spells a $\oo$-legal run and hence, due to $\mathcal L$'s playing providently, 
$D$ contains infinitely many steps with empty buffer contents. Then so
does $B$. That is, $B$ is provident.
\end{description}

\item[{\rm\em CLAUSE 2}] Assume $\mathcal L$ is a prudent solution of $H(\vec{y})$. We want to show that the run $\Gamma^{B}_{\infty}$ is $\pp$-prudent. For a contradiction, deny this. Let $i$ be the smallest number such that 
$\Gamma^{B}_{i}$  is not $\pp$-prudent. Note that $i>0$. It is obvious that  a move is made in $B$ at step $i$. Let us call that move $\phi$. 

Assume $W_{i-1}^{B}$ is legitimate. Let $D$ be the $W_{i-1}^{B}$-induced branch of $\mathcal L$.  Clauses 3 and 4 of Lemma \ref{wha} imply that 
$\Gamma_{i}^{D}$ is not $\pp$-prudent, and then the same holds for the extension $\Gamma_{\infty}^{D}$ of $\Gamma_{i}^{D}$. At the same time, by clause 2 of the same lemma,  $\Gamma_{\infty}^{D}$ is legal. This is a contradiction, because $\mathcal L$ is a prudent solution of $H(\vec{y})$ and, as such, it could not have generated a ($\oo$-)legal run ($\Gamma_{\infty}^{D}$) that is not $\pp$-prudent. 

Now assume $W_{i-1}^{B}$ is not legitimate. Then, just as in the last paragraph of our proof of (\ref{po2}),  $i\geq 2$, $W_{i-2}^{B}$ is legitimate, and  $\phi$  is a retirement move. Let $D$ be the $W_{i-2}^{B}$-induced branch of $\mathcal L$. 
Analyzing the conditions under which $\mathcal M$ makes a retirement move, we see that $\phi$ (rather than some proper prefix of it) was the content of $\mathcal M$'s buffer at step $i-1$ of $B$. Then, by clause 4 of Lemma \ref{wha}, the same is the case for $\mathcal L$'s buffer in branch $D$. But, since $\mathcal L$ plays providently and (by clause 2 of Lemma \ref{wha}) $\Gamma^{D}_{\infty}$ is legal, in $D$, sooner or later $\mathcal L$ will have to make a move $\phi'$ such that $\phi$ is a prefix of $\phi'$. Obviously such a move $\phi'$, if legal, will inherit the imprudence of $\phi$. This, together with clause 2 of Lemma \ref{wha}, contradicts  our assumption that $\mathcal L$ is a prudent solution of $H(\vec{y})$.\medskip

\item[\rm\em CLAUSE 3]: Assume $\mathcal L$ is a  $(\mathfrak{a},\mathfrak{s},\mathfrak{t})$ tricomplexity solution of $H(\vec{y})$.\vspace{4pt}

{\em Amplitude}: Assume $\mathcal M$ makes a move $\phi$ at a step $i$. Let $m_\phi$ be the magnitude of $\phi$. We want to show that $m_\phi\leq \mathfrak{a}(\ell_{i}^{B})$. 

First, suppose  
$W_{i-1}^{B}$ is legitimate. Let $D$ be the $W_{i-1}^{B}$-induced branch of $\mathcal L$. In view of clause 4 of Lemma \ref{wha}, in $D$, $\mathcal L$ makes the same move $\phi$ at  the same time $i$. Since $\mathcal L$ plays in amplitude $\mathfrak{a}$ and since, by clause 2 of Lemma \ref{wha}, the run $\Gamma_{\infty}^{D}$ is legal, we have $m_\phi\leq \mathfrak{a}(\ell_{i}^{D})$. 
The desired $m_\phi\leq \mathfrak{a}(\ell_{i}^{B})$ follows from here by clause 5 of Lemma \ref{wha}.

Now suppose $W_{i-1}^{B}$ is not legitimate.  Then, as  in the last paragraph of our proof of (\ref{po2}) $i\geq 2$,  $W_{i-2}^{B}$ is legitimate, and  $\phi$ is a retirement move. Let $D$ be the $W_{i-2}^{B}$-induced branch of $\mathcal L$. And let $\beta$ be the content of $\mathcal M$'s buffer at step $i-1$ of $B$. By clause 4 of Lemma \ref{wha}, the same $\beta$ is in the buffer of $\mathcal L$ at step $i-1$ of $D$. At some step $s\geq i$ of $D$, the provident $\mathcal L$ should make a move $\gamma$ such that $\beta$ is a prefix of $\gamma$. Let $m_\gamma$ be the magnitude of that move.  Since the 
run spelled by $D$ is legal (clause 2 of Lemma \ref{wha}) and  $\mathcal L$ plays in amplitude $\mathfrak{a}$, we have $m_\gamma\leq \mathfrak{a}(\ell_{s}^{D})$. But, in view of clause 1 of Lemma \ref{wha}, $\ell_{s}^{D}=\ell_{i}^{D}$. Thus, $m_\gamma\leq \mathfrak{a}(\ell_{i}^{D})$. This, in view of clause 5 of Lemma \ref{wha}, implies $m_\gamma\leq \mathfrak{a}(\ell_{i}^{B})$. 
From the way we measure magnitudes  and from the way the windup operation is defined, it is clear that $m_\phi\leq m_\gamma$. Consequently, $m_\phi\leq \mathfrak{a}(\ell_{i}^{B})$.\vspace{4pt}

{\em Space}:  
Let $i$ be the $W$-stabilization point, and let $D$ be the $W_{i}^{B}$-induced branch of $\mathcal L$. If  $W_{i}^{B}$ is legitimate, then, as observed in the last paragraph of our proof of (\ref{po3}), $\mathcal M$'s behavior throughout $B$ is indistinguishable from that of  $\mathcal L$ in $D$; this, in view of clause 5 of Lemma \ref{wha}, 
 means that $B$, just like $D$, does not violate the $\mathfrak{s}$ space limits.  Now suppose $W_{i}^{B}$ is not legitimate. Whatever we said above still applies to the behavior of $\mathcal M$ up to (including) step $i-1$. After that it makes a transition to step $i$ and retires without consuming any additional space. So, the space consumption again remains within the limits of $\mathfrak{s}$.\vspace{4pt}

{\em Time}: 
Again, let  $i$ be the stabilization point, and let $D$ be the $W_{i}^{B}$-induced branch of $\mathcal L$.
If $W_{i}^{B}$ is legitimate, then, for the same reasons as in the case of space, $B$ does not violate the $\mathfrak{t}$ time limits.
 Now suppose $W_{i}^{B}$ is not legitimate. Whatever we said in the preceding sentence 
still applies to the behavior of $\mathcal M$ in $B$ up to (including) step $i-1$. Then $\mathcal M$  makes a transition to step $i$ and retires. If no move is made upon this transition, all is fine. And if a move is made, then, in view of the relevant clauses of Lemma \ref{wha}, it can be seen that the timecost of that move does not exceed the timecost of the move that the provident $\mathcal L$ would have to make in $D$ sooner or later after time $i-1$. So, the time bound $\mathfrak{t}$ is not violated.  
\end{description}

\section{Proof of Lemma \ref{reason}}\label{sap1}
%\marginpar{sap1}

Lemma  \ref{reason}, to a proof of which this appendix is exclusively devoted, reads: 
\begin{quote}{\em 
There is an effective procedure that takes an arbitrary bounded formula $H(\vec{y})$,  
 an arbitrary HPM $\mathcal N$ and constructs an HPM $\mathcal K$ such that, for any regular boundclass triple $\mathcal R$, if $H(\vec{y})$ is ${\mathcal R}\spa$-bounded and $\mathcal N$ is an $\mathcal R$ tricomplexity solution of $H(\vec{y})$, then $\mathcal K$ is a provident and prudent $\mathcal R$ tricomplexity solution of $H(\vec{y})$. 
}\end{quote}

\subsection{Getting started}\label{sgst}
%\marginpar{sgst}
Pick and fix an HPM $\mathcal N$ and a bounded formula $H=H(\vec{y})=H(y_1,\ldots,y_\mathfrak{u})$ with all free variables displayed. The case of $H(\vec{y})$ being elementary is trivial, so we assume that $H(\vec{y})$ contains at least one choice operator.  Fix $\mathfrak{D}$\label{xetr} as the maximum number of labmoves in any legal run of $\ada H$.
Further fix   $\mathfrak{S}$\label{xgtrft} as the superaggregate bound of $H$.  

 Assume $ {\mathcal R} $ is a regular boundclass triple such that the formula $H(\vec{y})$ is ${\mathcal R}\spa$-bounded and 
$\mathcal N$ is an ${\mathcal R}$ tricomplexity solution of $H(\vec{y})$. Note  that, by Lemma \ref{lagg}, $\mathfrak{S}\in{\mathcal R}\spa$. It is important to point out that our construction of $\mathcal K$ below does not depend on $\mathcal R$ or any assumptions on it.

In view of Lemma 10.1 of \cite{cl12} and with Remark 2.4 of \cite{AAAI} in mind, we may and will assume that $\mathcal N$ plays $H$ providently. Then Lemma \ref{vasa} (whose proof does not rely on the present lemma) allows us to further assume that $\mathcal N$ is a quasilegal, unconditionally provident and unconditionally $\mathcal R$ tricomplexity solution of $H$.

Following the notational practice of Section \ref{sls}, 
we shall write ${\mathcal R}\spa(\ell)$ as an abbreviation of the phrase 
``$O\bigl(\mathfrak{p}(\ell)\bigr)$ for some $\mathfrak{p}(z)\in{\mathcal R}\spa$''. Similarly for  ${\mathcal R}\tim(\ell)$ and ${\mathcal R}\amp(\ell)$.

The technique that we employ below is very similar to the one used in Section 11 of \cite{cl12}. 
 Our goal is to construct a machine \ ${\mathcal K}$\label{xkk} such that $\mathcal K$ is a provident and prudent ${\mathcal R} $-tricomplexity solution of 
$H(\vec{y})$.  From our construction it will be immediately clear that the construction is effective  as required. 

In both our description of the work of  $\mathcal K$ and our subsequent analysis of it, we shall rely --- usually only implicitly ---  on the  Clean Environment Assumption.  Making this assumption is safe because the desired properties of $\mathcal K$ are (1) being a solution of $H(\vec{y})$, (2) playing $H(\vec{y})$ providently, (3) playing $H(\vec{y})$ prudently  and (4)  playing $H(\vec{y})$ in   $\mathcal R$ tricomplexity. The definitions of all four  of these properties, unlike, for instance, the definitions of the {\em unconditional} versions of the last three  (cf. Section \ref{sq}),  only look at the $\oo$-legal plays of $\ada H$ by $\mathcal K$. This means  that it virtually does not matter what happens if $\mathcal K$'s adversary starts playing illegally. 

We design  $\mathcal K$  as a single-work-tape HPM.\label{xghjy} At the beginning of the play, as usual,  it waits --- without consuming any space --- till Environment chooses constants $\vec{c}$ for all $\mathfrak{u}$ free variables $\vec{y}$ of $H$. If this never happens, $\mathcal K$ is an automatic winner trivially complying with the providence, prudence and $\mathcal R$ tricomplexity conditions. Having said that, for the rest of this construction and our subsequent analysis of it, we shall assume that, in the scenario that we are considering, Environment indeed chose the constants $\vec{c}$ (fix them!) for $\vec{y}$  during an initial episode of the play. 

Let us agree that 
a {\bf quasilegal move}\label{xqlm} (of $H(\vec{c})$)  means a move that may appear,  with either label, in some quasilegal run of $H(\vec{c})$. And the {\bf truncation}\label{xtrun} of a move $\alpha$ is    the $H(\vec{c})$-prudentization of the 
longest  prefix  $\alpha'$ of $\alpha$  such that $\alpha'$ is also a prefix of some quasilegal move   
 Note that, in view of our earlier assumption that $H$ is not elementary, every move  has a (possibly empty) truncation.  

 Once all constants $\vec{c}$ are chosen by Environment, ${\mathcal K}$ 
computes the value of $\mathfrak{S}|\max(\vec{c})|$ and remembers it for possible use in the future. It is not hard to see that, in view of  the basic closure properties of boundclasses and the relevant conditions of Definition 2.2 of \cite{AAAI}, $\mathfrak{S}|\max(\vec{c})|$ can be computed and recorded in space ${\mathcal R}\spa|\max(\vec{c})|$ and time ${\mathcal R}\tim|\max(\vec{c})|$. For this reason, when trying to show that  $\mathcal K$  runs in tricomplexity $\mathcal R$, the present episode of computing and remembering  $\mathfrak{S}|\max(\vec{c})|$ can (and will) be safely ignored.

Upon the completion of the above step, ${\mathcal K}$  
 starts simulating $\mathcal N$ in the scenario where, at the very beginning of the play --- on cycle $0$, that is ---  the imaginary adversary of the latter chose the same constants $\vec{c}$ for the free variables of $H$ as ($\mathcal K$'s real) Environment did. A simulation would generally require  maintaining and continuously updating configurations of $\mathcal N$. However, the challenge is that ${\mathcal K}$ cannot afford to fully represent such configurations on its work tape. For instance, if all bounds in 
${\mathcal R}\spa$ are sublinear, representing the run tape content of $\mathcal N$ would require more than ${\mathcal R}\spa$ space. Similarly, the size of the content of the buffer of $\mathcal N$ could occasionally go beyond the ${\mathcal R}\spa$ bound. 
For the above reasons, when dealing with a $j$th computation step of the simulated $\mathcal N$,  
 we let  ${\mathcal K}$, on its work tape, only keep representations of the other (and some additional, previously redundant)  components of the corresponding configuration of ${\mathcal N}$. 
 Namely, with ``current'' below referring to an arbitrary given $j$th computation step of $\mathcal N$, on its work tape ${\mathcal K}$ maintains  the following pieces of information\footnote{
Together with the never-changing representation of the transition function of $\mathcal N$, as well as the earlier computed $\mathfrak{S}|\max(\vec{c})|$. Whenever possible, we prefer 
 not to mention explicitly these or similar, asymptotically irrelevant/superseded, pieces of information or events.} --- call them together the {\bf sketch}\label{xsketch} of the $j$th configuration (computation step) of ${\mathcal N}$:\medskip
\begin{description}
\item[1st component] The current state of ${\mathcal N}$. 

\item[2nd component] The current contents of the work tapes of ${\mathcal N}$. 

\item[3rd component]  The current  locations of the work-tape heads of ${\mathcal N}$. 

\item[4th component] The current  location of the run-tape head of ${\mathcal N}$. 
 
\item[5th component] The number of moves that ${\mathcal N}$ has made so far (at steps $\leq j$) in the play. 

\item[6th component] The current number of symbols in the buffer of ${\mathcal N}$. 

\item[7th component] The (possibly empty) string $\alpha$ that has been added to  the buffer of $\mathcal N$ when it made a transition to the $j$th step from the preceding, $(j-1)$th, step; here we stipulate that, if $j=0$, i.e., if there is no preceding step, then such a  string $\alpha$ is empty. 

\item[8th component] The truncation $\alpha'$ of the move $\alpha$ currently written in the buffer. 
\end{description}

\begin{lem}\label{lsk}
%\marginpar{lsk}
For any $j$, with $\ell$ standing for the background of the $j$'th step of the simulated $\mathcal N$, 
maintaining the sketch for that step takes ${\mathcal R}\spa(\ell)$ space. 
\end{lem}

\begin{proof} It is sufficient to verify that each of the eight components of the sketch, individually, can be maintained/recorded with ${\mathcal R}\spa(\ell)$ space. Below we shall implicitly rely on Remark 2.4 of \cite{AAAI}.
\begin{description}
\item[{\rm\em  1st component}] Recording this component, of course, takes a constant and hence  ${\mathcal R}\spa(\ell)$ amount of space. 

\item[{\rm\em  2nd component}] Since $\mathcal N$ runs in unconditional space ${\mathcal R}\spa$, this component can be represented with ${\mathcal R}\spa(\ell)$ space. 

\item[{\rm\em  3rd component}]  The amount of space needed for recording this component obviously does not exceed the preceding amount --- in fact, it is logarithmic in  ${\mathcal R}\spa(\ell)$.

\item[{\rm\em  4th component}] By our definition of HPMs from \cite{cl12}, the run-tape head can never go beyond the leftmost blank cell. So, how many non-blank cells may be on the imaginary run tape of $\mathcal N$? Since $\mathcal N$ plays in unconditional amplitude ${\mathcal R}\amp$, and since it plays $H$  quasilegally and hence makes at most $\mathfrak{D}$ moves, the $\pp$-labeled moves residing on $\mathcal N$'s  run tape only take ${\mathcal R}\amp(\ell)$ space. Next, as we are going to see later, all $\oo$-labeled moves residing on $\mathcal N$'s   run tape are copies (made by $\mathcal K$) of $\oo$-labeled moves residing on $\mathcal K$'s run tape, by the Clean Environment Assumption meaning that their quantity is bounded by  $\mathfrak{D}$, and also implying that those moves are quasilegal, due to which (not only their magnitudes but also) their sizes do not   exceed  $O(\ell)$.
For this reason, the $\bot$-labeled moves of $\mathcal N$'s run tape,  just like the $\pp$-labeled moves, 
 only take ${\mathcal R}\amp(\ell)$ of total space. Thus,  there are at most 
${\mathcal R}\amp(\ell)$ different possible locations of $\mathcal N$'s run-tape head. Representing any of such locations  takes $|{\mathcal R}\amp(\ell)|$ and hence --- by clause 5   of Definition 2.2 of \cite{AAAI} 
 ---  ${\mathcal R}\spa(\ell)$ space.  

\item[{\rm\em  5th component}] Since $\mathcal N$ plays $H$ quasilegally, the number of moves that ${\mathcal N}$ has made so far can never exceed $\mathfrak{D}$, so holding the 5th component in memory only takes a constant amount of space.

\item[{\rm\em  6th component}] Let $m$ be the number of symbols currently in $\mathcal N$'s buffer. Assume $m>0$, for otherwise holding it takes no space. Consider the scenario where $\mathcal N$'s adversary does not make any moves beginning from the current point. Since $\mathcal N$ is unconditionally provident, sooner or later it should make a move $\alpha$ that is an extension of the move currently in the buffer, so the number of symbols in $\alpha$ is at least $m$. But, since $\mathcal N$ plays $H$ quasilegally and runs in unconditional ${\mathcal R}\amp$ amplitude, the number of symbols in $\alpha$ cannot exceed ${\mathcal R}\amp(\ell)$. That is, $m$ does not exceed 
${\mathcal R}\amp(\ell)$. Holding such an $m$ therefore requires at most  $|{\mathcal R}\amp(\ell)|$ space, and hence --- again by clause 5 of Definition 2.2 of \cite{AAAI} ---   ${\mathcal R}\spa$ space. 

\item[{\rm\em  7th component}] Recording this component, of course,  only takes a constant amount of space. 

\item[{\rm\em  8th component}] With a moment's thought and with Lemma
  \ref{lagg} in mind, it can be seen that, since  $\alpha'$ is a
  truncation, the number of symbols in it does not exceed ${\mathcal
    R}\spa(\ell)$.\qedhere
\end{description}
\end{proof}\smallskip%%%JK

\noindent Unfortunately, the sketch of a given computation step $j$ of ${\mathcal N}$ alone is not sufficient to fully trace the subsequent steps of ${\mathcal N}$ and thus successfully conduct simulation. 
The reason is that, in order to compute (the sketch of) the $(j+1)$th step of ${\mathcal N}$, one needs to know the content of the cell scanned by the run-tape head of ${\mathcal N}$. However, sketches do not keep track of what is on $\mathcal N$'s run tape, and that information --- unless residing on the run tape of ${\mathcal K}$ itself by good luck --- is generally forgotten. 
We handle this difficulty by letting the simulation routine  recompute the missing information every time such information is needed. This is done  through recursive calls to the routine itself. 
Properly materializing this general idea requires quite some care though. Among the crucial conditions for our recursive procedure to work within the required space limits is to make sure that the depth of the recursion stack never exceeds a certain constant bound. 

To achieve the above goal, we let ${\mathcal K}$, in addition to the sketches for the simulated steps of $\mathcal N$, maintain what we call the {\bf global history}.\label{xglhis} The latter is a list of all moves made by $\mathcal N$ and its adversary  throughout the  imaginary play of $H$ ``so far''. More precisely, this is not a list of moves themselves, but rather entries with certain partial information on those moves. Namely, the entry for each  
 move $\alpha$ 
does not indicate the actual content of $\alpha$ (which could require more than ${\mathcal R}\spa$ space),  but rather only the label of $\alpha$ ($\twg$ or $\tlg$, depending on whether $\alpha$ was made by $\mathcal N$ or its adversary) and the size of $\alpha$, i.e., the number of symbols in $\alpha$. Recording this information only takes $|{\mathcal R}\amp(\ell)|$ and hence ${\mathcal R}\spa(\ell)$  space.   Further, according to the forthcoming observation (\ref{gfjs}), the number of entries in the global history never exceeds $2\mathfrak{D}$ (in fact $\mathfrak{D}$, but why bother). Since $\mathfrak{D}$ is a constant, we find that ${\mathcal K}$ only consumes an ${\mathcal R}\spa(\ell)$  amount of space for maintaining the overall global history. While a move $\alpha$ is not the same as the entry for it in the global history, in the sequel we may terminologically identify these two. 

What do we need the global history for? As noted earlier, during its work, ${\mathcal K}$ will often have to resimulate some already simulated portions of the work of ${\mathcal N}$. To make such a resimulation possible, it is necessary to have information on the times at which the adversary of ${\mathcal N}$ has made its moves in the overall scenario that we are considering and re-constructing. Recording the actual move {\em times} as they were detected during the initial simulation, however, could take us beyond our target space limits. After all, think of a situation where $\mathcal N$ waits ``very long'' before its environment makes a move. So, instead, we only keep track --- via the global history --- of the {\em order} of moves. Then we neutralize the problem of not remembering the ``actual'' times of ${\mathcal N}$'s adversary's moves by simply assuming that ${\mathcal N}$'s adversary always makes its moves instantaneously in response to ${\mathcal N}$'s moves. The point is that, if ${\mathcal N}$ wins $H$, it does so in all scenarios, including the above scenario of instantaneously responding adversary.

It is important to note that, as will be immediately seen from our description of the work of ${\mathcal K}$,  the moves recorded in the global history at any step of the work of ${\mathcal K}$ are the same as the  moves on the run tape of ${\mathcal N}$. And   the latter, in turn, are copies 
of moves on the run tape of $\mathcal K$, with the only difference that, on $\mathcal K$'s run tape, the $\pp$-labeled moves appear in truncated forms. 	The orders of moves in the global history and on the run tape of $\mathcal N$ are exactly the same.  As for the run spelled on the run tape of $\mathcal K$, even if  truncation did not really modify $\mathcal N$'s moves, it may not necessarily be the same as the run spelled on the run tape of $\mathcal N$. Instead, the former is only guaranteed to be a $\pp$-delay of the latter (see Section 3 of \cite{cl12}). However, this kind of a difference, just like having the $\pp$-labeled moves truncated,  for our purposes (for $\mathcal K$'s chances to win) is just as good as --- or ``even better than'' --- if the two runs were exactly the same. 

The work of ${\mathcal K}$ relies on the three subprocedures called \transition,  \call\ and \main. We start with \transition.

\subsection{Procedure \transition}\label{sus}
%\marginpar{sus}
 In the context of a given global history $\mathbb{H}$, this procedure  
takes the sketch $\mathbb{S}_{j}$ of a given computation step $j$ of $\mathcal N$, and returns the sketch $\mathbb{S}_{j+1}$ of the next computation step $j+1$ of $\mathcal N$. 

 Let $m$ be the 5th component of $\mathbb{S}_{j}$. The number $m$ tells us how many moves ${\mathcal N}$ had made by time $j$. 
In most cases, \transition\ will be used while re-constructing some past episode of ${\mathcal N}$'s work. It is then possible that the global history contains an $(m+1)$th  move by ${\mathcal N}$ (i.e., with label $\twg$). If so, then such a move, as well as all subsequent moves of $\mathbb{H}$,  are ``future moves'' from the perspective of the $j$th step of ${\mathcal N}$ that \transition\ is currently dealing with.
 This means that, when ``imagining'' the situation at the $j$th step of 
 ${\mathcal N}$, those moves should be discarded. So, let $\mathbb{H}'$ be the result of deleting from  $\mathbb{H}$ the $(m+1)$th $\pp$-labeled move  and all subsequent, whatever-labeled moves (if there are no such moves, then simply $\mathbb{H}'=\mathbb{H}$).
 Thus, $\mathbb{H}'$ is exactly a record of the moves that $\mathcal N$ would see --- in the same order as they appear in $\mathbb{H}'$ --- on its run tape at step $j$.  

The information contained in $\mathbb{S}_{j}$ is ``almost'' sufficient for \transition\ to calculate the sought value of $\mathbb{S}_{j+1}$. The only missing piece of information is the symbol $s$ scanned by the run-tape head of ${\mathcal N}$ on step $j$.  \transition\ thus needs, first of all, to figure out what that symbol $s$ is. To do this, \transition\ computes the sum $p$ of the sizes of all moves (including their labels) of $\mathbb{H}'$. 
Next, let $q$ (found in the 4th component of $\mathbb{S}_{j}$) be the number indicating the location of the run-tape head of ${\mathcal N}$ on step $j$. Note that, in the scenario that \transition\  is dealing with, the length of the ``active'' content of ${\mathcal N}$'s run tape is $p$, with cell $\#(p+1)$ and all subsequent cells being blank. So, \transition\ compares $q$ with $p$. If $q>p$, it concludes that $s$ is \blank. Otherwise, if $q\leq p$, $s$ should be one of the symbols of one of the  moves $\alpha$ recorded in $\mathbb{H}'$. From $\mathbb{H}$,  using some easy logarithmic-space arithmetic, \transition\ figures out the author/label $\xx$ of $\alpha$, and also finds two integers $k$ and $n$. Here $k$ is the number of moves made by $\xx$ before it made the move $\alpha$. And $n$ is the number such that the sought symbol $s$ is the $n$th symbol of $\alpha$. If $\xx=\oo$, using $k$ and $n$, \transition\ finds the sought symbol $s$ on the run tape of ${\mathcal K}$. Otherwise, if $\xx=\pp$, \transition\ calls the below-described procedure \call\ on $(k,n)$. As will be seen later, \call\  then returns the sought symbol $s$. Thus, in any case, \transition\ now knows the symbol $s$ read by the run-tape head of ${\mathcal N}$ on step $j$.

Keeping the above $s$ as well as the earlier computed value $\mathfrak{S}|\max(\vec{c})|$ in mind,\footnote{This value is (could be) needed for determining the 8th component of $\mathbb{S}_{j+1}$.} \transition\ now additionally consults $\mathbb{S}_{j}$  
 and finds 
(all 8 components of) the sought sketch $\mathbb{S}_{j+1}$   using certain rather obvious logarithmic space calculations, details of which we omit.

\subsection{Procedure \call}\label{sfs}
%\marginpar{sfs} 
 In the context of a given global history $\mathbb{H}$, this procedure takes  two numbers $k,n$, where $k$ is smaller than the number of $\pp$-labeled moves in $\mathbb{H}$,
and $n$ is a positive integer not exceeding the length of the $(k+1)$th $\pp$-labeled move there. The goal of \call\ is to return, through rerunning ${\mathcal N}$, the $n$th symbol of the $(k+1)$th $\pp$-labeled move of $\mathbb{H}$.

To achieve the above goal, \call\ creates a sketch-holding variable $\mathbb{S}$, and sets the initial value of $\mathbb{S}$ to the {\bf initial sketch}.\label{xinsk} By the latter we mean the sketch of the initial configuration of ${\mathcal N}$, i.e., the configuration where ${\mathcal N}$ is in its start state, the buffer and the work tapes are empty,\footnote{As for the run tape, what is on it is irrelevant because a sketch has no record of the run-tape content anyway.} and all  scanning heads are looking at the leftmost cells of their tapes. 

After the above initialization step, \call\ performs the following subprocedure:

\begin{enumerate}
  \item Perform \transition\ on $\mathbb{S}$. Let $\mathbb{S}'$ be the resulting sketch, and let $\sigma$ be the 7th component of $\mathbb{S}'$. Below, as always, $|\sigma|$ means the length of (number of symbols in) $\sigma$. 
  \item Let $a$ and $b$ be the 5th and 6th components  of $\mathbb{S}$, respectively. If $a= k$ and $b< n\leq b+ |\sigma|$, then return the $(n- b)$th symbol of $\sigma$. Otherwise, update the value of 
 $\mathbb{S}$  to $\mathbb{S}'$,  and go back to step 1.
\end{enumerate}
Before proceeding, the reader may want to convince himself or herself that, as promised, \call\ indeed returns the $n$th symbol of the $(k+1)$th $\pp$-labeled move of $\mathbb{H}$.\enlargethispage{\baselineskip}

\subsection{Procedure \main}\label{smh}
%\marginpar{smh}
This procedure takes a global history $\mathbb{H}$ as an argument  and, treating $\mathbb{H}$ as a variable that may undergo updates,  acts according to the following prescriptions: 
\begin{description}
  \item[Stage $1$] Create a variable 
$\mathbb{S}$ and initialize its value to the initial sketch of $\mathcal N$. Proceed to Stage $2$. 
\item[Stage $2$] Check out $\mathcal K$'s run tape to see if Environment has made a new move (this can be done, say, by counting the $\oo$-labeled moves on the run tape, and comparing their number with the number of $\oo$-labeled moves recorded in the global history). If yes,   update  $\mathbb{H}$ by adding to it a record for that move, and repeat \main.  If not, go to Stage $3$. 
\item[Stage $3$]  \ 
\begin{enumerate}[label=\({\alph*}]
  \item Perform \transition\ on $\mathbb{S}$. Let $\mathbb{T}$ be the resulting sketch.
  \item If $\mathcal N$ did not make a {\bf globally new} move\label{xgnm} on its transition from $\mathbb{S}$ to $\mathbb{T}$,\footnote{Here and later in similar contexts, we terminologically identify sketches with the corresponding steps of $\mathcal N$.} change the value of the variable $\mathbb{S}$ to $\mathbb{T}$, and  go back to Stage $2$.  
Here and later in similar contexts, by a ``globally new'' move we mean a move not recorded in the global history $\mathbb{H}$. Figuring out whether $\mathcal N$ made a globally new move is easy. Technically, $\mathcal N$ made a globally new move if and only if, firstly, it {\em did} make a move, i.e., the 1st component of $\mathbb{T}$ is a move state; and secondly, such a move is not recorded in $\mathbb{H}$, meaning that the 5th component of $\mathbb{T}$ exceeds the total number of $\pp$-labeled moves recorded in $\mathbb{H}$. 
  \item Suppose now $\mathcal N$ made a globally new move $\alpha$.  Let $\alpha'$ be the 8th component of $\mathbb{S}$. Thus, $\alpha'$ is the truncation of $\alpha$. 
Copy $\alpha'$ to the buffer (of ${\mathcal K}$) symbol by symbol, after which   go to a move state. This results in ${\mathcal K}$ making the move $\alpha'$  in the real play. Now update the global history $\mathbb{H}$ by adding to it a record for the move $\alpha$, and repeat \main.\enlargethispage{\baselineskip}%\vspace{10pt}  
\end{enumerate}  
\end{description}

\subsection{The overall strategy and an example of its run} 
We continue our description of the overall work ${\mathcal K}$, started on page \pageref{xghjy} but interrupted shortly thereafter.  As we remember, at the very beginning of the play, ${\mathcal K}$ waited till Environment specified the $\mathfrak{u}$ constants $\vec{c}=c_1,\ldots,c_\mathfrak{u}$ for all free variables of $H$.  What ${\mathcal K}$  does after that is that it  creates the variable $\mathbb{H}$, initializes its value to record the sequence $\seq{\oo c_1,\ldots,\oo c_\mathfrak{u}}$, and then switches to running \main\ forever. This completes our description of $\mathcal K$.

Here we look at an example scenario to make sure we understand the work of  ${\mathcal K}$. Let 
\[H\ =\ \ada y\Bigl(|y|\leq | x |\mli \ade z\bigl(|z|\leq | x |\mlc p(z, y)\bigr)\Bigr)\mld \ada u\Bigl(|u|\leq | x |\mli \ade v\bigl(|v|\leq | x |\mlc q(u,v)\bigr)\Bigr).\]

\noindent Note that the superaggregate bound  of this formula is the identity function  $\mathfrak{S}(w)= w$. 

At the beginning of its work,  ${\mathcal K}$ waits till Environment specifies a value for $x$. Let us say $1001$ is that value. After calculating $\mathfrak{S}|1001|$, which in the present case is $4$,  ${\mathcal K}$ creates the variable $\mathbb{H}$ and sets its value to contain a record for the (single) labmove $\oo \#1001$.  The rest of the work of ${\mathcal K}$  just consists in running \main. So, in what follows, we can use ``${\mathcal K}$'' and ``\main'' as synonyms. 
 
During its initialization Stage $1$, \main\ creates 
the variable $\mathbb{S}$ and sets its value to the initial sketch of ${\mathcal N}$. 
The result of this step reflects the start situation,  where ``nothing has yet happened'' in the mixture of the real play of
 $H$ by ${\mathcal K}$ and the simulated play of $H$ by  $\mathcal N$, except for Environment's initial move $\#1001$.  

Now \main\ starts performing, over and over,  Stages $2$ and $3$. The work in those two stages can be characterized as ``global simulation''. This is a routine that keeps updating, one step at a time, the sketch $\mathbb{S}$ (Stage $3$) to the sketch of the ``next configurations'' of $\mathcal N$ 
in the scenario where the imaginary adversary  of $\mathcal N$ has made the move $\#1001$ at the very beginning of the play; every time the simulated $\mathcal N$ is trying to read some symbol of this move, ${\mathcal K}$ finds that symbol on its own run tape and feeds it back to the simulation. Simultaneously, \main\ keeps checking (Stage $2$) the run tape of ${\mathcal K}$ to see if Environment has made a new move. This will continue until either Environment or the simulated $\mathcal N$ is detected to make a new move. In our example, let us imagine that Environment makes the move $0.\#10$, signifying choosing the constant $10$ for $y$ in $H$. What happens in this case? 

\main\ simply restarts the global simulation by resetting the sketch $\mathbb{S}$  to the initial sketch of ${\mathcal N}$.  The earlier-described ``Stage $2$ and Stage $3$ over and over'' routine will be repeated, with the only difference that the global history $\mathbb{H}$ is now showing the presence of both  $\oo\# 1001$ and $\oo 0.\#10$. This means that the simulation of $\mathcal N$ will now proceed in the scenario where, at the very beginning of the play, $\mathcal N$'s adversary had made the two moves $\# 1001$ and $ 0.\#10$.  So,  every time the simulated $\mathcal N$ tries to read one of the symbols of either move on its imaginary run tape, \main --- ${\mathcal K}$, that is --- looks that symbol up on its own run tape. By switching to this new scenario, \main, in fact, deems the previous scenario invalid, and simply forgets about it.
This routine will continue until either Environment or $\mathcal N$, 
 again, is detected to make a move.

Let us say it is now $\mathcal N$, which makes the imprudent move $0.1.\#1111111$, signifying choosing the ``oversized'' (of size $>4$) constant $1111111$ for $z$ in $\mathcal H$.    In this event,  \main\ --- ${\mathcal K}$, that is --- assembles the truncation $0.1.\#1111$ of 
$0.1.\#1111111$ in its buffer copying it from the 8th component of $\mathbb{S}$, and then makes the move $0.1.\#1111$ in the real play. After that, 
as always when a new (lab)move  is detected, the global simulation restarts. Now the global history $\mathbb{H}$ is showing records for the sequence $\seq{\oo \#1001, \oo 0.\#10, \pp 0.1.\#1111111}$ of three moves.  In the present, 3rd attempt of global simulation, just like in the 2nd attempt, $\mathcal N$ is resimulated in the scenario where, at the beginning of the play, its adversary had made the moves $\# 1001$ and $ 0.\#10$. The only difference between the present attempt of global simulation and the previous one is that, once $\mathcal N$ is detected to make the expected move $0.1.\#1111111$, nothing special happens. Namely, the global history is not updated (as $0.1.\#1111111$ is already recorded there);  the move $0.1.\#1111$ is not made in the real play (as it already has been made); and  the global simulation continues in the ordinary fashion rather than restarts. 
The present attempt of global simulation, again,  
 will be interrupted if and when either Environment or the simulated $\mathcal N$  is detected to make a globally new move, i.e., a move not recorded in the global history. 

Let us say it is again Environment, which makes the move $1.\#1$, signifying choosing the constant $1$ for $u$ in $H$. As always, a record for the new move is added to $\mathbb{H}$,  and the global simulation restarts. The resimulation of $\mathcal N$ will start in the scenario where, at the beginning of the play, its adversary had made the moves $\# 1001$ and $ 0.\#10$. We already know that, in this scenario, sooner or later, $\mathcal N$ will make its previously detected move 
$0.1.\#1111111$. Once this event is detected, ${\mathcal N}$'s simulation continues for the scenario where its adversary responded by the move $1.\#1$ {\em immediately} after $\mathcal N$ made the move $0.1.\#1111111$. 

Imagine that the final globally new move detected is one by $\mathcal N$, and such a move is $1.1.\#0$, signifying choosing the constant $0$ for $v$ in $H$.  \main\ copies this move in the truncated form --- which remains the same $1.1.\#0$ because this move is quasilegal and prudent --- in the real play. Then, as always, $\mathbb{H}$ is correspondingly updated,  and the global simulation is restarted with that updated $\mathbb{H}$.

The last attempt of global simulation (the one that never got discarded/reconsidered) corresponds to the ``ultimate'' scenario that determined  ${\mathcal K}$'s real play. Namely, in our present example, the ``ultimate'' scenario in which ${\mathcal N}$ was simulated is that, at the very beginning of the play, ${\mathcal N}$'s adversary had made the moves  $\# 1001$ and $ 0.\#10$, to which ${\mathcal N}$ later responded with $0.1.\#1111111$, to which ${\mathcal N}$'s adversary immediately responded with $1.\#1$, to which, some time later, $\mathcal N$ responded with $1.1.\#0$, and no moves were made ever after. While the imaginary run generated by $\mathcal N$ in this scenario is 
%\marginpar{amind1}
\begin{equation}\label{amind1}
\seq{\oo \# 1001,\ \oo 0.\#10, \ \pp 0.1.\#1111111,\ \oo 1.\#1,\ \pp 1.1.\#0},
\end{equation}
 the real run generated by ${\mathcal K}$ is
%\marginpar{amind2}
\begin{equation}\label{amind2}
\seq{\oo \# 1001,\  \oo 0.\#10, \ \pp 0.1.\#1111,\ \oo 1.\#1,\ \pp 1.1.\#0},
\end{equation}
with (\ref{amind2}) being nothing but the result of replacing in (\ref{amind1}) all $\pp$-labeled moves by their truncations.  Since it is our assumption that $\mathcal N$ wins $H$, (\ref{amind1}) is a $\pp$-won run of $H$. But then so is (\ref{amind2}) because, as noted earlier, truncating a given player's moves can (increase but) never decrease that player's chances to win.

Why do we need to restart the global simulation every time a globally new move is detected? The reason is that otherwise we generally would not be able to rely on calls of \call\ for obtaining required symbols.  Going back to our example, imagine we did not restart the global simulation (\main) after the moves $\# 1001$ and $0.\#10$ were made by Environment. Perhaps (but not necessarily), as before, $\mathcal N$ would still make its move $0.1.\#1111111$ sometime after  $0.\#10$. Fine so far. But the trouble starts when, after that event, ${\mathcal N}$ tries to read some symbol of $0.1.\#1111111$ from its imaginary run tape. A way to provide such a symbol is to invoke \call, which will resimulate $\mathcal N$ to find that symbol. However, in order to properly resimulate $\mathcal N$ up to the moment when it made the move $0.1.\#1111111$ (or, at least, put the sought symbol of the latter into its buffer), we need to know when (on which computation steps of $\mathcal N$), exactly, the labmoves $\oo \# 1001$ and $\oo 0.\# 10$ emerged on $\mathcal N$'s run tape. Unfortunately,  we do not remember this piece of information, because, as noted earlier, remembering the exact times (as opposed to merely remembering the order) of moves may require more space than we possess. So, instead, we assume that the moves $\# 1001$ and $0.\# 10$ were made right at the beginning of $\mathcal N$'s play. This assumption, however, disagrees with the scenario of the original simulation, where $\# 1001$ was perhaps only made at step $888$, and $0.\# 10$ perhaps at step $77777$. Therefore, there is no guarantee that $\mathcal N$ will still generate the same move $0.1.\#1111111$ in response to those two moves. Restarting the global simulation --- as we did --- right after $\# 1001$ was made, and then restarting it again after 
$0.\# 10$  was detected,  neutralizes this problem. If $\mathcal N$ made its move $0.1.\#1111111$ after $0.\# 10$ in this new scenario (the scenario where its imaginary adversary always acted instantaneously), then every later resimulation, no matter how many times \main\ is restarted, will again take us to the same move $0.1.\#1111111$ made after $0.\# 10$, because the global history, which ``guides''  resimulations, will always be showing the first three labmoves in the order $\oo \# 1001, \oo 0.\#10,  \pp 0.1.\#1111111$. To see this, note that all updates of the global history only add some moves to it, and otherwise do not affect the already recorded moves or their order.

We also want to understand one remaining issue. As we should have noticed, \call\ always calls \transition, and the latter, in turn, may again call \call. Where is a guarantee that infinitely many or ``too many'' nested calls will not occur? Let us again appeal to our present example, and imagine we (\transition, that is) are currently simulating a step of $\mathcal N$ sometime after it already has made the move $1.1.\#0$. Whenever $\mathcal N$ tries to read a symbol of $\pp 1.1.\#0$, \call\ is called to resimulate  ${\mathcal N}$ and find that symbol. While resimulating ${\mathcal N}$, however, we may find that, at some point, its run-tape head is trying to read a symbol of the earlier labmove $\pp 0.1.\#1111111$. To get that symbol, \call\ will be again called  to resimulate $\mathcal N$  and find that symbol. Can this process of mutual calls go on forever? Not really. Notice that, when \call\ is called  to find the sought symbol of $\pp 0.1.\#1111111$, \call, guided by the global history, will resimulate $\mathcal N$ only up to the moment when it made the move $0.1.\#1111111$. But during that episode of $\mathcal N$'s work, the labmove $\pp 1.1.\#0$ was not yet on its run tape. So, \call\ will not have to be called further. Generally, as we are going to see in Section \ref{sdfs}, there can be at most a constant number of nested invocations of \call\ or \transition.

\subsection{\texorpdfstring{$\mathcal K$}{K} is a provident and prudent solution of \texorpdfstring{$H$}{H}} 
Consider an arbitrary play by (computation branch of)  $\mathcal K$, and fix it for the rest of this appendix.

As seen from the description of \main, the $\oo$-labeled moves recorded in  $\mathbb{H}$ are the moves made by ($\mathcal K$'s real) Environment. Since the latter is assumed to play legally, the number of $\oo$-labeled moves in $\mathbb{H}$ cannot exceed $\mathfrak{D}$. Similarly, the $\pp$-labeled moves of $\mathbb{H}$ are the moves made by $\mathcal N$ in a certain play.  Therefore, as $\mathcal N$ is a quasilegal solution of $\ada H$, the number of such moves cannot exceed $\mathfrak{D}$, either. Thus, with ``never'' below meaning ``at no stage of the work of $\mathcal K$'',  we have:
%\marginpar{gfjs}
\begin{equation}\label{gfjs}
\mbox{\em The number of labmoves in $\mathbb{H}$ never exceeds $2\mathfrak{D}$.}\footnote{In fact, with some additional analysis,  $2\mathfrak{D}$ can be lowered to $\mathfrak{D}$ here, but why bother.}
\end{equation}
Since every iteration of \main\ increases the number of labmoves in $\mathbb{H}$, an immediate corollary of (\ref{gfjs}) is that 
%\marginpar{opop}
\begin{equation}\label{opop}
\mbox{\em \main\ is iterated at most $2\mathfrak{D}$ times.}
\end{equation}

Since \main\  is restarted only finitely many times, the last iteration of it never terminates. Let $\Gamma$ be the sequence of labmoves recorded in the final value of $\mathbb{H}$ (i.e., the value of $\mathbb{H}$ throughout the last iteration of \main). This is the run generated by the simulated $\mathcal N$ in what we referred to as the ``ultimate scenario'' in the preceding subsection (scenario = computation branch). Next, let $\Delta$ be the run generated by $\mathcal K$ in the  real play that we are considering. Since $\mathcal N$ is a solution of $\ada H$, $\Gamma$ is a $\pp$-won run of $\ada H$. We want to verify that then $\Delta$ is also a $\pp$-won run of $\ada H$, meaning that $\mathcal K$, too, is a solution of $\ada H$. 

How do $\Gamma$ and $\Delta$ differ from each other? As noted at the end of Section \ref{sgst}, an analysis of the work of $\mathcal K$, details of which are left to the reader, reveals that there are only two differences. 

The first difference is that the $\pp$-labeled moves of $\Gamma$ appear in $\Delta$ in truncated forms. This is so because, whenever $\mathcal K$ makes a move  (according to the prescriptions of Stage 3(c) of \main), it copies that move 
from the 8th component of the sketch of the step of $\mathcal N$ on which the latter made a move $\alpha$; but the 8th component of a sketch always holds the truncation of the move residing in $\mathcal N$'s buffer; thus, the move $\alpha'$ made by $\mathcal K$ in the real play/run $\Delta$ is the truncation of the move $\alpha$ made by $\mathcal N$ in the imaginary play/run $\Gamma$.  

Let us use $\Omega$ to denote the result of changing in $\Gamma$  all $\pp$-labeled moves to their truncations.

The second difference between $\Gamma$ and $\Delta$ is that, even if we ignore the first difference --- that is, even if we consider $\Omega$ instead of $\Gamma$ --- the run is still not guaranteed to be exactly the same as $\Delta$; rather, we only know that the latter is a $\pp$-delay of the former. The reason for this discrepancy is that, while performing \main,  $\mathcal K$ may notice a move $\delta$ by Environment with some delay, only after it has first noticed a move $\gamma$ by $\mathcal N$ and made the truncation $\gamma'$ of $\gamma$ as a move in the real play; if this happens, $\pp\gamma$ will appear before $\oo\delta$ in $\Gamma$ but after $\oo\delta$ in $\Delta$. But the game $\ada H$ is static, as are all games studied in CoL. And, by the very definition of static games (cf. Section 3 of \cite{cl12}), $\Delta$'s being a $\pp$-delay of $\Omega$ implies that, if $\Omega$ is a $\pp$-won run of $\ada H$, then so is $\Delta$. 
This means that, in order to achieve our goal of proving that $\Delta$ is a $\pp$-won run of $\ada H$, it is sufficient to simply show that $\Omega$ is a $\pp$-won of $\ada H$. This is what the rest of this subsection is devoted to, for the exception of the last paragraph of it.  

We may and will assume that different occurrences of quantifiers in $\ada H$ bind different variables. This is a legitimate assumption, because, if it is false,  we can rename variables in $\ada H$ so as to make it true, with the new sentence, as a game, being  virtually the same  as the old sentence. 

By a {\bf unit}\label{xcwo} we shall mean a subformula $U$ of $H$ of the form $\ada r(|r|\leq \mathfrak{b}|\vec{s}|\mli E)$ (a {\bf $\ada$-unit}) or $\ade r(|r|\leq \mathfrak{b}|\vec{s}|\mlc E)$ (a {\bf $\ade$-unit}). Here $r$ is said to be the 
{\bf master variable}\label{xmstvarb} of $U$, and $|r|\leq \mathfrak{b}|\vec{s}|$ is said to be the {\bf master condition}\label{xbptrg} of $U$. ``Subunit'' and ``superunit'', applied to units, mean the same as ``subformula'' and ``superformula''. The {\bf depth} of a unit\label{xdpthu} $U$ is the number of its superunits (including $U$ itself).  
A unit $U$ is {\bf resolved}\label{xrsutt} iff $\Gamma$ contains a move signifying choosing a constant for $U$'s master variable. For instance, if $H$ is $\ade x(|x|\leq |y|\mlc x=0)\mlc \ada z(|z|\leq |y|\mli \ade t(|t|\leq |z|\successor \mlc t=z+z))$ and $\Gamma$ is $\seq{\oo \#1000,\oo 1.\#11,\pp 1.1.\#110}$, then the units $\ada z(|z|\leq |y|\mli \ade t(|t|\leq |z|\successor \mlc t=z+z))$ and 
$\ade t(|t|\leq |z|\successor \mlc t=z+z)$ are resolved while $\ade x(|x|\leq |y|\mlc x=0)$ is not.  When $w$ is a free variable of $H$ or the master variable  of some resolved unit, then the {\bf resolvent}\label{xrslf} of $w$ is the constant chosen for $w$ in (according to) $\Gamma$. 
For instance, if $H$ and $\Gamma$ are as above, $1000$ is the resolvent of $y$, $11$ is the resolvent of $z$ and $110$ is the resolvent of $t$. 
A unit $U$ is {\bf well-resolved}\label{xuwrte} iff $U$ is resolved and the result of replacing all free variables by their resolvents in $U$'s master condition is true.  A unit is   
{\bf ill-resolved} iff it is resolved but not well-resolved.\label{xillres} A {\bf critical} \label{xcrun} unit is an ill-resolved unit all of whose proper 
superunits are well-resolved.  

Let $f$ be the subaggregate bound of $H$. For $i\in\{1,2,\ldots\}$, we define \[\mathfrak{S}_i(z)\label{xgikl} \ = \ \max(f(z),f^2(z),\ldots,f^i(z)).\] Note that the superaggregate bound $\mathfrak{S}$ of $H$ is nothing but $\mathfrak{S}_{\mathfrak{H}}$, where $\mathfrak{H}$ is the total number of all units. For this reason, taking into account that the depth of no unit can exceed $\mathfrak{H}$, we have:
%\marginpar{jsd}
\begin{equation}\label{jsd}
\mbox{\em Whenever $i$ is the depth of some unit,  $\mathfrak{S}_i\preceq \mathfrak{S}$.}
\end{equation} 

\begin{lem}\label{clai2}
%\marginpar{clai2}
Consider an arbitrary resolved unit $U$. Let $i$ be its depth, and $a$ be the resolvent of its master variable.  If all superunits of $U$ (including $U$ itself) are well-resolved, then  $|a|\leq \mathfrak{S}_i|\max(\vec{c})|$. 
\end{lem}

\begin{proof} Induction on 
$i$. Assume the conditions of the lemma.  Let $w$ be the master variable of $U$, and 
let $|w|\leq \mathfrak{b}(|x_1|,\ldots,|x_k|,|z_1|,\ldots,|z_m|)$ be the master condition of $U$, with all free variables displayed, where $x_1,\ldots,x_k$ are from among the free variables of $H$, and $z_1,\ldots,z_m$ are from among the master variables of the proper superunits of $U$. Let $d_1,\ldots,d_k,e_1,\ldots,e_m$ be the resolvents of $x_1,\ldots,x_k,z_1,\ldots,z_m$, respectively. 
Below we shall use $c$, $d$ and $e$ as abbreviations of $\max(\vec{c})$, $\max(d_1,\ldots,d_k)$ and $\max(e_1,\ldots,e_m)$, respectively. Let $\mathfrak{b}'$ be the unarification of $\mathfrak{b}$. \enlargethispage{\baselineskip}

$U$'s being well-resolved means that $|a|$ does not exceed $\mathfrak{b}(|d_1|,\ldots,|d_k|,|e_1|,\ldots,|e_m|)$. Hence, by the monotonicity of $\mathfrak{b}$, we have 
 $|a|\leq \mathfrak{b}'|\max(d,e)|$.  But, of course, $\mathfrak{b}'\preceq f$ (recall that $f$ is the subaggregate bound of $H$). 
Thus, $|a|\leq f|\max(d,e)|$. This means that, in order to verify our target  $|a|\leq \mathfrak{S}_i|c|$, it is sufficient  to show that both $f|d|\leq \mathfrak{S}_i|c|$ and $f|e|\leq \mathfrak{S}_i|c|$. 

That $f|d|\leq \mathfrak{S}_i|c|$ follows from the straightforward observations that $d\leq c$ and $f\preceq \mathfrak{S}_i$.   

As for $f|e|\leq \mathfrak{S}_i|c|$, first assume $i=1$. Then $m=0$ and hence $e=0$; also, $G_i$ is $f$. Thus, we want to show that $f|0|\leq f|c|$. But this is immediate from the monotonicity of $f$. Now assume $i>1$. 
By the induction hypothesis, $|e|\leq \mathfrak{S}_{i-1}|c|$. So, $f|e|\leq f(\mathfrak{S}_{i-1}|c|)$. But, of course, 
$ f(\mathfrak{S}_{i-1}|c|)\leq \mathfrak{S}_{i}|c|$. Thus, $f|e|\leq \mathfrak{S}_i|c|$. 
\end{proof}

We are now in the position to see that $\Omega$ inherits  $\Gamma$'s being a $\pp$-won run of $\ada H$. 
Let 
\[\ada u_1\bigl(|u_1|\leq \mathfrak{p}_1|\vec{r}_1| \mli A_1\bigr),\ \ldots,\ \ada u_a\bigl(|u_a|\leq \mathfrak{p}_a|\vec{r}_a| \mli A_a\bigr)\]   be  all the critical $\ada$-units, and let $u_{1}^{\circ},\vec{r}_{1}^{\hspace{1pt}\circ},\ldots, u_{a}^{\circ},\vec{r}_{a}^{\hspace{1pt}\circ}$ be the resolvents of $u_{1},\vec{r}_{i},\ldots, u_{a},\vec{r}_{a}$, respectively. 
Similarly, let 
\[\ade v_1\bigl(|v_1|\leq \mathfrak{q}_1|\vec{s}_1| \mlc B_1\bigr),\ldots,\ade v_b\bigl(|v_b|\leq \mathfrak{q}_b|\vec{s}_b| \mlc B_b\bigr)\]   be all the critical $\ade$-units, and let $v_{1}^{\circ},\vec{s}_{1}^{\hspace{1pt}\circ},\ldots, v_{b}^{\circ},\vec{s}_{b}^{\hspace{1pt}\circ}$ be the resolvents of $v_{1},\vec{s}_{1},\ldots, v_{b},\vec{s}_{a}$, respectively. 

It is not hard to see that, following the notational conventions of Section 5.3 of \cite{AAAI} and Section 7 of \cite{cl12}, the paraformula  $\seq{\Gamma}!\ada H$ can be written as 
%\marginpar{sa1}
\begin{equation}\label{sa1}
\begin{array}{l}
X\Bigl[|u_{1}^{\circ}|\leq \mathfrak{p}_1|\vec{r}_{1}^{\hspace{1pt}\circ}| \mli A^{\circ}_{1},\ \ldots,\ |u_{a}^{\circ}|\leq \mathfrak{p}_a|\vec{r}_{a}^{\hspace{1pt}\circ}| \mli  A^{\circ}_{a},\\
\ |v_{1}^{\circ}|\leq \mathfrak{q}_1|\vec{s}_{1}^{\hspace{1pt}\circ}| \mlc B^{\circ}_{1},\ \ldots,\  |v_{b}^{\circ}|\leq \mathfrak{q}_b|\vec{s}_{b}^{\hspace{1pt}\circ}| \mlc B^{\circ}_{b}\Bigr]\end{array}\end{equation}
 for some $X$, $A_{1}^{\circ},\ldots,A_{a}^{\circ},B_{1}^{\circ},\ldots,B_{b}^{\circ}$. With some additional analysis of the situation and with (\ref{jsd}) and Lemma \ref{clai2} in mind, one can see that 
the paraformula $\seq{\Omega}!\ada H$ can then be written as 
%\marginpar{sa2}
\begin{equation}\label{sa2}
\begin{array}{l}
 X\Bigl[|u_{1}^{\circ}|\leq \mathfrak{p}_1|\vec{r}_{1}^{\hspace{1pt}\circ}| \mli A^{\bullet}_{1},\ \ldots,\ |u_{a}^{\circ}|\leq \mathfrak{p}_a|\vec{r}_{a}^{\hspace{1pt}\circ}| \mli  A^{\bullet}_{a},\\
\ |v_{1}^{\bullet}|\leq \mathfrak{q}_1|\vec{s}_{1}^{\hspace{1pt}\circ}| \mlc B^{\bullet}_{1},\ \ldots,\  |v_{b}^{\bullet}|\leq \mathfrak{q}_b|\vec{s}_{b}^{\hspace{1pt}\circ}| \mlc B^{\bullet}_{b}\Bigr]\end{array}\end{equation}
for  some  $v_{1}^{\bullet},\ldots,v_{b}^{\bullet}$, $A_{1}^{\bullet},\ldots,A_{a}^{\bullet}$, $B_{1}^{\bullet},\ldots,B_{b}^{\bullet}$ (and with all other parameters the same as in (\ref{sa1})).

By the definition of the prefixation operation (Definition 2.2 of \cite{cl12}), the fact that $\Gamma$ is a 
$\pp$-won run of $\ada H$ --- written as $\seq{\Gamma}\ada H=\pp$ --- implies (in fact, means the same as) that the empty run $\seq{}$ is a $\pp$-won run of $\seq{\Gamma}\ada H$, which, since $\seq{\Gamma}\ada H= (\ref{sa1})$, can be written as 
%\marginpar{ioloi}
\begin{equation}\label{ioloi}
\begin{array}{l}
\seq{}X\Bigl[|u_{1}^{\circ}|\leq \mathfrak{p}_1|\vec{r}_{1}^{\hspace{1pt}\circ}| \mli A^{\circ}_{1},\ \ldots,\ |u_{a}^{\circ}|\leq \mathfrak{p}_a|\vec{r}_{a}^{\hspace{1pt}\circ}| \mli A^{\circ}_{a},\\
\ |v_{1}^{\circ}|\leq \mathfrak{q}_1|\vec{s}_{1}^{\hspace{1pt}\circ}| \mlc B^{\circ}_{1},\ \ldots,\ 
 |v_{b}^{\circ}|\leq \mathfrak{q}_b|\vec{s}_{b}^{\hspace{1pt}\circ}| \mlc B^{\circ}_{b}\Bigr]  =\pp.
\end{array}\end{equation}
Consider any $i\in\{1,\ldots,b\}$. Since the unit $\ade v_i\bigl(|v_i|\leq \mathfrak{q}_i|\vec{s}_i| \mlc B_i\bigr)$ is critical and hence ill-resolved,  $|v_{i}^{\circ}|$ exceeds $\mathfrak{q}_i|\vec{s}_{i}^{\hspace{1pt}\circ}|$. Hence 
 $ \seq{} \bigl(|v_{i}^{\circ}|\leq \mathfrak{q}_i|\vec{s}_{i}^{\hspace{1pt}\circ}| \mlc B^{\circ}_{i}\bigr)=\oo$.
This clearly allows us to rewrite (\ref{ioloi}) as 
\[\seq{}X\Bigl[|u_{1}^{\circ}|\leq \mathfrak{p}_1|\vec{r}_{1}^{\hspace{1pt}\circ}| \mli A^{\circ}_{1},\ \ldots,\ |u_{a}^{\circ}|\leq \mathfrak{p}_a|\vec{r}_{a}^{\hspace{1pt}\circ}| \mli A^{\circ}_{a},\ \oo,\ \ldots,\ \oo \Bigr]  =\pp.\]
The monotonicity of the operators $(\hspace{-2pt}\mlc\hspace{-2pt},\hspace{-2pt}\mld\hspace{-2pt},\cla,\cle)$ of $X$, just as in classical logic, allows us to replace the $\oo$s by whatever games in the above equation, so the latter can be   further rewritten as 
%\marginpar{poil}
\begin{equation}\label{poil}
\begin{array}{l}
\seq{}X\Bigl[|u_{1}^{\circ}|\leq \mathfrak{p}_1|\vec{r}_{1}^{\hspace{1pt}\circ}| \mli A^{\circ}_{1},\ \ldots,\ |u_{a}^{\circ}|\leq \mathfrak{p}_a|\vec{r}_{a}^{\hspace{1pt}\circ}| \mli A^{\circ}_{a},\\
\ |v_{1}^{\bullet}|\leq \mathfrak{q}_1|\vec{s}_{1}^{\hspace{1pt}\circ}| \mlc B^{\bullet}_{1},\ \ldots,\ 
 |v_{b}^{\bullet}|\leq \mathfrak{q}_b|\vec{s}_{b}^{\hspace{1pt}\circ}| \mlc B^{\bullet}_{b}\Bigr]  =\pp.\end{array}
\end{equation} 
 Next,  for similar reasons, for every $i\in\{1,\ldots,a\}$ we have  $|u_{i}^{\circ}|> \mathfrak{p}_i|\vec{r}_{i}^{\hspace{1pt}\circ}|$ and hence 
$ \seq{} \bigl(|u_{i}^{\circ}|\leq \mathfrak{p}_i|\vec{r}_{i}^{\hspace{1pt}\circ}| \mli A^{\bullet}_{i}\bigr)=\pp$, which allows us to rewrite (\ref{poil}) as 
%\marginpar{po}
\begin{equation}\label{po}
\begin{array}{l}
\seq{}X\Bigl[|u_{1}^{\circ}|\leq \mathfrak{p}_1|\vec{r}_{1}^{\hspace{1pt}\circ}| \mli A^{\bullet}_{1},\ \ldots,\ |u_{a}^{\circ}|\leq \mathfrak{p}_a|\vec{r}_{a}^{\hspace{1pt}\circ}| \mli A^{\bullet}_{a},\\
\ |v_{1}^{\bullet}|\leq \mathfrak{q}_1|\vec{s}_{1}^{\hspace{1pt}\circ}| \mlc B^{\bullet}_{1},\ \ldots,\ 
 |v_{b}^{\bullet}|\leq \mathfrak{q}_b|\vec{s}_{b}^{\hspace{1pt}\circ}| \mlc B^{\bullet}_{b}\Bigr]  =\pp.
\end{array}\end{equation}
However, 
the $X[\ldots]$ part of (\ref{po}) is identical to (\ref{sa2}), which, in turn, is nothing but 
$\seq{\Omega}\ada H$. If so, the target  $\seq{\Omega}\ada H=\pp$ is an immediate consequence of (\ref{po}).

Thus, $\mathcal K$ is a solution of $H$, as desired. As such, it is both provident and prudent. $\mathcal K$ is provident because, as a simple examination shows, it only puts something into its buffer while acting according to clause 3(c) of the description of \main; however, at the end of the same clause, we see a prescription for $\mathcal K$ to move, and thus empty the buffer. As for prudence, it is automatically achieved because $\mathcal K$ only makes truncated moves, and such moves are always prudent.

\subsection{\texorpdfstring{$\mathcal K$}{K} plays in target tricomplexity} It remains to show that $\mathcal K$ plays $H$ in tricomplexity $\mathcal R$. Our analysis is going to be asymptotic,  
implicitly relying on Remark 2.4 of \cite{AAAI}. 

\subsubsection{Amplitude}
Since $\mathcal K$ merely mimics --- perhaps in the truncated form and perhaps with some delay --- $\mathcal N$'s moves, it is obvious that the amplitude complexity of the former does not exceed that of the latter.

In fact, $\mathcal K$'s  running in the target amplitude is also  guaranteed by the facts that 
${\mathcal R}\spa\preceq {\mathcal R}\amp$ (clause 5 of Definition 2.2 of \cite{AAAI}), $H$ is ${\mathcal R}\spa$-bounded and  ${\mathcal K}$ plays $H$ prudently.

\subsubsection{Space}\label{sdfs}
%\marginpar{sdfs}
Let $\mathbb{H}$ be a global history, and $m$ a natural number. We define the 
{\bf $\mathbb{H}$-index}\label{xhindx} of $m$ as the number of moves in $\mathbb{H}'$, where $\mathbb{H}'$ is the result of deleting from $\mathbb{H}$ the $(m+1)$th $\pp$-labeled  move  and all subsequent whatever-labeled moves; if here $\mathbb{H}$ does not contain more than $m$ $\pp$-labeled moves, then $\mathbb{H}'$ is simply $\mathbb{H}$. Next, where 
$\mathbb{S}$ is a sketch, we define the {\bf $\mathbb{H}$-index} of $\mathbb{S}$  as the $\mathbb{H}$-index of $m$, where $m$ is the value of the 5th component of $\mathbb{S}$. We extend the concept of $\mathbb{H}$-index to particular runs/iterations  of \transition\ and \call\ in the process of performing \main. Namely, \transition\ is always run on a sketch $\mathbb{S}$, and we define the {\bf $\mathbb{H}$-index} of that  run of \transition\ to be the $\mathbb{H}$-index of $\mathbb{S}$. Similarly, 
\call\ is always called on a pair $(k,n)$ for some numbers $k$ and $n$, and we define the {\bf $\mathbb{H}$-index} of such a  call/run of \transition\ as the $\mathbb{H}$-index of  $k$ ($n$ is thus irrelevant here). If $\mathbb{H}$ is fixed or clear from the context, 
we may omit ``$\mathbb{H}$-'' and simply say ``{\bf index}''.

 \begin{lem}\label{mmm}
%\marginpar{mmm}
In the process of any given iteration of \main\ and in the context of the then-current value of the global history variable $\mathbb{H}$, we have: 
\begin{enumerate}[label=\arabic*.]
\item The index of any run of \transition\  does not exceed $2\mathfrak{D}$.

\item Whenever a given  run of \transition\ calls \call, the index of the callee is strictly smaller than  that of the caller.  

\item Whenever a given run of \call\ calls \transition, the index of the callee does not exceed 
that of the caller.
\end{enumerate}
\end{lem}

\begin{proof} Clause 1 is immediate from the obvious fact that an index can never exceed the number of moves in the global history, and the latter, according to (\ref{gfjs}), is bounded by $2\mathfrak{D}$.  Clauses 2 and 3 can be verified through a rather straightforward (albeit perhaps somewhat long) analysis of the two procedures \transition\ and \call;  details of such an analysis are left to the reader. \end{proof}

We are now ready to examine the space complexity of ${\mathcal K}$. The space consumption of ${\mathcal K}$ comes from the need to simultaneously maintain the global history and various sketches. As observed earlier,  maintaining the global history consumes ${\mathcal R}\spa$ space, and, by  Lemma \ref{lsk}, each sketch also consumes ${\mathcal R}\spa$ space. At any given time, the global history is kept in memory in a single copy. So, to show that the overall space consumption is ${\mathcal R}\spa$, we need to show that, at any given time, the number of sketches simultaneously kept in the memory of ${\mathcal K}$ does not exceed a certain constant. But this is indeed so. Looking 
back at the work of \main, we see that,  at any time, its top level  maintains a single sketch. It also keeps going through this sketch and updating it through \transition,  one step at a time. Since updates are done sequentially, space used for them can be recycled, so  space consumptions for updating different sketches (this includes not only the top-level sketch of \main, but also many additional sketches that will emerge during calls to \call\ when updating each individual sketch) do not add together unless those sketches happen to be on a same branch of nested recursive calls that \transition\ and \call\ make to each other. In view of Lemma \ref{mmm}, however, the depth of recursion (the height of the recursion stack at any time) is bounded, because the index of \transition\ in the topmost level of recursion does not exceed $2\mathfrak{D}$, and every pair of successor levels of recursion strictly decreases the index of the corresponding call of \transition. 

\subsubsection{Time}

As we observed in (\ref{opop}), during the entire work of ${\mathcal K}$, \main\ is iterated at most $2\mathfrak{D}$ times. The last iteration  
runs forever, but ${\mathcal K}$ is not billed for that time because it makes no moves during that period. 
 Likewise, ${\mathcal K}  $ will not be billed for the time spent on an iteration of \main\ that was terminated at Stage 2, because a move by Environment resets $\mathcal K$'s time counter to $0$. Call the remaining sorts of iterations of \main\ --- namely, the iterations that terminate according to the scenario of case (c) of Stage 3 --- 
{\bf time-billable}.\label{xtmbl}  So, it 
is  sufficient for us to understand how much time a single time-billable iteration of \main\ takes.
 Pick  any such iteration and fix it throughout the 
context of the rest of this section, including the forthcoming Lemma \ref{mmmm}. 
We will use $\ell$ to denote the background of the last clock cycle of that iteration.  
\enlargethispage{\baselineskip}

\begin{lem}\label{mmmm} 
%\marginpar{mmmm}
The time consumed by any single run of \transition\ or 
 \call\   is  ${\mathcal R}\tim(\ell)$.
\end{lem}

\begin{proof} We verify this lemma by induction on the index $i\in\{0,\ldots,2\mathfrak{D}\}$ of the corresponding call/run of \transition\ or \call. Assume $i\geq 0$ is the index of a given run of \transition.
Looking back at our description of \transition, we 
see that this routine makes at most one call of \call. First, assume no such call is made.  Due to $\mathcal K$'s playing prudently, $\max(\ell,\mathfrak{S}(\ell))$ is 
the maximum magnitude of any move that may appear on $\mathcal K$'s run tape at any given time of the iteration. We also know from Lemma \ref{lagg}  that $\mathfrak{S}\preceq{\mathcal R}\spa$.  So, ${\mathcal K}$'s run-tape 
size (by which, as usual, we mean the size of the non-blank segment of the tape) is $O(\ell)+ {\mathcal R}\spa(\ell)$ and hence, by the relevant clauses of Definition 2.2 of \cite{AAAI}, is ${\mathcal R}\tim(\ell)$. We also know 
that the sketch and the global history are both of size ${\mathcal R}\spa(\ell)$ and hence ${\mathcal R}\tim(\ell)$. Keeping these facts 
in mind, 
with some analysis it is obvious that, in this case, \transition\  spends ${\mathcal R}\tim(\ell)$ time.   Now assume \transition\ {\em does} call \call. 
By clause 2 of Lemma \ref{mmm}, the index $j$ of such a call is less than $i$. Hence, by the induction hypothesis, the time taken by the latter is ${\mathcal R}\tim(\ell)$. In addition to this, \transition\ only spends the same amount ${\mathcal R}\tim(\ell)$
of time as in the preceding case to complete its work. Thus, in either case,  the time consumption of \transition\ is  
${\mathcal R}\tim(\ell)$.   

 Now consider a run of \call, and let $i\geq 0$ be its index. By clause 3 of Lemma \ref{mmm}, the index of any call of \transition\ that the given run of \call\  makes is at most $i$. By the induction hypothesis,  each such call of \transition\ consumes at most ${\mathcal R}\tim(\ell)$ time.  Processing any such  call (doing additional work related to it), in turn, obviously takes at most ${\mathcal R}\tim(\ell)$ time. So, each call of \transition\ costs our run of \call\ at most 
${\mathcal R}\tim(\ell)$ time. How many such calls of \transition\ will \call\ make? Since $\mathcal N$ runs in time ${\mathcal R}\tim$,  with a little thought one can see that the number of calls of \transition\ is at most 
${\mathcal R}\tim(\ell)$. So, the overall time cost of the run of \call\ is ${\mathcal R}\tim(\ell)\times {\mathcal R}\tim(\ell)$, which, in view of the closure of ${\mathcal R}\tim$ under $\times$, remains ${\mathcal R}\tim(\ell)$. \end{proof}

We are now ready to look at the time consumption of the single time-billable iteration of \main\ fixed earlier. 

Stage 1 of \main\ obviously takes a constant amount of time, and this stage is iterated only once. So, asymptotically, it contributes nothing to the overall time consumption of the procedure.  

Stage 2 checks out the run tape, which may require moving the run-tape head of ${\mathcal K}$  from one end  of the (non-blank segment of the) run tape to the other end.  Additionally, the global history needs to be updated, but  this can be done even faster. So, this stage  
obviously takes as much time as the size of  $\mathcal K$'s run tape, which, as observed in the proof of Lemma \ref{mmmm}, is ${\mathcal R}\tim(\ell)$.   

Stage 3 starts with performing \transition\ (Substage 1), and this, by Lemma \ref{mmmm}, takes ${\mathcal R}\tim(\ell)$ time.  With a little thought, the time taken by Substages (b) and (c) of Stage 3 can be seen to be at most quadratic in the size of  $\mathcal K$'s run tape. We know that the latter is ${\mathcal R}\spa(\ell)$. Hence so is $\bigl({\mathcal R}\spa(\ell)\bigr)^2$,  because 
${\mathcal R}\tim$ is closed under $\times$.

To summarize, none of the $3$ stages of \main\ takes more than ${\mathcal R}\tim(\ell)$ time. Stage 1 is repeated only once, and the remaining two stages are repeated at most ${\mathcal R}\tim(\ell)$ times as can be seen with a little thought, keeping in mind that the iteration of \main\ that we are dealing with is a time-billable one. If so, due to ${\mathcal R}\tim$'s closure under $\times$, the overall time consumption is ${\mathcal R}\tim(\ell)$, which obviously implies that $\mathcal K$ plays $\ada H$ in time ${\mathcal R}\tim$,
as desired.

\makeatletter
\renewenvironment{theindex}
               {\section*{\indexname}%
                \@mkboth{\MakeUppercase\indexname}%
                        {\MakeUppercase\indexname}%
                \thispagestyle{plain}\parindent\z@
                \parskip\z@ \@plus .3\p@\relax
                \columnseprule \z@
                \columnsep 15\p@
                \let\item\@idxitem}
               {}
\makeatother
\twocolumn
\begin{theindex}
\item $\mathbb{B}_{n}^{h}$ \pageref{eqdef}
\item $\mathbb{B}_{n}^{h}\hspace{-2pt}\uparrow$, $\mathbb{B}_{n}^{h}\hspace{-2pt}\downarrow$ \pageref{x235}
\item birthtime (of an entry) \pageref{xbirth}
\item body \pageref{xbody}  
\item body (of an entry) \pageref{xbodyof}
\item born \pageref{xbrno}
\item central triple \pageref{xcentral} 
\item common entry \pageref{xtailentry} 
\item complete semiposition \pageref{xcsmp}
\item critical unit \pageref{xcrun}
\item completion (of a semiposition) \pageref{xcdfd}
\item compression (of a semiposition) \pageref{xcdcfgf}
\item consistent (bodies) \pageref{xconsis}
\item   $\mathfrak{D}$ \pageref{xetr}
\item   depth (of a unit) \pageref{xdpthu} 
\item $\mathfrak{e}$, $\mathfrak{e}_\top$, $\mathfrak{e}_\bot$  \pageref{xe}
\item $\vec{E}$ \pageref{xeu}
\item $\vec{E}_{i}$ \pageref{xnot3}
\item early moves \pageref{xem} 
\item  entry (of an aggregation) \pageref{xentry} 
\item essentially the same (aggregations,\\ bodies etc.) \pageref{xestsm}
\item $^{\even}$ (superscript) \pageref{xbodd} 
\item  extension (of a body) \pageref{xbe} 
\item $F$ \pageref{r2}
\item $F'$ \pageref{xfp}
\item \call\ \pageref{sfs}
\item $\mathfrak{g}$ \pageref{xgqa}   
\item $\mathfrak{G}$ \pageref{xagrba}
\item $\Gamma^{C}_{\infty}$, $\Gamma^{C}_{i}$ \pageref{xgamm}
\item global history \pageref{xglhis}
\item globally new  move \pageref{xgnm}
\item $\mathfrak{h}$ \pageref{xhqa}
\item $\hbar$ \pageref{xhbar}
\item ${\mathcal H}_0$ \pageref{xhnill1}
\item ${\mathcal H}_n$ \pageref{xhnill2}
\item HPM \pageref{xHPM}
\item $\mathbb{I}$ \pageref{xiii}
\item $\mathbb{I}^h$ \pageref{xih}
\item $\mathbb{I}_{!}^{h}$ \pageref{xihe}
\item $\mathbb{I}^{h}_{\bullet}$ \pageref{xdea}
\item $\mathbb{I}^{h}_{\bullet\bullet}$ \pageref{xdea7}
\item $\mathbb{I}^{h}_{\bullet\bullet\bullet}$ \pageref{xdea8}
\item ill-resolved unit \pageref{xillres}
\item incomplete semiposition \pageref{xcsmp}
\item index (of an entry) \pageref{xindex}  
\item  index ($\mathbb{H}$-index) \pageref{xhindx}
\item initial sketch \pageref{xinsk}
\item $k$ \pageref{xkqa}
\item ${\mathcal K}$ \pageref{xkk}
\item $\mathfrak{l}$ \pageref{xl}
\item $\mathfrak{L}$ \pageref{xlla}
\item $\ell^{C}_{i}$ \pageref{xellc}
\item late moves \pageref{xlm}
\item legitimate semiposition \pageref{xdswe} 
\item locking (iteration of \mainn) \pageref{xlocking} 
\item $\mathcal M$ \pageref{xmm1},\pageref{xmm2},\pageref{xmm3},\pageref{xmm4}
\item ${\mathcal M}_k$ \pageref{xmk}
\item \mainn\ \pageref{xmain}
\item \main\ \pageref{smh}
\item master condition \pageref{xbptrg}
\item master entry \pageref{xheadentry} 
\item master: organ, payload, scale \pageref{xmaster}
\item master variable (of a unit) \pageref{xmstvarb} 
\item $\max(\vec{c })$ \pageref{xmax}
\item negative (signed organ) \pageref{xneggg}  
\item numer \pageref{xnmr}
\item numeric (lab)move \pageref{xnm}
\item $^{\odd}$ (as a superscript) \pageref{xbodd} 
\item organ \pageref{xorgan}
\item payload \pageref{xpayload}
\item positive (signed organ) \pageref{xpostv}
\item provident(ly) (branch, solution, play) \pageref{xprovid} 
\item prudent move \pageref{xprm}
\item prudent(ly) play \pageref{xprpl}
\item prudent run \pageref{xprrun}
\item prudent  solution \pageref{xprsol}
\item prudentization \pageref{xpz}
\item $\mathfrak{q}$ \pageref{xqqa} 
\item  quasilegal move \pageref{xqlm}
\item quasilegal(ly) play \pageref{xqp},\pageref{xgh}
\item quasilegal run \pageref{xqr}
\item quasilegal solution \pageref{xqs}
\item quasilegitimate semiposition \pageref{xdswepp}
\item $\mathfrak{r}$ \pageref{xrqa}
\item reasonable (play, solution) \pageref{xreas},\pageref{xgh} 
\item \repeatt\ \pageref{xrepeat}  
\item repeating (iteration of \mainn) \pageref{xring}
\item resolved unit \pageref{xrsutt}
\item resolvent \pageref{xrslf}
\item \restart\ \pageref{xrestart} 
\item restarting (iteration of \mainn) \pageref{xring} 
\item restriction (of a body) \pageref{xbr} 
\item retirement move \pageref{xretmo}
\item  $\mathfrak{S}$ \pageref{xgtrft}
\item $\mathfrak{S}_i$ \pageref{xgikl}
\item saturated \pageref{xsaturated}
\item scale \pageref{xscale}
\item  semiposition \pageref{xsemiposition}
\item signed organ \pageref{xsuperorgan}
\item $\simm$ \pageref{xsim1},\pageref{xsim}
\item $\simm^{\bullet}$ \pageref{xsb}
\item $\simm^{\leftarrow}$ \pageref{xslar}
\item $\simm^{\rightarrow}$ \pageref{xsrar}
\item $\simm$-appropriate triple \pageref{xsat}
\item size (of a body) \pageref{xsize}
\item sketch \pageref{xsketch}
\item so-far-authored   semiposition \pageref{xsofarau}
\item so-far-seen   semiposition \pageref{xsofar}
\item subaggregate bound \pageref{xsubggf} 
\item superaggregate bound \pageref{xab}
\item symbolwise length \pageref{xsyl}
\item synchronizing \pageref{xsync}
\item time-billable \pageref{xtmbl}
\item transient (iteration of \mainn) \pageref{xtransient} 
\item truncation \pageref{xtrun} 
\item \transition\ \pageref{sus}
\item $U$ \pageref{xeu}
\item unconditional (amplitude, space, time) \pageref{x0amplitude}
\item unconditionally provident(ly) \pageref{xupp}
\item unconditionally prudent(ly) \pageref{xup},\pageref{xgh}
\item unit \pageref{xcwo}
\item unreasonable (play, solution) \pageref{xreas}
\item $V_{i}^{C}$ \pageref{xvublie}
\item $W_{i}^{C}$ \pageref{xdublie}
\item $W_{i}^{B}$-induced branch of $\mathcal L$  \pageref{xwiin}
\item $W$-stabilization point \pageref{xlvw}
\item well-resolved unit \pageref{xuwrte}
\item  windup \pageref{xwindup}
\item $\mathfrak{v}$ \pageref{xvaq}
\item  \ 
\item  \ 
\item  \ 
\item  \ 
\item  \ 
\item  \ 
\item  \ 
\item  \ 
\item  \ 
\item $\xx\omega\vec{\alpha}$, $\pp\omega\vec{\alpha}$, $\oo\omega\vec{\alpha}$ \pageref{xder}
\item $\overline{B}$ (where $B$ is a body) \pageref{xoverl}
\item $\Gamma^\top$,  $\Gamma^\bot$ \pageref{xcb}
\item $\Gamma^{0.}$, $\Gamma^{1.}$ \pageref{xgnol}
\item $\oplus$ \pageref{xopl}
\item $\preceq$ (as a relation between runs) \pageref{xx85}
\item $\gneg$ (as an operation on runs) \pageref{xnki}
\item $\uparrow$, \ $\downarrow$ \pageref{x235}
\end{theindex}

\end{document}